%% file: JMST.tex
\documentclass[10pt]{article}
\usepackage[letterpaper]{geometry}

\def\withcolors{0}
\def\withnotes{0}

\def\BibTeX{{\rm B\kern-.05em{\sc i\kern-.025em b}\kern-.08em
    T\kern-.1667em\lower.7ex\hbox{E}\kern-.125emX}}
    
 \input{preamble}
\input{tangocolors.sty}

\definecolor{red1}{rgb}{0.4,0,0}

\definecolor{red1}{rgb}{0.4,0,0}
\hypersetup{
colorlinks,
 citecolor=blue,
  linkcolor=red1,
  urlcolor=Blue
}

\def\BibTeX{{\rm B\kern-.05em{\sc i\kern-.025em b}\kern-.08em
    T\kern-.1667em\lower.7ex\hbox{E}\kern-.125emX}}

\newcommand\blfootnote[1]{%
  \begingroup
  \renewcommand\thefootnote{}\footnote{#1}%
  \addtocounter{footnote}{-1}%
  \endgroup
}
\allowdisplaybreaks
\begin{document}
\title{Wyner-Ziv Estimators for Distributed Mean Estimation with Side Information  and Optimization}
\author{Prathamesh Mayekar\\ National University of Singapore\\\url{pratha22@nus.edu.sg}\and Shubham Jha\\ Indian Institute of Science\\ \url{shubhamkj@iisc.ac.in} \and  Ananda Theertha Suresh\\ Google Research\\\url{theertha@google.com} \and Himanshu Tyagi\\Indian Institute of Science\\ \url{ htyagi@iisc.ac.in} }
\date{}
\maketitle 
  \blfootnote{This work was supported by a grant from Robert Bosch Center for Cyber Physical Systems, Indian Institute of Science, and a grant on Security and Privacy for Smart Cities sponsored by National Security Council, India. }  
 \newest{ \blfootnote{Parts of this paper appeared in the proceedings of International Conference on Artificial Intelligence and Statistics 2021 and the IEEE International Symposium on Information Theory 2022 ( \cite{mayekar2021wyner}  and \cite{mayekar2022wyner}, respectively).
  }}
\begin{abstract}
 Communication efficient distributed mean estimation is an important primitive that arises in many distributed learning and optimization scenarios such as federated learning. Without any probabilistic assumptions on the underlying data, we study the problem of distributed mean estimation where the server has access to side information. 
 We propose \emph{Wyner-Ziv estimators}, which are communication and computationally efficient and near-optimal when an upper bound for the distance between the side information and the data is known. 
  As a corollary, we also show that our algorithms provide efficient schemes for the classic  Wyner-Ziv problem in information theory. 
  In a different direction, when there is no knowledge assumed about the distance between side information and the data, we present an alternative
  Wyner-Ziv estimator that uses correlated sampling. This latter setting offers {\em universal recovery guarantees}, and perhaps will be of interest in practice when the number of users is large and keeping track of the distances between the data and the side information may not be possible.

 With this mean estimator at our disposal, we revisit basic problems in
 decentralized optimization and compression where our Wyner-Ziv estimator yields algorithms with almost optimal performance.
 First, we consider the problem of  communication  constrained distributed optimization and provide an algorithm which attains the optimal convergence rate by exploiting the fact that the gradient estimates are close to each other.
Specifically, the gradient compression scheme in our algorithm 
first uses
half of the parties to form side information and then uses our Wyner-Ziv estimator to compress the remaining half of the gradient estimates.

Finally, we apply our Wynzer-Ziv estimators to the classic Wyner-Ziv compression problem in information theory to get compression schemes that are computationally efficient and are almost optimal under much more relaxed assumptions than the standard probabilistic setting.

\end{abstract}

\newpage
\tableofcontents
\newpage

%\mnote{Reminder: No need to mention subgradient in this paper because we consider smooth convex functions, which are differentiable}
%\snote{Edited the setup. We only have communication constraints.}

\section{Introduction}
\subsection{Background}
Consider the problem of distributed mean estimation for $n$ vectors $\{x_i\}_{i=1}^{n}$ in $\R^d$, where $x_i$ is available to client $i$. Each client communicates to a server using a few bits to enable the server to compute the  empirical mean
\begin{align}\label{e:sample_mean}
\bar{x}=\frac{1}{n}\sum_{i=1}^n x_i.
\end{align}

This estimation problem has become a crucial primitive for distributed optimization scenarios such as federated learning, where the data is distributed across multiple clients (see \cite{bottou2010large}, \cite{kairouz2019advances}, \cite{konevcny2016federated}, \cite{alistarh2017qsgd}, \cite{ramezani2019nuqsgd}, \cite{ faghri2020adaptive}, \cite{jhunjhunwala2021adaptive}, \cite{ghosh2020distributed}, \cite{saha2021decentralized}, \cite{lin2020achieving}, \cite{reisizadeh2020fedpaq} \cite{gandikota2019vqsgd}, \cite{basu2019qsparse},  \cite{seide20141}, \cite{wang2018atomo}, \cite{Sparse_SGD}, \cite{wen2017terngrad}, \cite{wangni2018gradient}, \cite{lu2020moniqua}, \cite{vogels2019powersgd},  \cite{acharya2019distributed}).  One of the main bottlenecks in such distributed scenarios is the significant communication cost incurred due to client communication at each iteration of the distributed algorithm. This has spurred a recent line of work which seeks  to design quantizers  to express $x_i$s using  a low precision and, yet, enable the server to compute a high accuracy estimate of $\bar{x}$ (see \cite{suresh2017distributed}, \cite{konevcny2018randomized}, \cite{chen2020breaking}, \cite{huang2019optimal},  \cite{mayekar2020ratq}, \cite{safaryan2020uncertainty}, \cite{albasyoni2020optimal}, \cite{OTA22JMT},  and the references therein).

Most of the recent works on distributed mean estimation focus on the setting where the server  must estimate the sample mean based on the client vectors, and nothing else. However, in practice, the server may also have access to some side information. For example, consider the task of training a machine learning model based on remote client data as well as some publicly accessible data  \cite{augenstein2022mixed}. At each iteration, the server communicates its global model to the client, based on which the clients compute their updates (the gradient estimates based on their local data), compress them, and then send them to the server. The server may choose to compute its own update using the publicly available dataset to complement the updates from the client.
%% as it may help in forming better estimates of the clients' updates. 
%% In the instances where the updates computed by the clients are large relative to the distance between the client update and the public update, using the %% side information  may lead to a significant increase in the accuracy of the estimate. 
In a related setting, the server can use the previously received gradients as side information for the next gradients expected from the clients.  
{Alternatively, the server may `simulate' side information from some client updates.  It can then use this side information to form much more accurate estimates of other clients' updates, leading to a faster distributed training algorithm.  We discuss this application in detail in Section \ref{s:appl}.} Similarly, distributed mean estimation with side information can be used for variance reduction in other problems such as power iteration or parallel SGD ($cf.$~\cite{davies2020distributed}).

Motivated by these observations, for the distributed mean estimation problem described at the start of the section, we study the setting in which the server has access to the side information $\{y_i\}_{i=1}^{n}$ in $\R^d$, in addition to the communication from clients. Here, $y_i$ can be viewed as server's initial 
estimate (guess)
of $x_i$. We emphasize that the side information $y_i$ is available only to the sever and can, therefore, be used for estimating the mean at the server, but is not available to the clients while quantizing the updates $\{x_i\}_{i=1}^n$. 

%%%%%%
\subsection{The model} 
Consider the input $\mathbf{x} :=(x_1,\ldots, x_n)$ and the side information $\mathbf{y} :=(y_1,\ldots, y_n)$.
The clients use a communication protocol to send $r$ bits each about their observed vector
to the server. For the ease of implementation, we restrict to non-interactive
protocols. Specifically,
we allow {\em simultaneous message passing} (SMP) protocols $\pi=(\pi_1, ...,\pi_n)$ where the communication
$C_i=\pi_i(x_i, U)\in \{0,1\}^r$ of client\footnote{$[n] :=\{1, \ldots,  n\}$.} $i$, $i\in [n]$,   can only depend on its local observation $x_i$ and public randomness $U$.
%\footnote{\theerthaedit{}{SMP protocols are non-interactive.}}
Note that the clients are not aware of side information $\mathbf{y}$, which is available only to the server.
In effect, the message $C_i$ is obtained by {\em quantizing} $x_i$ using an appropriately chosen
randomized quantizer. 
Denoting the overall communication by $C^n:=(C_1, C_2, ..., C_n)$,
the server uses the transcript $(C^n, U)$ of the protocol
and the side information $\mathbf{y}$ to form the estimate 
of the sample mean\footnote{While side information $y_i$ is associated with client $i$, we do
not enforce this association in our general formulation at this point.} $\hat{\bar{x}}=\hat{\bar{x}}(C^n, U, \mathbf{y})$;
see Figure~\ref{fig:setup} for a depiction of our setting.
We call such a $\pi$ an {\em $r$-bit SMP protocol} 
with input $(\mathbf{x}, \mathbf{y})$
and output $\hat{\bar{x}}$.

\begin{figure}[ht]
\begin{center}

\begin{tikzpicture}[scale=1, transform shape,
    pre/.style={=stealth',semithick},
    post/.style={->,shorten >=1pt,>=stealth',semithick},
dimarrow/.style={->, >=latex, line width=1pt},
dimmarrow/.style={<->, >=latex, line width=1pt},
    ]
\clip (2,17.35) rectangle  (11,13.65) ;

\draw[ fill = brown!30! ] (3.5, 17) rectangle(7.5, 16);
\node[align=center] at (5.5,16.5) {$(y_1,\ldots, y_n)$} ;
\node[align=center] at (5.5,17.2) {Server} ;
%\node[align=center] at (3.4, 18.8) {{{$\min_{x \in \X}f(x)$,}}};
%\node[align=center] at (2.7, 17.8) {{{f is convex.}}};

\draw[ fill = orange!50!] (3, 15) rectangle(4, 14);
\node[align=center] at (3.5,14.5) {$x_1$};
\node[align=center] at (3.5,13.8) {Client 1} ;

\draw[ fill = orange!50!] (4.5, 15) rectangle(5.5, 14);
\node[align=center] at (5,14.5) {$x_2$};
\node[align=center] at (5,13.8) {Client 2} ;

\draw[ fill = orange!50!] (7, 15) rectangle(8, 14);
\node[align=center] at (7.5, 14.5) {$x_n$}; 
\node[align=center] at (7.5,13.8) {Client n} ;

\draw[dashed, line width=1pt, ta3aluminium] (5.7,14.5) to (6.8,14.5);

\draw[dimarrow, ta3aluminium] (3.5,15) to (8-3,15.9); 
\draw[dimarrow, ta3aluminium]  (5,15) to(9.5-3.5,15.9) ;
\draw[dimarrow, ta3aluminium] (7.5,15) to(9.5-3,15.9)  ;

\end{tikzpicture}

\end{center}
\caption{Problem setting of mean estimation with side information}
\label{fig:setup}
\end{figure}

We measure the performance of protocol $\pi$ for inputs $\mathbf{x}$ and $\mathbf{y}$ 
and output $\hat{\bar{x}}$
using mean squared error (MSE)
given by
 \[
 \MSE(\pi, \mathbf{x}, \mathbf{y}):=\E{\norm{\hat{\bar{x}}-\bar{x}}_2^2},
 \]
  
where the expectation is over the public randomness $U$ and $\bar{x}$ is given in \eqref{e:sample_mean}.
We study the MSE of protocols for $\mathbf{x}$
and $\mathbf{y}$ such that 
the Euclidean distance between $x_i$ and $y_i$ is at most $\Delta_i$, i.e.,
\begin{align}
\norm{x_i-y_i}_2 \leq \Delta_i, \quad \forall \,i \in[n].
\label{eq:delta_cond}
\end{align}
Denoting $\mathbf{\Delta} :=(\Delta_1, \ldots, \Delta_n)$, we are interested in the performance
of our protocols for the following two settings: 

{\bf 1. The {\em \known}setting}, where $\Delta_i$ is known to client $i$ and the server; 

{\bf 2. The {\em \unknown}setting}, where $\Delta_i$s are unknown to everyone.

In both these settings, we seek to find efficient $r$-bit quantizers for $x_i$ that will allow accurate sample mean estimation.
In the \known setting, the quantizers of different clients can be chosen using the knowledge of $\mathbf{\Delta}$;
in the \unknown setting, they must be fixed irrespective of $\mathbf{\Delta}$.

In another direction, we distinguish the {\em low-precision} setting of $r\leq d$ from the {\em high-precision} setting of $r>d$.
The former is perhaps of more relevance for federated learning and high-dimensional distributed optimization, while the latter
has received a lot of attention in the information theory literature on rate-distortion theory. Moreover, the distributed estimation problem is a lot more interesting in the low-precision setting. We, therefore, focus more on this regime while also providing extensions of our protocols to the high-precision regime.

As a benchmark, we recall the result for distributed mean estimation with no side-information from \cite{suresh2017distributed}.
When all $x_i$s lie in the Euclidean ball of radius $1$, \cite{suresh2017distributed} showed that  the minmax MSE in the no side-information case is 
\begin{equation}
    \label{eq:opt}
{\Theta}\left(\frac{d}{nr} \right).
\end{equation}
%%%%%%
\subsection{Our contributions}\label{s:OurC}
Drawing on ideas from distributed quantization problem 
in information theory ($cf.$~\cite{wyner1976rate}), 
specifically the Wyner-Ziv problem, we present {\em Wyner-Ziv estimators}
for distributed mean estimation.
In the \known setting, for a fixed $\mathbf{\Delta}$,
and the low-precision setting of $r\leq d$,
we propose an {\em $r$-bit SMP protocol} $\pi^*_{\tt k}$ which satisfies\footnote{We denote by $\log(\cdot)$ logarithm to the base $2$ and by $\ln(\cdot)$ logarithm to the base $e$.}
\[
\MSE(\pi^*_{\tt k},  \mathbf{x}, \mathbf{y})={O}\left(\sum_{i=1}^{n}\frac{\Delta_i^2}{n} \cdot \frac{d \log \log n}{nr} \right),
\]
for all $\mathbf{x}$ and $\mathbf{y}$ satisfying~\eqref{eq:delta_cond}. 
 Thus, in the case where all $x_i$s lie in the Euclidean ball of radius $1$, we improve upon the optimal estimator for distributed mean estimation \eqref{eq:opt} in the regime $\sum_{i=1}^{n}\frac{\Delta_i^2 \log \log n}{n}\leq 1$. Our estimator is motivated by the classic Wyner-Ziv problem, and hence, we refer to it as
 the {\em Wyner-Ziv estimator}. The details of the algorithm are given in Section~\ref{s:sRMQ}.
 
Our protocol uses the same (randomized) $r$-bit quantizer for
each client's data and simply uses the sample mean of the quantized vectors as the estimate for $\bar{x}$. 
Furthermore, the common quantizer used by the clients is efficient and has nearly linear time-complexity of $O(d \log d)$.
Our proposed quantizer first applies a random rotation (proposed in \cite{ailon2006approximate}) to the input vectors $x_i$ 
at client $i$
and the side information vector $y_i$ at the server. This ensures that the $\Delta_i$ upper bound on the $\ell_2$ distance of $x_i$ and $y_i$ is converted to roughly a $\Delta_i/\sqrt{d}$ upper bound on the $\ell_{\infty}$ distance between $x_i$ and $y_i$. This then enables us to use efficient one-dimensional  quantizers for each coordinate of the $x_i$, which can now operate with the knowledge that the server knows a $y_i$
with each coordinate within roughly ${\Delta_i}/{\sqrt{d}}$ of $x_i$'s coordinates.

Moreover, we show that this protocol $\pi^*_{\tt k}$ has optimal (worst-case) MSE up to an $O(\log \log n)$ factor. That is, we show that for any other
$r$-bit SMP protocol $\pi$ for $r\leq d$, we can find $\mathbf{x}$ and $\mathbf{y}$ satisfying~\eqref{eq:delta_cond} such that
\begin{align*}
\MSE(\pi, \mathbf{x}, \mathbf{y})={\Omega}\left(\min_{ i \in \{1, \ldots, n\}} \Delta_i^2 \cdot \frac{d }{nr} \right).
\label{e:lower_bound}
\end{align*}

In the \unknown setting, we propose a protocol $\pi^*_{\tt u}$ which adapts to the unknown distance $\Delta_i$ between $x_i$ and $y_i$ and, remarkably,  provides MSE guarantees dependent on $\mathbf{\Delta}$. Specifically, 
for the low-precision setting of $r\leq d$,
the protocol satisfies\footnote{We denote by $\ln^*(a)$ the minimum number of iterated logarithms to the base $e$ that must be applied to $a$ to make it less than $1$.}
 \[
 \MSE(\pi^*_u, \mathbf{x}, \mathbf{y})={O}\left(\sum_{i=1}^{n}\frac{\Delta_i}{n}  \cdot \frac{d \log^*d}{nr} \right),
 \] 
 for all $\mathbf{x}$ and $\mathbf{y}$ in the unit Euclidean ball
 $\B:= \{x \in \R^d: \norm{x}_2 \leq 1\}$ and  satisfying~\eqref{eq:delta_cond}.
 Thus,  we improve upon the  optimal estimator for the no side information counterpart {\eqref{eq:opt}} in the regime $\sum_{i=1}^{n}\frac{\Delta_i \ln^* d}{ n} \leq 1.$
 Once again, the quantizer employed by the protocol is efficient and has nearly linear time-complexity of $O(d \log d)$. At the heart of our proposed quantizer is the technique of correlated sampling from~\cite{Holenstein07} which enables to derive a $\mathbf{\Delta}$ dependent MSE bound. 
 
Furthermore, both our quantizers can be extended to the high-precision regime of  $r>d$. The quantizer for the \known setting directly extends by using $r/d$ bits per dimension. The MSE of the SMP protocol using this quantizer for all the clients is only a factor of $\log n + r/d$ from the lower bound derived in \cite{davies2020distributed} for the high-precision regime. The quantizer for the \unknown setting can be extended by sending the ``type'' of the communication vector, 
 following an idea proposed in~\cite{mayekar2020limits}.  The MSE of the SMP protocol using this quantizer for all the clients falls as $2^{-r/d\ln^*d}$ as opposed to $d/r$ 
that can be obtained using naive extensions of our quantizer.
 
As remarked at the outset, mean estimator is a basic primitive that can be used in
 problems related to decentralized optimization. Indeed, we apply our Wyner-Ziv estimator to a 
 basic communication-constrained optimization problem and show that it leads to much faster algorithms for communication-constrained distributed optimization. Our first algorithm {\tt WZ-SGD} significantly improves over the baseline Parallel SGD algorithm and is almost optimal for a large number of remote clients. We also propose a universal distributed optimization algorithm {\tt UWZ-SGD}, where the remote clients can operate without the knowledge of the stochastic gradient's variance. {\tt UWZ-SGD}, too, improves the performance of the baseline Parallel SGD algorithm for large enough remote clients.

Finally, in a different direction, we revisit the classic Gaussian rate-distortion problem ($cf.$~\cite{Oohama97}) in information theory. 
In this problem, the encoder observing an Gaussian vector $X$ wants to send it to a 
decoder observing a correlated Gaussian vector $Y$ using $r$ bits. 
Using the quantizer developed in the \known setting, we obtain an efficient scheme for this classic problem which requires a minuscule excess rate over the optimal asymptotic rate.
Our scheme for this classic problem is  interesting for two reasons: The first that
it gives almost optimal result while using ``covering'' for each coordinate separately and hence is computationally efficient. All the existing schemes rely on high-dimensional covering constructed using structured codes and  most of them are computationally inefficient.
The second reason is that we do not require the distribution to be exactly Gaussian and subgaussianity suffices.

%%%%%%%
\subsection{Prior work}
The \known setting described above was first considered in \cite{davies2020distributed}.   The scheme of~\cite{davies2020distributed}
 relies on 
lattice quantizers with information theoretically optimal covering radius. Explicit lattices to be used and computationally efficient
decoding is not provided. 

In contrast, we provide explicit computationally efficient protocols
for both low- and high-precision settings. Also, we
establish lower bounds showing the optimality of our quantizer upto a multiplicative factor of $\log \log n$
in the low-precision regime of $r \leq d$ .  In comparison, the scheme of \cite{davies2020distributed} is off by a factor of $\frac{d}{r}$ from this lower bound. Thus, when $r \ll d$, 
our scheme performs significantly better than 
that in \cite{davies2020distributed}.
We remark that the \unknown setting, which is perhaps more important in certain applications where estimating the distance of
side information of each client is infeasible, has not been considered before. 

In the classic information theoretic setting, related problems of quantization with side information at the decoder have been considered
in rate-distortion theory starting with the seminal work of Wyner and Ziv~\cite{wyner1976rate}. Practical codes
for settings where the observations are generated from known distributions have been constructed
using channel codes; see, for instance,~\cite{Zamir02, Pradhan03, KoradaUrbanke10, LingGaoBelfiore12, LiuL15a}. However, these codes are computationally too expensive for our setting, cannot
be directly used for our distribution-free setup, and are designed for the high-precision setting of $r>d$. We remark that the scheme
proposed in~\cite{davies2020distributed} is similar to lattice schemes in~\cite{Zamir02, LingGaoBelfiore12, LiuL15a}.

The version of the distributed mean estimation problem with no side information at the server has been extensively studied. For any protocol in this setting operating with a precision constraint of $r \leq d$  bits per client, using a strong data processing inequality from \cite{duchi2014optimality}, \cite{suresh2017distributed} shows a lower bound on MSE of  $\displaystyle{\Omega\left(\frac{d}{nr}\right),}$ when all $x_i$s lie in the Euclidean ball of radius one.  \cite{suresh2017distributed} propose a rotation based  uniform quantization scheme which matches this lower bound up to a factor of $\log \log d$ for any precision constraint $r$. This upper bound is further improved by a random rotation based adaptive quantizer in \cite{mayekar2020ratq} to a much tighter $\log \log^*d$ factor.  For a precision constraint of $r =\Theta(d)$, the variable-length quantizers proposed in \cite{suresh2017distributed}, \cite{alistarh2017qsgd}, \cite{ramezani2019nuqsgd} as well as the fixed-length quantizers in \cite{mayekar2020limits}, \cite{gandikota2019vqsgd} are order-wise optimal.

A recent work on distributed mean estimation \cite{Jhunjhunwala21}, which came after the conference version of our paper \cite{mayekar2021wyner}, proposed two different schemes for distributed mean estimation. The first scheme improves the performance of the standard Rand-k ($cf.$ \cite{Sparse_SGD}) estimator when data across the clients are correlated. The second scheme uses previous gradient updates to improve the performance of the standard scheme. Using previous gradient updates can be seen as a special case of our setup when we use a historical gradient as side information. Interestingly, the second scheme in \cite{Jhunjhunwala21} uses the idea of centering the gradient estimate around the side information \cite[ Equation 12]{Jhunjhunwala21}, which is similar to the decoding rule used in our second Wyner-Ziv estimate \eqref{e:s_RDAQ}. A follow-up  work of \cite{mayekar2021wyner}, \cite{ LiangWu2021} also proposed using correlation amongst clients to improve over standard sample mean estimators. {Another recent work \cite{suresh2022correlated}, which also came up after our conference version \cite{mayekar2021wyner},
 proposed using correlated randomness for stochastic quantization across clients to improve the performance of the standard scheme. }

\cite{liang2021wynerziv} and an application considered in  \cite{suresh2022correlated} are closest to the application of communication distributed optimization considered in Section \ref{s:appl}. \cite{liang2021wynerziv} builds on the distributed mean estimation schemes in \cite{mayekar2021wyner} and proposes an algorithm for non-convex distributed optimization. However, unlike our proposed schemes, \cite{liang2021wynerziv} suggests using historical gradients as side information, and its optimality is unclear.  

 \cite{suresh2022correlated} considers the same setting for communication-constrained distributed optimization as considered in this paper. The proposed scheme, too, is similar to  $\tt UWZ-SGD$, one of the schemes proposed in this paper. In more detail, both schemes leverage the fact that the stochastic estimates of the gradients across clients are close to each other to reduce the compression error. Moreover, they do this by using correlated randomness, and the compression can operate without knowing how close the stochastic gradients are across clients.  However, there are crucial differences between the two schemes. At a high level, our scheme is designed for the low precision setting (where per client precision is less than the dimension) and only uses a fixed length code, the scheme in \cite{suresh2022correlated} is designed for the high precision setting and uses a variable length code in this setting.

Our results for the low-precision regime  in \known setting are provided in Section~\ref{s:known}
and in the \unknown setting are provided in Section~\ref{s:unknown}. In Section \ref{s:hpr}, we extend our results to the high-precision regime. In Section \ref{s:appl}, we derive new algorithms for communication-constrained distributed optimization using our distributed mean estimation protocols. In Section  \ref{s:GWZ}, we provide an application of the quantizer developed for the known-setting to the Gaussian Wyner-Ziv problem. Finally, we close with all the proofs in Section \ref{s:proofs}. 
Before presenting these results, we review some preliminaries in the next section.

%%%
\section{Preliminaries and the structure of our protocols}\label{s:preliminaries}
While our lower bound for the \known setting holds for an arbitrary SMP protocol,
both the protocols we propose in this paper, for the \known and the \unknown settings, have a common structure.
We use $r$-bit quantizers to form estimates of $x_i$s at the server and then compute the sample mean of
the estimates of $x_i$s. To describe our protocols and facilitate our analysis, we begin by 
concretely defining the distributed quantizers needed for this problem. Further, we present a simple result relating
the performance of the resulting protocol to the parameters of the quantizer.

An $r$-bit quantizer  $Q$ for input vectors in $\X \subset \R^d$ 
and side information $\Y\subset \R^d$
consists of randomized mappings\footnote{We can use public randomness $U$
for randomizing.}
$(\Qenc, \Qdec)$ with the encoder mapping $\Qenc:\X \to\{0,1\}^{r}$  used by the client to quantize
and the decoder mapping $\Qdec: \{0,1\}^r \times \Y \to \X$ used by the server to aggregate quantized vectors. 
The overall quantizer $Q$ is given by the composition mapping $Q(x, y)=\Qdec(
(\Qenc(x), y)$. 

In our protocols, for input $\mathbf{x}$ and side information $\mathbf{y}$,
client $i$ uses the encoder $\Qenc_i$ for the $r$-bit quantizer $Q_i$ to 
send $\Qenc_i(x_i)$. The server uses $\Qenc_i(x_i)$ and $y_i$ to form the estimate 
$\hat{x}_i=Q_i(x_i,y_i)$
of $x_i$. We assume that the randomness used in quantizers $Q_i$ for different
$i$ is independent, whereby $\hat{x}_i$ are independent of each other for different
$i$. Then server finally forms the estimate of the sample mean as
\begin{align}
  \hat{\bar{x}}:=\frac 1 n\sum_{i =1}^n \hat{x}_i.   
  \label{e:estimate}
\end{align}

For any quantizer $Q$, the following two quantities
will determine its performance when used in our distributed mean estimation protocol:
\eq{
\alpha(Q; \Delta) &:= \sup_{x\in \X, y\in \Y : \norm{x-y}_2\leq \Delta}\E{\norm{Q(x, y)-x}_2^2},\\
\beta(Q; \Delta) &:= \sup_{x\in \X, y\in \Y : \norm{x-y}_2\leq \Delta}\norm{\E{Q(x, y)-x}}_2^2,
}
where the expectation is over the randomization of the quantizer. Note that 
$\alpha(Q; \Delta)$ can be interpreted as the worst-case MSE  and $\beta(Q, \Delta)$ the worst-case bias 
over $x\in \X$ and $y\in \Y$ such that $\norm{x-y}_2\leq \Delta$.

The result below will be very handy for our analysis. 
\begin{lem}\label{l:main}
For $\mathbf{x}\in \X^n$ and $\mathbf{y}\in \Y^n$ satisfying~\eqref{eq:delta_cond} and $r$-bit quantizers $Q_i$, $i \in [n]$,
using independent randomness for different $i\in[n]$,
the estimate $\hat{\bar{x}}$ in~\eqref{e:estimate} 
and the sample mean  $\bar{x}$ in~\eqref{e:sample_mean}
satisfy
\[
\E{ \|\hat{\bar{x}} - \bar{x}\|_2^2}
\leq \sum_{i=1}^n\frac{\alpha(Q_i; \Delta_i)}{n^2} +  \sum_{i=1}^n\frac{\beta(Q_i; \Delta_i)}{n}.
\]
\end{lem}

\section{Distributed mean estimation with \known}\label{s:known}
In this section, we present our Wyner-Ziv estimator for the \known setting. As described in Section~\ref{s:preliminaries}, we use the the same (randomized) 
quantizer across all the clients and form the estimate of sample mean as in~\eqref{e:estimate}. We only need to define the common quantizer 
used by all the clients, which we do in Section \ref{s:sRMQ}. In Sections \ref{s:MQ} and \ref{s:RMQ}, we provide the basic
building blocks of our final quantizer.
Further, in Section~\ref{s:LB}, we derive a lower bound for the worst-case MSE that  
establishes  the near-optimality of our protocol. 
 Throughout we restrict to the low-precision setting of $r\leq d$.
%%%%%%
\subsection{Modulo Quantizer (MQ)}\label{s:MQ}
The first subroutine used by our larger quantizer is the {\em Modulo Quantizer } (MQ).
MQ is a one dimensional distributed quantizer that can be applied to the input $x\in\R$
with side information $y \in \R$. We give an input parameter $\Delta^\prime$
to MQ where $|x-y|\leq \Delta^\prime$. In addition to $\Delta^\prime$,
 MQ also has the resolution parameter $k$ 
and the lattice parameter $\eps$
 as inputs.

For an appropriate
$\eps$ to be specified later, 
we consider the lattice $\Z_{\eps}= \{ \eps z: z \in \Z\}$. 
For a given input $x$, the encoder $\Qenc_{\tt M}$ 
finds the closest points in $\Z_{\eps}$ larger and smaller than $x$.
Then, one of these points is sampled randomly to get an unbiased estimate of $x$. The sampled point will be of the form $\tilde{z} \eps$, where $\tilde{z}$ is in $\Z$.  We note that
the chosen point $\tilde{z}$ satisfies
\begin{align}
\eps\E{\tilde{z}} &= x\,\, \text{and}
\nonumber
\\
|x-\eps \tilde{z}|&< \eps, \quad \text{almost surely}.
    \label{e:tildez_close}
\end{align}
The encoder sends $w=\tilde{z} \mod k$ to the decoder, which requires $\log k$ bits.

Upon receiving this $w$, the decoder $\Qdec$ looks at the set $\Z_{w,\eps}=\{(zk+w )\cdot \eps: z \in \Z\}$ and decodes the point closest to $y$, which we denote by $Q_{\tt M}(x, y)$.  
Note
that declaring $y$ will already give a MSE of less than $\Delta$. 
A useful property of this decoder is that
its output is always within a bounded distance from $y$; namely,
since in Step 1 of Alg.~\ref{a:D_MCQ} we look for the closest 
point to $y$ in the lattice $Z_{w,\eps}:=\{(zk+w)\cdot \eps: z\in \Z\}$,
the output $Q_{{\tt M}}(x,y)$ satisfies 
\begin{align}
|Q_{{\tt M}}(x,y)-y|\leq k\eps, \quad \text{almost surely}.
    \label{e:boundedness_MQ_decoder}
\end{align}

We summarize MQ in Alg.~\ref{a:E_MCQ} and~\ref{a:D_MCQ}.
\begin{figure}[ht]
\centering
\begin{tikzpicture}[scale=1, every node/.style={scale=1}]
\node[draw,text width= 8 cm, text height=,] {%       
\begin{varwidth}{\linewidth}
            
            \algrenewcommand\algorithmicindent{0.7em}
\begin{algorithmic}[1]
\vspace{-0.5cm}
\Require Input $ x\in \R$, Parameters 
{$k$, $\Delta^{\prime}$, and $\eps$}

  %% \State $\eps=5\Delta^{\prime}/ k$
\State Compute $z_u = \ceil{x/\eps}$,  $z_l = \floor{x/\eps}$

 \State  Generate
     $
      {\tilde{z}} =
\begin{cases}
z_u, \quad w.p. ~ x/\eps - z_l
 \\ z_l,  \quad w.p. ~  
  z_u-x/\eps
\end{cases}
$
   
 \State \textbf{Output:} $\Qenc_{\tt M}(x)=\tilde{z} \text{ mod } k$
\end{algorithmic}  
\end{varwidth}};
 \end{tikzpicture}
 
 \renewcommand{\figurename}{Algorithm}
 \caption{Encoder  $\Qenc_{\tt M}(x)$ of MQ}\label{a:E_MCQ}
 \end{figure}
 
 \begin{figure}[ht]
\centering
\begin{tikzpicture}[scale=1, every node/.style={scale=1}]
\node[draw, text width= 8 cm, text height=,] {%       
\begin{varwidth}{\linewidth}
            
            \algrenewcommand\algorithmicindent{0.7em}
 \renewcommand{\thealgorithm}{}
\begin{algorithmic}[1]
  \Require Input $w \in \{0, \ldots, k-1\}$, $y \in \R$
   \State Compute $\hat{z}= \arg\min \{ |(zk +w)\cdot \eps -y| \colon z \in \Z\} $
    \State \textbf{Output:} $\Qdec_{\tt M}(w, y)=(\hat{z}k+w)\eps$\label{step:output_coordinate}
\end{algorithmic}
\end{varwidth}};
 \end{tikzpicture}
 
 \renewcommand{\figurename}{Algorithm}
 \caption{Decoder $\Qdec_{\tt M}(w, y)$ of MQ}\label{a:D_MCQ}
 \end{figure}

%% \begin{figure}[ht]
%% \begin{cente

%% \begin{tikzpicture}[scale=1, 
%%  mynode/.style={fill,circle,inner sep=1pt,outer sep=0pt}
%% ]

%% \draw[color=ta3aluminium, thick] (-1.6,0) -- (6.6, 0);
%% \draw[color=red1, very thick] (-1.5,-0.5) -- (-1.5, 0.5);
%% \draw[color=ta3aluminium, thick] (-1.5+1,-0.5) -- (-1.5+1, 0.5);
%% \draw[color=ta3aluminium, thick] (-1.5+2,-0.5) -- (-1.5+2, 0.5);
%% \draw[color=ta3aluminium, thick] (-1.5+3,-0.5) -- (-1.5+3, 0.5);
%% \draw[color=red1, very thick] (-1.5+4,-0.5) -- (-1.5+4, 0.5);
%% \draw[color=ta3aluminium, thick] (-1.5+5,-0.5) -- (-1.5+5, 0.5);
%% \draw[color=ta3aluminium, thick] (-1.5+6,-0.5) -- (-1.5+6, 0.5);
%% \draw[color=ta3aluminium, thick] (-1.5+7,-0.5) -- (-1.5+7, 0.5);
%% \draw[color=red1, very thick] (-1.5+8,-0.5) -- (-1.5+8, 0.5);
%% \node[mynode, fill=blue!40!] at(-1.5+3.3,0){$x$};
%% \node[mynode, fill=green!40!] at(-1.5+ 4,0){$\tilde{z}\eps $};
%% \node[mynode, fill=brown!50!] at(-1.5+ 2.4,0){$y$};
%% \draw[color=red1, very thick] (-1.5+8,-1.2) -- (-1.5+8, -0.7);
%% \node at(-1.5+ 7.55,-1){$\Z_{w,\eps}$};
%% \draw[color=ta3aluminium, very thick] (-1.5+7,-1.2) -- (-1.5+7, -0.7);
%% \node at(-1.5+ 6.7,-1){$\Z_{\eps}$};
%% \end{tikzpicture}
%% \label{f:MQ}
%% \caption{A pictorial description of $MQ$ with $k=4$, $\Delta^{\prime}$, and $w =\tilde{x} \mod k.$}
%% \end{center}
%% \end{figure}

The result below provides performance guarantees for $Q_{\tt M}$. The key observation is that the output $Q_{\tt M}(x, y)$ of the quantizer equals $\tilde{z}\eps$ with $\tilde{z}$ found at the encoder, if
$\eps$ is set appropriately. 
\begin{lem}\label{t:MQ}
Consider the Modulo Quantizer $Q_{\tt M}$ described in 
Alg.~\ref{a:E_MCQ} and~\ref{a:D_MCQ} with parameter
$\eps$ set to satisfy 
\begin{align}
 k\eps\geq 2(\eps + \Delta^\prime).
 \label{e:parameter_condition}
\end{align}
%% the Modulo Quantizer $Q_{\tt M}$ described in 
%% Alg.~\ref{a:E_MCQ} and~\ref{a:D_MCQ}
%%  with parameters $k \geq 4$ and $\Delta^{\prime}$
Then, for every $x, y$ in $\R$ such that $|x-y| \leq \Delta^\prime$,
the output $Q_{\tt M}(x, y)$ of MQ satisfies
 \begin{align*}
\E{Q_{\tt M}(x, y)}&=x\,\,\text{ and }
\\
|Q_{\tt M}(x, y)-x|&\leq \eps,
\quad \text{almost surely.}
 \end{align*}
 In particular, we can set $\eps=2\Delta^\prime/(k-2)$, 
 to get 
 $|Q_{\tt M}(x, y)-x|\leq 2\Delta^\prime/(k-2)$.
%%\theerthaedit{almost surely}{always}
Furthermore, the output of $Q_{\tt M}$ can be described in $\log k$ bits.
\end{lem}

We close with a remark that the modulo operation
used in our scheme is the simplest and easily implementable version of 
classic
coset codes obtained using nested lattices used in 
distributed quantization ($cf.$~\cite{Forney88, Zamir02, liu2016polar}) and was used in~\cite{davies2020distributed} as well.

%%%%
\subsection{Rotated Modulo Quantizer (RMQ)}\label{s:RMQ} We now describe {\em{Rotated Modulo Quantizer} (RMQ)}. RMQ  and the subsequent quantizers in this section will be used to quantize input vector $x$ in $\R^d$ with side information $y$ in $\R^d$, where $\norm{x-y}_2\leq \Delta$. RMQ first preprocesses the input $x$ and side information $y$ by randomly rotating them and then simply applies MQ for each coordinate. For rotation, we multiply both $x$ and $y$ with a matrix $R$
given by
\begin{equation}\label{e:R}
R = \frac{1}{\sqrt{d}}\cdot HD,
\end{equation}
where $H$ is the $d\times d$
Walsh-Hadamard Matrix (see \cite{horadam2012hadamard})\footnote{We
  assume that $d$ is a power of $2$. If it isn't, we can pad the vector by zeros to make it a power of 2; even in the worst-case, this only doubles the required bits.}
  and $D$ is a diagonal matrix with each diagonal entry generated uniformly from $\{-1, +1\}$. Note that we use public randomness\footnote{In practice, this can be implemented by using the same seed for pseudo-random number generator at encoder and decoder.} to generate the same $D$ at both the encoder and the decoder. We formally describe 
  the quantizer in\footnote{We denote by $(e_1, ..., e_d)$   the standard basis of $\R^d$.}  Alg.~\ref{a:E_RMQ} and~\ref{a:D_RMQ}.
  
\begin{rem}\label{r:subg}
We remark that the vector $R\left(x-y\right)$ has zero mean subgaussian coordinates with a variance factor of $\Delta^2/d$. This implies that for all coordinates $i$ in $[d]$, we have \[P\left(|R\left(x-y\right)(i)| \geq \Delta^{\prime} \right) \leq 2e^{-\frac{{\Delta^{\prime}}^{2}d}{2\Delta^2}}\] (see, for instance, \cite[Theorem 2.8]{boucheron2013concentration}). This
observation allows us to use $\Delta^\prime\approx \Delta/\sqrt{d}$
for MQ applied to each coordinate.
\end{rem}

\begin{figure}[ht]
\centering
\begin{tikzpicture}[scale=1, every node/.style={scale=1}]
\node[draw,text width= 15 cm, text height=,] {%       
\begin{varwidth}{\linewidth}
\algrenewcommand\algorithmicindent{0.7em}
\begin{algorithmic}[1]
\Require Input $ x\in \R^d$,  Parameters $k$ and $\Delta^\prime$
\State Sample $R$ as in \eqref{e:R} using public randomness
  \State $x^\prime = Rx$
\State \textbf{Output:}  $\Qenc_{{\tt M},R}(x)=[\Qenc_{{\tt M}}(x^\prime(1)), \ldots ,\Qenc_{{\tt M}}(x^\prime(d)]^T$
using parameters $k$, $\eps$, and $\Delta^\prime$
for $\Qenc_{{\tt M}}$ of Alg.~\ref{a:E_MCQ}
\end{algorithmic}  
\end{varwidth}};
 \end{tikzpicture}
 \renewcommand{\figurename}{Algorithm}
 
 \caption{Encoder  $\Qenc_{{\tt M},R}(x)$ of RMQ}\label{a:E_RMQ}
 \end{figure}
 \begin{figure}[ht]
\centering
\begin{tikzpicture}[scale=1, every node/.style={scale=1}]
\node[draw, text width= 15 cm, text height=,] {%       
\begin{varwidth}{\linewidth}
            \algrenewcommand\algorithmicindent{0.7em}
 \renewcommand{\thealgorithm}{}
\begin{algorithmic}[1]
  \Require Input $w \in \{0, \ldots, k-1\}^d$, $y \in \R^d$, 
 \Statex ${~~~~~~~~}$ Parameters $k$ and $\Delta^\prime$
 \State Get $R$ from public randomness.
   \State $y^\prime = Ry$ 
    \State \textbf{Output:} $\displaystyle{\Qdec_{{\tt M},R}(w, y)=R^{-1} \sum_{i \in [d]}\Qdec_{{\tt M}}(w(i), y^\prime(i))e_i}$  
     using parameters $k$, $\eps$, and $\Delta^\prime$
for $\Qdec_{{\tt M}}$ of Alg.~\ref{a:D_MCQ}, 

\end{algorithmic}
\end{varwidth}};
 \end{tikzpicture}
 \renewcommand{\figurename}{Algorithm}
 \caption{Decoder $\Qdec_{{\tt M},R}(w, y)$ of RMQ}\label{a:D_RMQ}
 \end{figure}

 \begin{lem}\label{t:RMQ}
Fix $\Delta \geq 0$. Let $Q_{{\tt M}, R}$ be 
RMQ described in 
Alg.~\ref{a:E_RMQ} and~\ref{a:D_RMQ}. Then,  
 for\footnote{In the proof, we provide
 a general bound which holds for all $k$.} $k \geq 4$, $\delta\in(0,\Delta)$, $\Delta^{\prime}=\sqrt{6(\Delta^2/d)\ln (\Delta/\delta)}$
 and the parameter $\eps$ of MQ set to $\eps=2\Delta^\prime/(k-2)$, we have
 for $\X=\Y=\R^d$ that
 \begin{align*}
 &\alpha(Q_{{\tt M}, R}; \Delta) \leq \frac{24\, \Delta^2 }{(k-2)^2} \ln \frac{\Delta}{\delta}   + 154\, \delta^2 
 \quad \text{and}
 \\
 & \beta(Q_{{\tt M}, R}; \Delta) \leq 154\, \delta^2.
 \end{align*}
Furthermore, the output of quantizer {$Q_{{\tt M}, R}$} can be described in $d \log k$ bits.
 \end{lem}

\begin{rem}
The choice of $\Delta^{\prime}$ in the first statement of the Lemma \ref{t:RMQ} is based on Remark \ref{r:subg}.
  We note that $\delta$ is a parameter to control the bias incurred by our quantizer. 
{
By setting $\Delta^\prime =\Delta$ we can get an unbiased quantizer, but it only recovers
the performance obtained by simply using MQ for each coordinate, an algorithm
considered in~\cite{davies2020distributed} as well.
}  
\end{rem}

 %%%%%
\subsection{Subsampled RMQ: A Wyner-Ziv quantizer for $\R^d$}\label{s:sRMQ}
 Our final quantizer is a modification of RMQ of previous section 
 where we make the precision less than $r$ bits by randomly sampling a subset of coordinates.
 Specifically, note that $\Qenc_{{\tt M},R}(x)$ sends $d$ binary strings of $\log k$
 bits each. We reduce the resolution by sending only a random subset $S$ of
 these strings. This subset is sampled using shared randomness and is available to
 the decoder, too. Note that $\Qdec_{{\tt M},R}$ applies $\Qdec_{\tt M}$
 to these strings separately; now, we use $\Qdec_{\tt M}$ to decode the entries in $S$ alone.
We describe the overall quantizer in Alg.~\ref{a:E_RCS_RMQ} and~\ref{a:D_RCS_RMQ}. 

\begin{figure}[ht]
\centering
\begin{tikzpicture}[scale=1, every node/.style={scale=1}]
\node[draw,text width= 14cm , text height=,] {%       
\begin{varwidth}{\linewidth}
            
            \algrenewcommand\algorithmicindent{0.7em}
\begin{algorithmic}[1]
\vspace{-0.5cm}
\Require Input $x\in \R$, Parameters $k$, $\Delta^{\prime}$, and $\mu$
\State Sample $S\subset[d]$ u.a.r. from all subsets of $[d]$ of
cardinality $\mu d$ and sample $R$ as in~\eqref{e:R} using public randomness

  \State \textbf{Output:} $\Qenc_{\tt WZ}(x)=\{\Qenc_{{\tt M}}(Rx(i)): i \in S \}$
  using parameters $k$, $\eps$, and $\Delta^\prime$
for $\Qenc_{{\tt M}}$ of Alg.~\ref{a:E_MCQ}
\end{algorithmic}  
\end{varwidth}};
 \end{tikzpicture}
 
 \renewcommand{\figurename}{Algorithm}
 \caption{Encoder  $\Qenc_{\tt WZ}(x)$ of subsampled RMQ}\label{a:E_RCS_RMQ}
 \end{figure}
 
 \begin{figure}[ht]
\centering
\begin{tikzpicture}[scale=1, every node/.style={scale=1}]
\node[draw, text width= 14cm, text height=,] {%       
\begin{varwidth}{\linewidth}
            
            \algrenewcommand\algorithmicindent{0.7em}
 \renewcommand{\thealgorithm}{}
\begin{algorithmic}[1]
  \Require Input $w \in \{0, \ldots, k-1\}^{\mu d}$, $y \in \R$
   \State Get $S$ and $R$ from public randomness
   \State Compute $\tilde{x}=(\Qdec_{{\tt M}}( w(i), Ry(i)), i\in S)$
   using parameters $k$, $\eps$, and $\Delta^\prime$ for $\Qdec_{{\tt M}}$ of Alg.~\ref{a:D_MCQ} 
   \State $\hat{x}_R =  \frac{1}{\mu}\sum_{i\in S}\left(\tilde{x}(i)-Ry(i) \right)e_i +Ry  $
    \State \textbf{Output:} $\displaystyle{\Qdec_{\tt WZ}(w, y)=R^{-1} \hat{x}_R}$
 
\end{algorithmic}
\end{varwidth}};
 \end{tikzpicture}
 
 \renewcommand{\figurename}{Algorithm}
 \caption{Decoder $\Qdec_{\tt WZ}(w, y)$ of subsampled RMQ}\label{a:D_RCS_RMQ}
 \end{figure}
 \begin{rem}
We remark that, typically, when implementing random sampling, we set the unsampled components to $0$. 
However, to get $\Delta$ dependent bounds on  MSE, we set the unsampled coordinates
to the corresponding coordinate of side information and center our estimate appropriately to only have small bias.
\end{rem}
The result below relates the performance of 
our final quantizer $Q_{{\tt WZ}}$ to that of $Q_{{\tt M}, R}$, which
was already analysed in
the previous section.
\begin{lem}\label{t:RCS_RAQ_alpha_beta}
Fix $\Delta  > 0$.
Let $Q_{{\tt WZ}}$ and $Q_{{\tt M}, R}$ be the quantizers described in 
Alg.~\ref{a:E_RCS_RMQ} and~\ref{a:D_RCS_RMQ} and Alg.~\ref{a:E_RMQ} and~\ref{a:D_RMQ}, respectively.
Then, for $\mu d \in [d]$,  we have for $\X=\Y=\R^d$  that
\eq{ &\alpha(Q_{{\tt WZ}}; \Delta)\leq
\frac{2\alpha(Q_{{\tt M}, R}; \Delta)}{\mu} + \frac{2\Delta^2 }{\mu}\quad \text{and}
\quad \\& \beta(Q_{{\tt WZ}}; \Delta)=\beta(Q_{{\tt M}, R}; \Delta).}
Furthermore, the output of quantizer $Q_{{\tt WZ}}$ can be described in $\mu d \log k$ bits.
\end{lem}

We are now equipped to prove our first main result. Our protocol $\pi^*_{\tt k}$ uses $Q_{\tt WZ}$ for each client as described
in Section~\ref{s:preliminaries} and forms the estimate $\hat{\bar{x}}$
as in~\eqref{e:estimate}. 
We set the parameters needed for $Q_{\tt WZ}$ in Alg.~\ref{a:E_RCS_RMQ} and~\ref{a:D_RCS_RMQ} as follows: For client $i$, we 
set the parameters of MQ as 
\begin{align}
\delta =\frac{\Delta_i}{\sqrt{n}},
\quad
\log k= \ceil{\log (2+ \sqrt{12\ln n})},
\quad 
\Delta^{\prime}= \sqrt{6(\Delta_i^2/d)\ln (\Delta_i/\delta)},
\quad 
\eps=2\Delta^\prime/(k-2),
\label{e:param_sRMQ1}
\end{align}
and set the parameter $\mu$ as
\begin{align}
 \mu d = \left\lfloor\frac{r}{\log k}\right\rfloor.   
 \label{e:param_sRMQ2}
\end{align}
We characterize the resulting error performance in the next result.
\begin{thm}
For a $n\geq 2$, a fixed $\mathbf{\Delta}=(\Delta_1, ...,\Delta_n)$,
and  $d \geq r \geq 2\ceil{\log (2+ \sqrt{12\ln n})}$,
the protocol $\pi^*_k$ with parameters as set in \eqref{e:param_sRMQ1}
and \eqref{e:param_sRMQ2} is an $r$-bit protocol which 
satisfies 
\[
\MSE(\pi^*_k, \mathbf{x}, \mathbf{y})\leq
(79\, \lceil\log(2+\sqrt{12\ln n})\rceil + 26)\,
\left(\sum_{i=1}^n\frac{\Delta_i^2}{n} \cdot \frac{d }{nr} \right), 
\]
for all $\mathbf{x},\mathbf{y}$ satisfying~\eqref{eq:delta_cond}.
\end{thm}
\begin{proof}
Denoting by $Q_{i}$ the quantizer $Q_{\tt WZ}$
with parameters set for user $i$,
by Lemmas~\ref{l:main} and~\ref{t:RCS_RAQ_alpha_beta}, we get
\begin{align*}
{\E{ \|\hat{\bar{x}} - \bar{x}\|_2^2}}
&\leq \sum_{i=1}^n\frac{\alpha(Q_i; \Delta_i)}{n^2} +  \sum_{i=1}^n\frac{\beta(Q_i; \Delta_i)}{n}
\\
&\leq 
\frac 1{\mu n^2}\sum_{i=1}^n(\alpha(Q_{{\tt M}, R,i}; \Delta_i)
+\Delta_i^2) + 
\sum_{i=1}^n\frac{\beta(Q_{{\tt M}, R,i}; \Delta_i)}{n},
\end{align*}
where $Q_{{\tt M}, R,i}$ denotes RMQ with parameters set for user $i$.
Further, since $k \geq 4$ holds when $n\geq 2$ for our choice of parameters,
by using Lemma~\ref{t:RMQ} and substituting $\delta^2=\Delta_i^2/n$, we get
\begin{align*}
    \alpha(Q_{{\tt M}, R,i}; \Delta_i)
    &\leq \frac{12 \Delta_i^2 \ln n}{(k-2)^2} + \frac{ 154\Delta_i^2}{n},
    \\
    \beta(Q_{{\tt M}, R,i}; \Delta_i)&\leq \frac{ 154\Delta_i^2}{n},
\end{align*}
which with the previous bound gives
\begin{align*}
{\E{ \|\hat{\bar{x}} - \bar{x}\|_2^2}}
&\leq 
\frac 1{\mu d}\left(
\frac{12 \ln n}{(k-2)^2} + \frac{154}{n}
+1+ 154\mu
\right)
\sum_{i=1}^n\frac{d\Delta_i^2}{n^2}
\\
&\leq 
\frac {79\lceil\log(2+\sqrt{12\ln n})\rceil + 26}{r}
\sum_{i=1}^n\frac{d\Delta_i^2}{n^2},
\end{align*}
where in the final bound
we used our choice of $k$, the assumption that $n\geq 2$ (which
implies that $d\geq r\geq 6$), 
and the fact that $\lceil r/\log k\rceil\geq
r/2$ if $r\geq 2\log k$.
\end{proof}

\begin{rem} \label{r:DMQ_lp}
We note that by using MQ for each coordinate without rotating
(or even with rotation using $R$ as above)
and with $\Delta^\prime=\Delta_i$ yields MSE less than
\[
{O}\left(\sum_{i=1}^n\frac{\Delta_i^2}{n} \cdot \frac{d \log d}{nr}\right),
\]
for $r\leq d$. Thus, our approach above allows us to remove the
$\log d$ factor at the cost of a (milder for large $d$) $\log \log n$ factor.

\end{rem}

Thus, as can be seen from the lower bound presented in Theorem  \ref{t:lb_k} below, our  Wyner-Ziv estimator $\pi^*_k$ is nearly optimal. Finally, $Q_{\tt WZ}$ can be efficiently implemented as both the encoding and decoding procedures have nearly-linear time complexity\footnote{ The most expensive operation at  both the encoder and decoder of this estimator is the Hadamard matrix multiplication operation, which requires $d \log d$ real operations.} of $O(d \log d)$.

\subsection{Lower bound}\label{s:LB}
We now prove a lower bound on the MSE incurred by any SMP protocol using $r$ bits per client. 
The proof relies on the strong data processing inequality in \cite{duchi2014optimality} and is similar in structure to the lower bound for distributed mean estimation without side-information in \cite{suresh2017distributed}.
\begin{thm}\label{t:lb_k} 
{
Fix $\mathbf{\Delta}=(\Delta_1, \ldots, \Delta_n)$.
There exists
a universal constant $c<1$ 
such that for any $r$-bit SMP protocol $\pi $,
with $r \leq  c d$,
there exists input $(\mathbf{x}, \mathbf{y})\in \R^{2d}$ 
satisfying \eqref{eq:delta_cond} and such that 
\[\MSE(\pi, \mathbf{x}, \mathbf{y}) \geq  c \min_{i \in [d]}\Delta_{i}^2\cdot \frac{d}{nr}.
\]
}
%% where $\Delta_min =  \min_{i \in [d]}\Delta_{i}^2$.
\end{thm}

\section{Distributed mean estimation for \unknown}\label{s:unknown}
  Finally, we present our Wyner-Ziv estimator for the \unknown setting. 
  We first, in Section \ref{s:CSI}, describe the idea of correlated sampling from \cite{Holenstein07}, which will serve as an essential building block for all our quantizers in this section. We then build towards our final quantizer, described in \ref{s:sRDAQ}, by first describing its simpler versions in Section \ref{s:DAQ} and \ref{s:RDAQ}. 
Once again, we restrict to the low-precision setting of $r\leq d$.

\subsection{The correlated sampling idea}\label{s:CSI}
Suppose we have two numbers $x$ and $y$ lying in $[0,1]$.
A $1$-bit unbiased estimator for $x$ is the random variable $\indic{\{U \leq x\}},$ where $U$ is a uniform random variable in $[0, 1]$. The variance of such an estimator is $x-x^2$. We 
consider a variant of this estimator given by: 
\begin{align}
\hat{X}=\indic{\{U \leq x\}} - \indic{\{U \leq y\}} +y,
\label{e:correlated_primitive}
\end{align}
where, like before, $U$ is a uniform random variable in $[0, 1].$ Such an estimator still uses only $1$-bit of information related to $x$. 
It is easy to check that this estimator
unbiased estimator of $x$, namely $\E{\hat X}=x$.
The variance of this estimator is given by
\[
\mathtt{Var}(\hat X) = \E{(\hat X - x)^2} = |x-y|-(x-y)^2,
\]
which is lower than that of the former quantizer when $x$ is close to $y$.
We build-on this basic primitive to obtain a quantizer with MSE
bounded above by a $\mathbf{\Delta}$-dependent expression, without requiring  the knowledge of 
$\mathbf{\Delta}$.

\subsection{Distance Adaptive Quantizer (DAQ)} \label{s:DAQ}
DAQ and subsequent quantizers in this Section will be described for input $x$ and side information $y$ lying in $\R^d$.
The first component of our quantizer, DAQ, which uses~\eqref{e:correlated_primitive} 
and incorporates the correlated sampling idea discussed earlier.
Both the encoder and the decoder of DAQ use the same  $d$  uniform random variables $\{U(i)\}_{i=1}^{d}$ between $[-1, 1]$, which are generated using public randomness.
At the encoder, each coordinate of vector $x$ is encoded to the bit $\indic{\{U(i) \leq x(i)\}}$. At the decoder, using the bits received from the encoder, side information $y$, and the public randomness $\{U(i)\}_{i=1}^{d}$, we first compute bits $\indic{\{U(i) \leq y(i)\}}$ 
for each $i\in[d]$. Then, the estimate of $x$ is formed as  follows:
\[
Q_{\tt D}(x, y)  =  \sum_{i=1}^{d} \left( \indic{\{U(i) \leq x(i)\}} - \indic{\{U(i) \leq y(i)\}}\right)e_i  +y. 
\]
We formally describe the quantizer in Alg. \ref{a:E_DAQ} and \ref{a:D_DAQ}.

\begin{figure}[ht]

\centering
\begin{tikzpicture}[scale=1, every node/.style={scale=1}]
\node[draw,text width= 12 cm, text height=,] {%       
\begin{varwidth}{\linewidth}
            
            \algrenewcommand\algorithmicindent{0.7em}
\begin{algorithmic}[1]

\Require Input $ x\in \R^d$
\State Sample  $U(i) \sim Unif[-1, 1], \forall i \in [d]$
 \State  
     $
      \tilde{x} = \sum_{i=1}^{d}\indic{\{U(i) \leq x(i)\}} \cdot e_i 
$
   
 \State \textbf{Output:} $\Qenc_{\tt D}(x)=\tilde{x}$, where $\tilde{x}$ is
 viewed as binary vector of length $d$
\end{algorithmic}  
\end{varwidth}};
 \end{tikzpicture}
 
 \renewcommand{\figurename}{Algorithm}
 \caption{Encoder  $\Qenc_{\tt D}(x)$ of DAQ}\label{a:E_DAQ}
 \end{figure}
 
 \begin{figure}[ht]
\centering
\begin{tikzpicture}[scale=1, every node/.style={scale=1}]
\node[draw, text width= 12 cm, text height=,] {%       
\begin{varwidth}{\linewidth}
            
\algrenewcommand\algorithmicindent{0.7em}
 \renewcommand{\thealgorithm}{}
\begin{algorithmic}[1]
  \Require Input $w \in \{0,1\}^d$, $y \in \R^d$, 
 \State Get $U(i), \forall i \in [d],$ using public randomness
   \State Set $\tilde{y}= \sum_{i=1}^{d}\indic{\{U(i) \leq y(i)\}} \cdot e_i $
    \State \textbf{Output:} $\Qdec_{\tt D}(w, y)= 2(w -\tilde{y})+y$, where $w$ is viewed as a vector
    in $\R^d$
    \label{step:output_coordinate}
\end{algorithmic}
\end{varwidth}};
 \end{tikzpicture}
 
 \renewcommand{\figurename}{Algorithm}
 \caption{Decoder $\Qdec_{\tt D}(w, y)$ of DAQ}\label{a:D_DAQ}
   
 \end{figure}
%\begin{align}\label{e:enc_DAQ}
%\Qenc_{\tt D}(x) =  \sum_{i=1}^{d}\indic{\{U(i) \leq x(i)\}} \cdot e_i
%.\end{align}

\noindent The next result characterizes the performance for DAQ.
\begin{lem}\label{t:DAQ}
Let $Q_{\tt D}$ 
denote DAQ described in Algorithms \ref{a:E_DAQ} and \ref{a:D_DAQ}. Then, for $\X=\Y=\B$ and every $\Delta > 0$, we have
\[
\alpha(Q_{\tt D}; \Delta) \leq 2\Delta\sqrt{d}\,\,\text{ and }\,\,\beta(Q_{\tt D}; \Delta)= 0.
\]
Furthermore, the output of quantizer $Q_{\tt D}$ can be described in $ d $ bits.
\end{lem}

%%\begin{rem}
 %%The estimator formed by $Q_{\tt D}$ is a natural high-dimensional %%extension of the second estimator in Section~\ref{s:CSI}. 
 %%Similarly, the standard uniform quantizer using $1$-bit per coordinate %%(see, for, instance \cite{suresh2017distributed}, %%\cite{mayekar2020ratq}) 
 %%is a natural extension of the first estimator in Section \ref{s:CSI}. 
 %%An easy calculation shows that the MSE for this uniform quantizer is %%$O(\|x\|_2\sqrt{d})$ for $x$ lying in $\B$. 
 %%On the other hand, we will show that the estimator in $Q_{\tt D}$, which %%uses the same encoder 
 %%as the uniform quantizer,
 %%\newest{but utilizes side information at the decoder,}
 %% as those standard uniform quantizers, 
%% will lead to a much lesser MSE of \newest{$O(\Delta \sqrt{d})$ when %% $\|x-y\|_2\leq \Delta$}. 
%% \end{rem}

%\noindent We defer the proof of Theorem \ref{t:DAQ} to the end of this section.

%Thus using DAQ ensures that if input vector $x$ is very close to $y$ in $\ell_1$ distance, then the worst-case MSE using DAQ is much smaller than standard quantizers  which don't use side information.

\subsection{Rotated Distance Adaptive Quantizer (RDAQ)} \label{s:RDAQ} 
Next, we proceed as for the \known setting and add
a preprocessing step of rotating $x$ and $y$
using random
matrix $R$ of~\eqref{e:R}, which is sampled using shared randomness.
We remark that here random rotation is used to exploit the subgaussianity of the rotated $x$ and $y$, whereas in RMQ of previous section it was used to exploit the subgaussianity of  $x-y$.
After this rotation step, we proceed with a quantizer similar to DAQ, but
we quantize each coordinate at multiple ``scales.'' 
We describe this step in detail below.

%% 
%% We will now further improve on the performance of DAQ by first %% preprocessing both the input vector $x$ and side information $y$ by %%randomly rotating them and then finally quantizing $x$ in an adaptive %%manner. This will allow us to further reduce the bound on worst-case MSE %%by a factor of $\sqrt{d}$. 
%% 
%% \paragraph{Rotation.} Similar to the previous setting, we preprocess the %% input $x$ and $y$  by randomly rotating them using the matrix $R$ given %% in \eqref{e:R}, which is generated via public randomness. Since the 
%% vectors $x$ and $y$ lie in the unit Euclidean ball, the random rotation %% makes both Rx and Ry subgaussian with a variance factor of $1/d$. This %% implies that for all coordinates $i \in [d]$, we have 
%% \begin{align}\label{e:suba_uDelta}
%%P(|Rx(i)| > t) \leq e^{-\frac{t^2 d}{2 }}\quad \text{and} \quad %% P(|Ry(i)| > t) \leq e^{-\frac{t^2 d}{2 }}.
%% \end{align}
%% 
%% \begin{rem}
%% In RDAQ, the random rotation is used to exploit the subgaussianity of %% the rotated $x$ and $y$, whereas in RMQ it was used to exploit the 
%% subgaussianity of  $x-y$.
%%\end{rem}

\paragraph{Using multiple scales.}
 %% Also,  in this case instead of using one uniform random variable  per %% coordinate, like was done in DAQ, 
In DAQ,
 we considered each coordinate $x$ to be anywhere between
 $[-1,1]$ and used one uniform random variable
 for each coordinate.  Now, 
 we will use $h$ independent uniform random variables
 for each coordinate,
 each corresponding to a different scale $[-M_j, M_j]$, $j\in \{0,1,2,\ldots, h-1\}$.
 For convenience, we abbreviate $[h]_0:=\{0,1,2,\ldots, h-1\}$.
 
 Specifically, let $U(i, j)$ be distributed uniformly over $[-M_j, M_j]$,
 independently for different  $i\in [d]$
  and different $j\in[h]_0$.
 The values $M_j$s correspond to different 
 scales and are set, along with $h$, as follows:
For all  $j \in[h]_0$,
\begin{align}
M_{j}^2:=  \frac{6}{d} \cdot e^{*j}, \quad
  \log h :=\ceil{\log(1+ \ln^*(d/6))},
  \label{e:levels}
  \end{align}
 where $e^{*j}$ denotes the $j$th iteration of $e$ given by
 $\displaystyle{e^{*0}:=1, \quad e^{*1}:=e, \quad e^{*j}:=e^{e^{*(j-1)}}}$.
  All the $dh$ uniform random variables are generated using public randomness and are available to both the encoder and the decoder.

The intervals $[-M_j, M_j]$ are designed to minimize the MSE of our quantizer 
by tuning its ``resolution'' to the ``scale'' of the input,
and while still ensuring unbiased estimates.
This idea of using multiple intervals $[-M_j, M_j]$  for quantizing the randomly rotated vector is from \cite{mayekar2020ratq}, where
it was used for the case with no side information.

\paragraph{Multiscale DAQ.}
After rotation, we proceed as in DAQ, except that 
we use different scale $M_j$ for different coordinates. 
Ideally, for the $i$th coordinate,  
we would like to use $M_{z^*(i)}$, where $z^*(i)$ is the smallest index such that  both $Rx(i)$ and $Ry(i)$ lie in $[-M_{z^*(i)}, M_{z^*(i)}]$.
However, since $y$ is not available to the encoder,
we simply resort to 
sending the smallest value $z(i)$ which is the smallest index such that  $ Rx(i) 
\in [-M_{z(i)}, M_{z(i)}]$ and apply the encoder of DAQ 
$h$ times to compress
$x$ at all scales, $i.e.$, we send $h$ bits
 $(\indic{\{U(i, j) \leq Rx(i)\}}, j\in [h]_0)$.

%% At RDAQs encoder, each coordinate $i$ of the vector $Rx$ is encoded  $h$ %% times using different $1$ bit representations, where the $j^{th}$ bit is  %% $\indic{\{U(i, j) \leq Rx(i)\}}$. Only one of these bits would be used %% to estimate $Rx$, eventually. Also encoded per each coordinate $i$ is %% the value $z(i)$ which is the smallest index such that  $ Rx(i) 
%% \in [-M_{z(i)}, M_{z(i)}]$. Note that $z(i)$ can be represented in $\ceil{\log h}$ bits. 
Thus, the overall number of bits used by RDAQ's encoder is $d \cdot( h+\ceil{\log h})$.
At RDAQ's decoder, using $z(i)$,
%% we first rotate the side information vector $y$ using the same random %% matrix $R$. Then using the input sent by the encoder and the side %% information $Ry$
we compute the smallest index $z^*(i)$ containing both  $Rx(i)$ and $Ry(i)$. In effect, the decoder emulates the decoder for DAQ applied
to $Ry$, but for scale $M_{z^*(i)}$.
The encoding and decoding algorithm of RDAQ are described in Alg. \ref{a:E_RDAQ}  and \ref{a:D_RDAQ}, respectively. 

%% The quantized output $Q_{{\tt D}, R}$ corresponding to input vector $x$ %% and side information $y$ is given by
%% \begin{align}
%% Q_{{\tt D}, R}(x, y) &= R^{-1} \Bigg[\sum_{i \in [d]} 2M_{z^*(i)}\Big(
%% \indic{\{U(i, z^*(i))\leq Rx(i)\}}
%% \nonumber
%% \\
%% &\hspace{1.2cm}- \indic{\{U(i, z^*(i))\leq Ry(i)\}}
%% \Big)  +Ry\Bigg],
%% \label{e:RDAQ}
%% \end{align}
%%
%% $\displaystyle{Q_{{\tt D}, R}(x, y) = R^{-1} (\sum_{i \in [d]} %% 2M_{z^*(i)}\left(B^x_i -B^y_i\right)  +Ry),}$
%%
%% where\newline
%% $B^x_i=\indic{\{U(i, z^*(i))\leq Rx(i)\}}$ and $B^y_i =  \indic{\{U(i, %% z^*(i))\leq Ry(i)\}}.$

\begin{figure}[ht]
\centering
\begin{tikzpicture}[scale=1, every node/.style={scale=1}]
\node[draw, text width= 14 cm, text height=,] {%       
\begin{varwidth}{\linewidth}
            
            \algrenewcommand\algorithmicindent{0.7em}
 \renewcommand{\thealgorithm}{}
\begin{algorithmic}[1]
\Require Input $x \in \B$
\State Sample $U(i, j) \sim Unif[-M_j, M_j]$, $i \in [d], j\in[h]_0$, and sample $R$ as in\eqref{e:R} using public randomness.

\State $x_R=Rx$
\For{$i \in [d]$} 
  \Statex \hspace{1cm} $z(i)= \min \{j \in [h]_0: |x_R(i)| \leq M_j\}$
  \EndFor

\For{$j \in [h]_0$} 
\Statex \hspace{1cm} $\tilde{x}_j= \sum_{i=1}^{d} \indic{\{U(i,j)\leq x_R(i)\}}e_i$
\EndFor

   \State \textbf{Output:} 
   $\Qenc_{{\tt D}, R}(x) = \left( [\tilde{x}_0, \ldots, \tilde{x}_{h-1}], z\right)$, where we view $\tilde{x}_j$s as binary vectors
\end{algorithmic}
\end{varwidth}};
 \end{tikzpicture}
 
 \renewcommand{\figurename}{Algorithm}
 \caption{Encoder $\Qenc_{{\tt D}, R}(x)$ at for RDAQ}\label{a:E_RDAQ}
 \end{figure}

\begin{figure}[ht]
\centering
\vspace{-0.1cm}
\begin{tikzpicture}[scale=1, every node/.style={scale=1}]
\node[draw, text width= 14 cm, text height=,] {%       
\begin{varwidth}{\linewidth}
            
            \algrenewcommand\algorithmicindent{0.7em}
 \renewcommand{\thealgorithm}{}
\begin{algorithmic}[1]
\Require   Input $ (w, z) \in \{0, 1\}^{d \times h} \times [h]_0^d$
and $y\in \B$

  \State\label{step:1} Get $U(i, j)$, $i \in [d]$, $j\in[h]_0$, and $R$ using public randomness.

  \State $y_R=Ry$
  \For{$i \in [d]$} 
  \Statex \hspace{1cm} $z^{\prime}(i)= \min \{j \in \{[h]_0\}: |y_R(i)| \leq M_j\}$
  \Statex \hspace{1cm} $z^*(i)= \max\{z(i), z^{\prime}(i)\}$
  
  \EndFor
  
 \State\label{step:4} $w^{\prime}= \sum_{i =1}^d 2M_{z^*(i)}\left(w(i, z^*(i) )  - \indic{\{U(i, z^*(i))\leq y_R\}} \right)$
\State $\hat{x}_R= w^{\prime}+Ry$
   \State \textbf{Output:}
   $\displaystyle{\Qdec_{{\tt D}, R}(w, y) =R^{-1} \hat{x}_R.}$
  
\end{algorithmic}
\end{varwidth}};
 \end{tikzpicture}
 
 \renewcommand{\figurename}{Algorithm}
 \caption{Decoder $\Qdec_{{\tt D}, R}(x)$ for RDAQ }\label{a:D_RDAQ}
 \end{figure}

%We could perhaps save on some number of bits as the estimator knows that $z^*(i)$ will be greater than $z(i)$ sent by the encoder. But since $z(i)$ with high-probability would be one of  the smaller indices in $[h]_0$, the savings by doing this would be negligible. }

 %We set our parameters $m$ and $h$ as follows
%\begin{align}
%m=\frac{6}{d},  \quad \log h=\ceil{\log(1+\ln^\ast(d/6))}.
%\label{e:RDAQ_levels}
%\end{align}
Then, the quantized output $Q_{{\tt D}, R}$ corresponding to input vector $x$ and side-information $y$ is
\begin{align*}
    Q_{{\tt D}, R}(x, y) = R^{-1} &\Bigg[\sum_{i=1}^d 2M_{z^*(i)}\left( \indic{\{U(i, z^*(i))\leq Rx(i)\}}   -\indic{\{U(i, z^*(i))\leq Ry(i)\}}\right)  +Ry \Bigg].
\end{align*}
We remark that since rotated  coordinates $Rx(i)$ and $Ry(i)$ have subgaussian tails, with very high probability $M_{z^*(i)} $ will be much less than $1$, which  helps in reducing the overall MSE significantly.
%% \footnote{As an aside, the estimator formed by $Q_{{\tt D}, R}$ also %% explains the necessity behind having $h$ different $1$ bit representations %%for each coordinate of $Rx$: the encoder does not know which index will %% be picked as $z^*(i)$.} 
The performance of the algorithm is characterized below.
\begin{lem}\label{t:RDAQ}
Let $Q_{{\tt D}, R}$ be RDAQ described in Alg. \ref{a:E_RDAQ} and \ref{a:D_RDAQ}. Then,  for $\X=\Y=\B$ 
 and every $\Delta > 0$,
we have 
\[\alpha(Q_{{\tt D}, R}; \Delta) \leq  16 \sqrt{3} \Delta \quad \text{and} \quad \beta(Q_{{\tt D}, R}; \Delta) =0.\]
Furthermore, the output of quantizer $Q$ can be described in $d( h + \log h)$ bits.
 \end{lem}

\subsection{Subsampled RDAQ: A universal Wyner-Ziv quantizer for unit Euclidean ball}\label{s:sRDAQ}
Finally, we bring down the precision of RDAQ to $r$, as before
for the \known setting,
by retaining the output of RDAQ for only coordinates $i\in S$,
where $S$ 
 is generated uniformly at random from all subsets of $[d]$
of cardinality $\mu d$ using public randomness.
Specifically, we execute 
Alg.~\ref{a:E_RDAQ} and~\ref{a:D_RDAQ} with 
$S$ replacing $[d]$ and multiplying $w^\prime$ in Step 4 of
Alg.~\ref{a:D_RDAQ} by normalization factor of $d/|S|$.
The output of the resulting encoder is given by
\begin{align}\label{e:esRDAQ}
\Qenc_{{\tt WZ}, u}(x)= \{\Qenc_{{\tt D},  R}(x)(i): i \in S \},
\end{align}
where  $ \Qenc_{{\tt D},  R}(x)(i) $ represents the encoded bits 
$([\tilde{x}_0(i), \ldots, \tilde{x}_{h-1}(i)  ], z(i) )$
for the $i${th} coordinate using RDAQ, and
the output of the resulting decoder is given by
\begin{align}
 Q_{{\tt WZ}, u}(x, y) &= R^{-1} \Bigg[   \frac{1} {\mu}\sum_{i \in S}  2M_{z^*(i)}\Big(\indic{\{U(i, z^*(i))\leq Rx(i)\}}
 -\indic{\{U(i, z^*(i))\leq Ry(i)\}}\Big)  +Ry\Bigg].
 \label{e:s_RDAQ}
 \end{align}

\begin{lem}\label{t:sRDAQ}
Let $Q_{{\tt WZ}, u}$ be the quantizers described in \eqref{e:esRDAQ} and
 \eqref{e:s_RDAQ} and $Q_{{\tt D}, R}$ be RDAQ described in
 Alg.  \ref{a:E_RDAQ} and \ref{a:D_RDAQ}.
 Then, for $\mu d \in [d]$, $\X=\Y=\B$, and every $\Delta > 0$, we have 
 \[
 \alpha(Q_{{\tt WZ}, u}; \Delta) \leq  \frac{\alpha(Q_{{\tt D}, R}; \Delta)}{\mu} \quad\text{and} \quad\beta(Q_{{\tt WZ}, u}; \Delta) =0.
 \]
Furthermore, the output of quantizer $Q_{{\tt WZ}, u}$ can be described in $\mu d( h + \log h)$ bits.
 \end{lem}
 
 We are now equipped to prove our second main result. Our protocol $\pi^*_{\tt u}$ uses $Q_{{\tt WZ}, u}$ for each client as described
in Section~\ref{s:preliminaries} and forms the estimate $\hat{\bar{x}}$
as in~\eqref{e:estimate}. 
Unlike for the \known setting, we now use the
same parameters for $Q_{{\tt WZ},u}$
for all clients, given by
\begin{align}\label{e:param_sRDAQ}
&\mu d =  \left\lfloor\frac{r}{h +\log h}\right\rfloor. 
\end{align}

\begin{thm}
For {$d \geq  r \geq 2(h+\log h)$}
and $h$ given in~\eqref{e:levels},
the $r$-bit protocol $\pi^*_u$ with parameters as set in \eqref{e:param_sRDAQ}
satisfies 
\[\MSE(\pi^*_u, \mathbf{x}, \mathbf{y})\leq 
(128\sqrt{3}\, (1+\ln^*(d/6)))
\left(\sum_{i \in [n]}\frac{\Delta_i}{n} \cdot \frac{d }{nr} \right),\]
for all $\mathbf{x},\mathbf{y}$ satisfying~\eqref{eq:delta_cond}, for every $\mathbf{\Delta}=(\Delta_1, ...,\Delta_n)$.
\end{thm}
\begin{proof}
Denote by $\hat{\bar{x}}$ the output of the protocol. Then,
by Lemmas~\ref{l:main} and Lemma~\ref{t:sRDAQ}, we get
\begin{align*}
\E{ \|\hat{\bar{x}} - \bar{x}\|_2^2}
&\leq \frac 1 {n^2\mu} \sum_{i=1}^n\alpha(Q_{{\tt D},R}; \Delta_i)
 \\
 &\leq 
 \frac {16\sqrt{3}} {n^2\mu} \sum_{i=1}^n\Delta_i,
\end{align*}
where the previous inequality is by Lemma~\ref{t:RDAQ}. 
The 
proof is completed by using 
$
\mu\geq \frac{r}{2d(h +\log h)}\geq
\frac{r}{4dh}$,
which follows from
\eqref{e:param_sRDAQ} and the assumption that $r\geq 2(h +\log h)$.
\end{proof}
%% We refer to Subsampled RDAQ as our second Wyner-Ziv estimator. Unlike %5 our first Wyner-Ziv estimator, 
The Wyner-Ziv estimator $\pi^*_u$ is universal in $\mathbf{\Delta}$: it operates without the knowledge of the distance between the input and the side information and yet gets
MSE depending on $\mathbf{\Delta}$. Moreover, it can be efficiently implemented as both the encoding and the decoding procedures have nearly linear time complexity of $O(d \log d)$. 

\section{Application: Communication constrained distributed optimization} \label{s:appl}
We consider the
problem of minimizing an unknown convex function ${f\colon\X
\to \R}$ over its domain ${\X\subset \R^d}$ 
using the set of $n$ clients who have access
to independent noisy gradients of the
function.
In particular, the server runs an optimization algorithm, which is not directly given access to the function
but can get $n$ different gradient estimates of the function at various points of its choice.
This class of optimization algorithms includes various descent algorithms,
which provide close to optimal convergence rate within the class and are appealing in practice due to their distributed nature.

Owing to our setup, the gradient estimates supplied by the $n$ clients must pass through $r$-bit quantizers, 
chosen 
from a fixed set of quantizers $\Q_r$\footnote{The set of $r$-bit quantizers $\Q_r$ is used to model the communication constraints in a distributed setting. }, and the optimization algorithm $\A$ only has access to the quantized outputs.

Our objective is to select quantizers $Q_{i,t}$, $\forall i \in [n], t \in [T],$ and an optimization algorithm $\A$ to guarantee the minimum worst-case optimization error defined below. In our setting, we allow for {\em adaptive gradient processing,} whereby, the quantizer $Q_{i,t}$ selected in $t$th iteration may depend on all the previous quantized outputs. 
%$\{Q_{i, t^{\prime}}\}_{i\in [K], t^{\prime} \in [t-1]}$.
  Specifically, 
  denoting by $\op_{i,t}$ the $i$th client's quantized output at time $t$, which takes values in the output alphabet $\R^d,$
 the \emph{adaptive quantizer selection strategy} $S\eqdef (S_1, \ldots, S_T)$ over $T$ iterations 
consists of mappings $S_t\colon 
%\prod_{i \in [n], t^{\prime} \in [t-1]} 
\R^{d\times n \times (t-1)}\to \Q_r^n$ that take $\{\op_{i, t^{\prime}}\}_{i \in [n], t^{\prime} \in [t-1]}$ as input and outputs a tuple of $n$ quantizers ${\{Q_{i,t}\}_{i \in [n]}\in\Q_r^n}$. 
We write $\mathcal{S}_{\Q_r, T}$ for the collection of all such quantizer selection strategies.
%{These quantized outputs $\{Y_{i,t}\}_{i\in [K]}$ are then used by the }
%At the end of $t$th iteration, the optimization algorithm $\A$ updates the query point $x_t$ to $x_{t+1}$ using all the messages $\{Q_{i,t^\prime}(x_{t^\prime})\}_{i \in [K], t^{\prime} \in [t]}.$
The entire framework can be summarized as follows:
\begin{enumerate}
\item At iteration $t$, the first-order optimization algorithm $\A$ makes a query for point $x_t$ to  clients $\tC_1, \ldots \tC_n$.
\item Upon receiving the point $x_t\in \X$, the client $\tC_i$, $i \in [n]$, outputs $\hat{g}_i(x_t)$, an unbiased estimate of  $\nabla f(x_t).$
%where $\E{\hat{g}_i(x_t)| x_t} = \nabla f (x_t)$ , $\forall i \in [K].$ 
\item The gradient estimate $\hat{g}_i(x_t)$ is passed through a quantizer $Q_{i,t} \in \Q_r$ chosen based on strategy $S$, and the output $Y_{i,t}$ is observed by the first-order optimization algorithm $\A$. The algorithm then uses all the messages $\{\op_{i,t^\prime}(x_{t^\prime})\}_{i \in [n], t^{\prime} \in [t]}$ to further update $x_t$ to $x_{t+1}$.
\end{enumerate}
Denote by $\clientset$ the collection of $n$ clients $(\tC_1, \ldots, \tC_n)$. Let $\cA_T$ be the set of all first-order optimization algorithms that make $T$ queries to $\clientset$ and for the $t$th query $x_t$, get back the outputs $\{Y_{i,t}\}_{i \in [n]}$.
% with distribution $\{W_{i, t}(\cdot \mid \hat{g}_i(x_t))\}_{i \in [n]}$.}
We measure the performance of
an optimization protocol $\A$ and a quantizer selection strategy $S$ for a given function $f$ and clients $\tC_i$, $i \in [n],$ using the metric
$\ep(f,  \clientset, \A, S)$ defined as  {\[\displaystyle{\hspace{1.5cm}
  \ep(f, \clientset, \A, S) = \E{f\left(\bar{x}_T\right)-\min_{x\in \X} f(x)},}\]}\\
%\end{equation} % C: either we put \eqdef for all three errors, or none; this one was the only one with \eqdef
where $\bar{x}_T:=\frac{1}{T}\sum_{t\in [T]}x_t$ and the expectation is over the randomness in $\bar{x}_T$.

For a set of various function and client pairs above, denoted by $\oO$, the set of $r$-bit quantizers $\mathcal{Q}_r$ and the number of iterations $T$, we define the \emph{minimax optimization error} as
\begin{align*}
%\label{eq:def:optimization:error}
  \ep^\ast(\X , \oO, T, \Q_r) =   \inf_{\A \in \cA_T}\inf_{S \in \mathcal{S}_{\Q_r, T}} \sup_{(f, \clientset)\in\oO}\ep(f, \clientset, \A, S)\,.
\end{align*}

We now define the class of functions and state the assumptions related to the  clients accessible to the algorithm $\A$.
  
%%%%%%%%%%%%%%%%%%%%%%%%%%%%%%%%%%%%%%%%%%
\paragraph{Convex and smooth function family}
Throughout, we restrict
 ourselves to convex and $L$-smooth functions over $\X\subset\R^d$, i.e.,
functions  satisfying, $\forall \lambda \in [0, 1], \forall x, y \in \R^d$, 
 \begin{align} 
 f(\lambda x +(1-\lambda)y)&\leq \lambda f(x) +(1-\lambda)f(y), 
 \label{e:convexity}\\
\norm{\nabla f(x) - \nabla f(y)}_2 &\leq L \norm{x-y}_2,\label{e:L_smooth}
\end{align}
where $\nabla f(x)\in \R^d$ denotes the gradient of $f$ at input $x.$

\paragraph{Stochastic gradients}
We assume that the output $\hat{g}_i(x)$ by client $\tC_i, 1\leq i\leq n,$ when a point $x\in \X$ is queried satisfies the following conditions:
\begin{align}\label{e:asmp_unbiasedness}
\E{\hat{g}_i(x)\mid x} &= \nabla f(x), \quad \text{(unbiased estimates)} \\
\label{e:varaince}
\norm{\hat{g}_i(x)-\nabla f(x)}_2^2 &\leq \sigma^2,\, \text{{(maximum deviation bound)}}
\\
\label{e:bounded_est}
\norm{\hat{g}_i(x)}_2^2 &\leq B^2.\qquad  \text{(a.s. bounded estimate)} 
\end{align}
Assumption {\eqref{e:asmp_unbiasedness} is standard in stochastic optimization literature ($cf.$ \cite{nemirovsky1983problem}, \cite{nemirovski1995information}, \cite{bubeck2015convex}. However, it is enough to assume a bound on the variance of stochastic gradients instead of \eqref{e:varaince} to prove convergence guarantees for smooth stochastic optimization without any communication constraints. The stronger assumption made here is to aid a much tighter analysis under communication constraints. In Section \ref{s:extension}, we provide a scheme which can operate under the standard variance bound.}

Denote by $\oO_{\tt sc}$ the set of tuples of function and $n$ clients, $(f, \clientset)$, satisfying \eqref{e:convexity}, \eqref{e:L_smooth}, \eqref{e:asmp_unbiasedness}, \eqref{e:varaince} and \eqref{e:bounded_est}. 
\subsection{Lower bound}
The following bound will serve as a basic benchmark for our problem.
 Let $D>0$
and $\mathbb{X}_2(D) \eqdef \{\X \subseteq \R^d: \max_{x,y\in\X}\norm{x-y}_2\leq
D\}$ be the collection of subsets of $\R^d$ whose $\ell_2$ diameter is
at most $D$.

\begin{thm}\label{t:lb_smooth_coom_constraints}
  There exists an absolute constant $0\leq c_0\leq 1$ such that for $r \leq d $ and $T \geq d/(6nr),$
  %%% CC: We prove 9/58, so 1/6 (>9/58) works and is nicer.
\vspace{-5pt}
\[
 \sup_{\X \in \mathbb{X}_2(D)} \ep^*(\X , \oO_{{\tt sc}}, T, \Q_{r  }) 
\geq 
\frac{c_0D\sigma}{\sqrt{nT}} \cdot \sqrt{\frac{d}{r}} .\]
\end{thm}

\subsection{A general convergence bound}
We present a general convergence bound based on a non-adaptive channel strategy.  In particular,
we fix same quantization process in every iteration, and
% i.e., $Q_{i,t}=Q_i, \forall t\in [T]$, and
%Recall that the corresponding quantized output $Q_i(\hat{g}_i(x_t))$ is denoted by $Y_{i,t}.$ 
the quantized outputs $\{\op_{i,t}\}_{i\in [n]}$ are passed through a mapping\footnote{For instance, averaging the quantized outputs at the server can possibly be one such mapping. } $\newer{\M: \R^{n \times d}\to \R^d}$ in order to update the query.

We use PSGD as the first-order optimization algorithm; the overall optimization procedure is described in Algorithm \ref{a:SGD_Q}.  PSGD proceeds as SGD, with the additional projection step where it projects the updates back to domain $\X$ using the map $\Gamma_{\X}(y):=\min_{x\in \X}\|x-y\|$, $\forall\, y\in \R^d$.

 \vspace{-0.1cm}
\begin{figure}[h]
 \vspace{-0.1cm}
\centering
\begin{tikzpicture}[scale=1, every node/.style={scale=1}]
\node[draw,text width= 6 cm , text height= ,] {%       
\begin{varwidth}{\linewidth}       
            \algrenewcommand\algorithmicindent{0.2em}
\begin{algorithmic}[1]
   \For{$t=0$ to $T-1$}
   %\State ~Observe $\bar{Y}_t=(Y_{1,t},...,Y_{n,t})$
	\State ~$x_{t+1}{=}\Gamma_{\X} \left(x_{t}-\eta_t \M(\op_{1,t},...,\op_{n,t})\right)$
   \EndFor \State Output $\bar{x}_T=\frac 1 T {\sum_{t=1}^T x_t}$
\end{algorithmic}
\end{varwidth}};
 \end{tikzpicture}
 \vspace{-0.1cm}
 \renewcommand{\figurename}{Algorithm}
\caption{PSGD using clients $\clientset$ }\label{a:SGD_Q}
\vspace{-0.1cm}
\end{figure}
The convergence rate of Algorithm~\ref{a:SGD_Q} is controlled by the worst-case $L_2$-norm $\alpha^\prime(\M)$ and the worst-case bias $\beta^\prime(\M)$ defined as
\vspace{-0.25cm}
\begin{align}
&\alpha^\prime(\M):=\sup_{\substack{\{\forall x,i\in [n],\hat{g}_i \in \R^d:\\ {\|\hat{g}_i-\nabla f(x)\|^2}\leq \sigma^2\}}} \sqrt{\E {\| \M(\op^n)-\nabla f(x)\|^2}}, \label{def_a}\\
&\beta^\prime(\M):=\sup_{\substack{\{\forall x,i\in [n],\hat{g}_i \in \R^d:\\ {\|\hat{g}_i-\nabla f(x)\|^2}\leq \sigma^2\}}}\|\E {(\M(\op^n)-\nabla f(x)}\|, \label{def_b}
\vspace{-0.25cm}
\end{align}
where $\op^n=(\op_{1},...,\op_{n})$ is the communication received at the server.
{Using a slight  modification of the standard proof of convergence for PSGD {in \cite[Theorem 6.3]{bubeck2015convex}
}, we can derive the following lemma.}
\begin{lem}\label{l:general}
For any mapping $\M$ and set of quantizers $\{Q_i\}_{i\in [n]}$ defined above, the output $\bar{x}_T$ of optimization algorithm given in Algorithm \ref{a:SGD_Q} satisfies
\begin{align*}
&\sup_{(f,\clientset)\in \cO}\cE(f, \clientset, \A, S) 
\leq \frac{\sqrt{2}\alpha^\prime(\M)D}{\sqrt{T}}+\beta^\prime(\M)\left(D+\frac{DB}{\alpha^\prime(\M)\sqrt{2T}}\right)+\frac{LD^2}{2T},
\end{align*}
with the learning rate 
$\eta_t = \frac{1}{L+\frac{\alpha^\prime(\M)\sqrt{2T}}{D}},\forall t\in [T]$.
%$\eta_t=\min\{\frac{1}{L}, \frac{D}{\alpha^\prime(\M)\sqrt{2T}}\},\forall t\in [T]$. 
%for all iterations $t \in [T]$.
\label{thm1} 
\end{lem}

\subsection{Baseline scheme: {\tt Parallel SGD}}\label{s:blscheme}
We begin by presenting the convergence  result for the baseline scheme in our setup: the {\tt Parallel SGD} algorithm. In {\tt Parallel SGD}, all clients compress their stochastic gradient estimates to $r$ bits using an efficient quantizer for the Euclidean ball and send it to the server, which then takes the average of the quantized gradients for the projected gradient descent step.
We choose subsampled RATQ (\cite{mayekar2020ratq}) for this efficient quantizer. We denote by $Q_{\tt RATQ}$ the subsampled version of RATQ using $r$ bits, which is described in \cite[Section 3.5]{mayekar2020ratq}. \newer{After receiving the quantized outputs $\op_{i,t}=Q_{\tt RATQ}(\hat{g}_i(x_t))~\forall i\in [n],$ from all the $n$ clients, the server takes the mapping $\M$ to be the average of these outputs,} 
%Denote by $\M_t=\M(\bar{\op}_t)$ the average of all the quantized gradients, 
i.e., 
\begin{align}\label{e:Q_ParSGD}
\M(\bar{\op}_t)=\frac{1}{n}\sum_{i=1}^{n}Q_{\tt RATQ}(\hat{g}_i(x_t)).
\end{align}
\begin{thm}\label{t:ParSGD}
  Let $S$ be the  quantizer selection strategy which fixes the quantizer to be $Q_{\mathtt{RATQ}}$ for all clients at all iterations. \newnewest{Let $\A$ be the optimization algorithm described in Algorithm \ref{a:SGD_Q} where $\M$ as described in \eqref{e:Q_ParSGD} is used to make the PSGD step after the $t$-th query and the learning rate $\eta_t = \frac{1}{L+\frac{\alpha^\prime(\M)\sqrt{2T}}{D}} $, where $\alpha^\prime(\M)= c_0\sqrt{\frac{\sigma^2}{ n} +\frac{c_2 d B^2 \log \log^\ast d}{n r} }$ for some positive universal constant $c_0$. } Then,
  for positive universal constants $c_1$ and $c_2$ and 
 $r$ such that ${d \geq r \geq c_1 \log \log^\ast d}$, we have
\begin{align*}
&\cE(f, \clientset, \A, S) \leq  \frac{c_2D}{\sqrt{nT}}\sqrt{\sigma^2 +\frac{c_2 d B^2 \log \log^\ast d}{r} }
+ \frac{L D^2}{2T}.
\end{align*}
\end{thm}

We note that the term  $\frac{ dB^2 \log \log^\ast d}{r}$ illustrates the slowdown in convergence due to quantization error. This is nearly the best rate which can be achieved when one uses  $r$-bit quantizers without any side information\footnote{Similar convergence bounds (upto $\log \log d$ factor) for parallel SGD can be achieved by using subsampled version of rotated quantizer in \cite{suresh2017distributed} or the subsampled version of uniform quantizer after preprocessing due to  Kashin's representation ($cf.$~\cite{kashin1977section}, \cite{lyubarskii2010uncertainty}).}. Note that in cases in which $B$ is large relative to $\sigma^2$,  the slowdown due to this term can be significant, and the algorithm maybe far away from our lower bound in Theorem \ref{t:lb_smooth_coom_constraints}.

\subsection{ {\tt WZ-SGD}: An almost optimal algorithm for distributed optimization} 
\setlength{\floatsep}{10pt plus 2pt minus 2pt}
\begin{figure}[t]

\centering
\begin{tikzpicture}[scale=0.9, every node/.style={scale=0.9}]
\node[draw,text width= 8 cm, text height=,] {%       
\begin{varwidth}{\linewidth}
            
            \algrenewcommand\algorithmicindent{0.7em}
\begin{algorithmic}[1]

%\State Set $\clientset_1 = \{C_1, \ldots, C_{K/2} \}, \clientset_2=\C \setminus \clientset_1$
\For{ \textbf{\color{taorange!50!black}Clients} $i \in [n]$} \  \Comment{Setting quantizers}
\If{$i \in \clientset_1$}
$Q_i =Q_{\mathtt{u}}$ 
%\Statex \Comment{Setting the quantizer to subsampled RATQ for clients in $C_1$.}
\Else~
$ Q_i = Q_{{\tt WZ}, i} $ 
%\Statex \Comment{Setting the quantizer to subsampled RDAQ  for clients in $C_2$.}
\EndIf

\EndFor

%\Statex

\State Initialize $x_0 \in \X$ 
 \For{\newnewest{$t=0$ to $T-1$}}
%\Statex ~
\For{\textbf{\color{red1}Server}}
\State Broadcast $x_t$ to clients
\EndFor
%\Statex ~
 \For{\textbf{\color{taorange!50!black}Clients} $i \in [n]$}  \ \Comment{Encoding}
 \State    Compute $\hat{g}_i(x_t)$
 \State Send $\Qenc_i(\hat{g}_i(x_t))$ to server 

\EndFor
%\Statex ~
\For{\textbf{\color{red1}Server}} \ \Comment{Decoding}
 \For{$i \in \clientset_1$}
 \State  $Q_i(\hat{g}_i(x_t)) = \Qdec_i(\Qenc_i(\hat{g}_i(x_t)))$ 
 
 \EndFor
 \State $Y_t=\frac{2}{n} \sum_{i \in \clientset_1}Q_i(\hat{g}_i(x_t))$ \ \Comment{Side information}
 \For{$i \in \clientset_2$}
 \State  $Q_i(\hat{g}_i(x_t), Y_t) = \Qdec_i(\Qenc_i(\hat{g}_i(x_t)), Y_t)$ 
 %\Statex ~ \Comment{Decoding $\hat{g}_i(x_t)$}
 \EndFor
\State  $x_{t+1}=\Gamma_{\X}\left(x_t - \eta_t\cdot \frac{2}{n} \sum_{i \in \clientset_2}Q_i(\hat{g}_i(x_t), Y_t)\right) $
%\Statex ~\Comment{$Proj_{\X}(y):=\min_{z \in \X}\norm{z-y}_2$}
\EndFor

\EndFor
%\Statex
\State At \textbf{\color{red1}Server~}\textbf{Output:} $\bar{x}_T=\frac{1}{T} \sum_{t \in [T]} x_{t}$
\end{algorithmic}  
\end{varwidth}};
 \end{tikzpicture}
 \renewcommand{\figurename}{Algorithm}
 \caption{$\tt{WZ}$-$\tt{SGD}$ algorithm}\label{a:WZ-SGD}
\vspace{-0.5cm}  
 \end{figure}

We now present our main algorithm : {\tt WZ-SGD}. \newest{ {\tt WZ-SGD} uses our first Wyner-Ziv estimator (see Section \ref{s:sRMQ})  based on subsampled RMQ as a subroutine to form much more accurate gradient estimates compared to  those formed in ${\tt Parallel SGD}.$} As a result of this, {\tt WZ-SGD} significantly improves over the convergence rate of Theorem \ref{t:ParSGD} and relegates the dependence of convergence rate on $B$ to only second order terms.
%\end{rem}

At each iteration $t$, {\tt WZ-SGD} uses the clients in $\clientset_1$ to form the side information estimate $Z_t$ at the server and  then uses the clients in $\clientset_2$ to estimate the gradient for performing the descent step, where\footnote{For simplicity, we assume that $n/2$ and $d/r_1$ are integers such that $d/r_1$ divides $n/2$.} $\clientset_1 \colon{=} \{\tC_1, \ldots, \tC_{n/2} \}, 
  \clientset_2 \colon{=}\clientset \setminus \clientset_1. $

\paragraph{The side information estimate $Y_t$.}
The side information is formed as follows. Under the $r$-bit communication constraint, we divide the coordinates into blocks of dimension $r_1$, where $r_1{:=}r/\log \ell_1$, and $\log \ell_1$ denotes the precision bits used by clients to represent each coordinate in the assigned block. This way we have $d/r_1$ 
blocks.  \newer{We also equally partition the set $\clientset_1$ into $d/r_1$ groups. Further, we assign every block of $r_1$ coordinates to every other distinct group of $\frac{n/2}{d/r_1}$ clients.}
%We assign each block to $Kr_1/(2d)$ clients to form the side information for the coordinates represented by that block. 
To quantize the coordinates within any block, the \newer{group of} clients assigned to that block will use a coordinate-wise uniform quantizer (CUQ). CUQ is an unbiased, uniform quantizer that has appeared recently in many works on gradient quantization.  We denote by $Q_{\tt u}:[-B, B] \to \{-B +2B \cdot (i-1)/(\ell_1-1): i \in [\ell_1]\}$ the \newer{$\ell_1$-level} CUQ quantizer. For a scalar input $x \in [-B, B],$
\begin{align}\label{e:CUQ}
Q_{\tt u}(x)= 
\begin{cases}  
\ceil{\frac{x(\ell_1-1)}{2B}} \cdot \frac{2B}{\ell_1-1}, \quad  \text{w.p.} \quad \frac{x - \floor{\frac{x(\ell_1-1)}{2B}} }{ \frac{2B}{\ell_1-1}},  \\
 \floor{\frac{x(\ell_1-1)}{2B}} \cdot \frac{2B}{\ell_1-1}, \quad  \text{w.p.} \quad \frac{ \ceil{\frac{x(\ell_1-1)}{2B}} -x}{ \frac{2B}{\ell_1-1}}.
\end{cases}
\end{align}
% Specifically, CUQ first uniformly quantizes the input domain $[-B, B]$ to $\ell_1$ levels: $\{-B +2B \cdot (i-1)/(\ell_1-1): i \in [\ell_1]\}$.
%Then, for any input $x\in[-B, B]$, CUQ randomly outputs either the largest level smaller than $x$ or the smallest level larger than $x$ such that the final output is an unbiased estimate of $x$.

Each client uses 
%CUQ with $\ell_1$ \newer{levels} 
{an $\ell_1$-level CUQ to quantize the associated block of coordinates separately.} Thus, the overall communication by each client is $r_1\cdot \log \ell_1 = r$ {and satisfies the communication} constraint.

For each block, we then form the side information by taking the average of the quantized outputs from all its associated clients. 
%\snote{Moved definitions by two paragraphs. Felt little break in the flow earlier.}
Denote by $Y_t$ the side information formed at the server by using the clients in $\clientset_1$ at iteration $t$. Then, from the  description of our scheme, for all coordinates $i \in \{r_1(j-1)+1, \ldots, r_1j\}$ and for all $j \in [d/r_1]$ we have 
\[
Y_t(i)= \frac{2d }{nr_1} \sum_{k \in \mathcal{S}_j}Q_{\tt u}( \hat{g}_{k}(x_t)(i)),
\]
where $\mathcal{S}_j$ denotes the set of $\frac{nr_1}{2d}$ clients assigned to form the side information for the  coordinates $\{r_1(j-1)+1, \ldots, r_1j\}$, i.e.,  
\begin{align}\label{e:SI_WZ}
\displaystyle{\mathcal{S}_j = \{ \tC_{(nr_1/(2d)) \cdot (j-1)+1}, \ldots,  \tC_{(nr_1/(2d)) \cdot j} \}.}
\end{align}
We remark that to decode each quantized gradient estimate sent by clients in $\clientset_2,$ we will use $Y_t$ as side information. However, $Y_t$ will not be used as is but a version which is rotated\footnote{For decoding each quantized gradient sent by clients in $\clientset_2$, $Y_t$ will be rotated using independent and identical versions of matrix $R$.} using a random matrix \eqref{e:R} will be used.

\paragraph{The Wyner-Ziv gradient estimate {$Q_{{\tt WZ }}$}.}
We use the clients in $\clientset_2$ to form the actual gradient estimate. The clients encode the stochastic gradients using a subsampled RMQ quantizer (see  Section \ref{s:sRMQ} for details). Therefore, for stochastic gradient $\hat{g}_j(x_t),$ the output encoded by client $\tC_j$ using subsampled RMQ is described as follows:
\[\displaystyle{ \Qenc_{{\tt WZ}, j}( \hat{g}_j(x_t))=\{\Qenc_{{\tt M}}(R_j\hat{g}_j(x_t)(i)): i \in \mathcal{D}_j \} }.\]

At the server, the communication for all $\tC_j \in \clientset_2$ is decoded as follows:
\vspace{-0.2cm}
\eq{ Q_{{\tt WZ}, j}( \hat{g}_j(x_t), Y_t){=}R_j^{-1}\left({\frac{d}{r_2}}\sum_{i\in \mathcal{D}_j}\left(\tilde{g}_j -R_jY_t(i) \right)e_i  +R_jY_t\right)}
where $ \tilde{g}_j(i)=Q_{\tt M}(R_j\hat{g}_j(x_t)(i), R_jY_t(i)).$ Finally, the server averages over all the quantized gradient estimates of clients in $\clientset_2$ to get (see, line 17 in Algorithm \ref{a:WZ-SGD})
%$\M_t=\M(Q_{{\tt WZ}, 1},\dots, Q_{{\tt WZ}, n})$, which in turn is used to make the PSGD step in line 2 of Algorithm \ref{a:SGD_Q}. That is, 
\begin{align}\label{e:descent_WZ_SGD}
\M(Q_{{\tt WZ}, 1},\dots, Q_{{\tt WZ}, n})=\frac{2}{n}\sum_{j=n/2+1}^{n}Q_{{\tt WZ}, j}( \hat{g}_j(x_t), Y_t)
\end{align}

Next, we present the convergence rate of the proposed $\tt{WZ\text{-}SGD}$ algorithm \newer{for communication constrained distributed optimization}.
\begin{thm} \label{t:conv_WZ_SGD}
  Let $S$ be the communication protocol which uses the CUQ quantizer for clients $\clientset_1$ and the subsampled RMQ quantizer for clients in $\clientset_2$. Let $\A$ be the optimization algorithm described in 
  \newnewest{Algorithm \ref{a:WZ-SGD} with the learning rate $\eta_t = \frac{1}{L+\frac{\alpha^\prime(\M)\sqrt{2T}}{D}}$, where $\alpha^\prime(\M)= c_0\sqrt{\frac{d \sigma^2 \log \log nT}{n r} }$ for some positive universal constant $c_0$.} Then, for positive
universal constants $c_1, c_2,$ and $c_3$ and
  $r, n$ such that $d \geq r \geq c_1\max\{\log \log n T,\log (B/\sigma)\}$ and $nr \geq c_2 d^2 \log (B/\sigma)$, we have
\begin{align*}
\cE(f, \clientset, \A, S)  \leq \frac{c_3D\sigma}{\sqrt{nT}}\cdot \sqrt{\frac{d \log \log nT}{r} } +\frac{LD^2}{2T}.
\end{align*}
%% where $C_1,$ $C_2,$ and $C_3$ are universal constants.

\end{thm}
\begin{rem}
The condition on $nr$ is needed to remove any $B$ dependence from the MSE upper bound.
\end{rem}
Thus, in the setting where the number of clients $n$ is large, we match the lower bound in Theorem \ref{t:lb_smooth_coom_constraints} upto a $\log \log nT$ factor.

\newest{\subsection{{\tt UWZ-SGD:} A universal Wyner-Ziv algorithm for distributed optimization}\label{s:extension}
\setlength{\floatsep}{10pt plus 2pt minus 2pt}
\begin{figure}[t]

\centering
\begin{tikzpicture}[scale=0.9, every node/.style={scale=0.9}]
\node[draw,text width= 8 cm, text height=,] {%       
\begin{varwidth}{\linewidth}
            
            \algrenewcommand\algorithmicindent{0.7em}
\begin{algorithmic}[1]

%\State Set $\clientset_1 = \{C_1, \ldots, C_{K/2} \}, \clientset_2=\C \setminus \clientset_1$
\For{ \textbf{\color{taorange!50!black}Clients} $i \in [n]$} \ \Comment{Setting quantizers}
\If{$i \in \clientset_1$}
$Q_i =Q_{\mathtt{RATQ}}$ 
%\Statex \Comment{Setting the quantizer to subsampled RATQ for clients in $C_1$.}
\Else~
$ Q_i = Q_{\tt WZ,u} $ 
%\Statex \Comment{Setting the quantizer to subsampled RDAQ  for clients in $C_2$.}
\EndIf

\EndFor

%\Statex

\State Initialize $x_0 \in \X$ 
 \For{\newnewest{$t=0$ to $T-1$}}
%\Statex ~
\For{\textbf{\color{red1}Server}}
\State Broadcast $x_t$ to clients
\EndFor
%\Statex ~
 \For{\textbf{\color{taorange!50!black}Clients} $i \in [n]$}  \ \Comment{Encoding}
 \State    Compute $\hat{g}_i(x_t)$
 \State Send $\Qenc_i(\hat{g}_i(x_t))$ to server 

\EndFor
%\Statex ~
\For{\textbf{\color{red1}Server}} \ \Comment{Decoding}
 \For{$i \in \clientset_1$}
 \State  $Q_i(\hat{g}_i(x_t)) = \Qdec_i(\Qenc_i(\hat{g}_i(x_t)))$ 
 
 \EndFor
 \State $Y_t=\frac{2}{n} \sum_{i \in \clientset_1}Q_i(\hat{g}_i(x_t))$ \ \Comment{Side information}
 \For{$i \in \clientset_2$}
 \State  $Q_i(\hat{g}_i(x_t), Y_t) = \Qdec_i(\Qenc_i(\hat{g}_i(x_t)), Y_t)$ 
 %\Statex ~ \Comment{Decoding $\hat{g}_i(x_t)$}
 \EndFor
\State  $x_{t+1}=\Gamma_{\X}\left(x_t - \eta_t\cdot \frac{2}{n} \sum_{i \in \clientset_2}Q_i(\hat{g}_i(x_t), Y_t)\right) $
%\Statex ~\Comment{$Proj_{\X}(y):=\min_{z \in \X}\norm{z-y}_2$}
\EndFor

\EndFor
%\Statex
\State At \textbf{\color{red1}Server~}\textbf{Output:} $\bar{x}_T=\frac{1}{T} \sum_{t \in [T]} x_{t}$
\end{algorithmic}  
\end{varwidth}};
 \end{tikzpicture}
 \renewcommand{\figurename}{Algorithm}
 \caption{$\tt{UWZ}$-$\tt{SGD}$ algorithm}\label{a:UWZ-SGD}
\vspace{-0.5cm}  
 \end{figure}

We now relax the almost sure \eqref{e:varaince} assumption on the gradients estimated by clients and present an {\em universal} algorithm $\tt UWZ$-$\tt SGD$, \newnewest{where the compression at the clients doesn't need the knowledge of $\sigma$ and only the server needs to know $\sigma$ to set the learning rate in Algorithm \ref{a:SGD_Q}. Specifically,  we assume that for all clients $i \in [n]$, 
\begin{align}\label{e:var}
\E{\norm{\hat{g}_i(x) - \nabla f(x)}^2} \leq \sigma^2. \  \text{(m.s.  deviation bound)}
\end{align}
The other assumptions \eqref{e:asmp_unbiasedness} and \eqref{e:bounded_est} about the estimated gradients still hold\footnote{Note that the lower bound in Theorem \ref{t:lb_smooth_coom_constraints} under the almost sure assumption \eqref{e:varaince} holds for the relaxed mean-squared assumption \eqref{e:var} too. }. We show how the dependence of $B$ in the naive scheme, presented in Theorem~\ref{t:ParSGD}, can be reduced using subsampled RDAQ.}

At every iteration, the client indexed by $\clientset_1$ use  \newnewest{subsampled RATQ to compress their gradient estimates. The side information is then formed by taking sample average of the decoded estimates, similar to \eqref{e:SI_WZ} (see line 14, in Algorithm \ref{a:UWZ-SGD}).

On the other hand, the clients in $\clientset_2$ use the subsampled RDAQ quantizer $Q_{\tt WZ,u}$ from section \ref{s:sRDAQ}. Note that the
subsampled RDAQ decoder \eqref{e:s_RDAQ} uses the side information constructed by $\clientset_1$. Finally, the server takes the sample average of the decoded values estimated by the $\clientset_2$ (see, line 17 in Algorithm \ref{a:UWZ-SGD}) to form the mapping $\M$. 
%We denote the overall $r$-bit quantization strategy by $Q_{\mathtt{UWZ}}$ and present its convergence rate below.
\begin{thm}\label{c:PSGD+}
Let $S$ be the communication protocol which uses the subsampled RATQ quantizer for clients $\clientset_1$ and the subsampled RDAQ quantizer for clients in $\clientset_2.$ Let $\A$ be the optimization algorithm described in Algorithm \ref{a:UWZ-SGD} with the learning rate $\eta_t=\frac{1}{L+\frac{\alpha^\prime(\M)\sqrt{2T}}{D}}$, where $\alpha^\prime(\M)=\sqrt{\frac{2\sigma^2}{n}+\frac{2\rho(B,\sigma,r,n)}{n}}$ with 
${\rho{=}\sqrt{\sigma^2{+}\frac{2\sigma^2}{n}{+}\frac{2dB^2}{n\left(\frac{r}{ 3+ \ceil{\log(1+ \ln^\ast(d/3))}} -1\right)}}\frac{16\sqrt{3}dB}{\frac{r}{\ceil{h+\log h}}-1}}$ and ${h=1+ \ln^\ast(d/6)}$.
Further,  suppose that the gradient estimated by all the clients satisfy the assumptions \eqref{e:asmp_unbiasedness}, \eqref{e:var}  and \eqref{e:bounded_est}. Then, for ${d \geq r \geq}$ $\max\{h+\log h,  3+ \ceil{\log(1+ \ln^\ast(d/3))}\}$, we have
%For $\mathtt{UWZ-SGD}$ described above with clients satisfying assumptions \eqref{e:asmp_unbiasedness}, \eqref{e:var}  and \eqref{e:bounded_est},  and using the quanitzation strategy $Q_{\mathtt{UWZ}}$ in every iteration, such that ${d \geq r \geq}$$\max\{h+\log h,  3+ \ceil{\log(1+ \ln^\ast(d/3))}\}$, we have
\begin{align*}
\cE(f, \clientset, \A, S)
&\leq  \frac{2D}{\sqrt{nT}}\sqrt{\sigma^2 +\rho(B, \sigma, r, n)}+\frac{LD^2}{2T}.
\end{align*}
%where ${\rho{=}\sqrt{\sigma^2{+}\frac{2\sigma^2}{n}{+}\frac{2dB^2}{n\left(\frac{r}{ 3+ \ceil{\log(1+ \ln^\ast(d/3))}} -1\right)}}\frac{16\sqrt{3}dB}{\frac{r}{\ceil{h+\log h}}-1}}$ with ${h=1+ \ln^\ast(d/6).}$
\end{thm}}

\begin{rem}
We remark that under the relaxed assumption of mean-square bounded deviation in \eqref{e:var}, for $nr\geq (B^2/\sigma^2)d\log(1+\ln^\ast(d/3))$,  the slowdown in the convergence rate is illustrated by $\rho\approx \frac{16\sqrt{3}dB\sigma\ln^\ast d}{r}$,
%$\frac{D\sqrt{B\sigma}}{\sqrt{KT}}\sqrt{\frac{d}{r}},$ 
 and the universal scheme surpasses the performance of parallel SGD presented in Section \ref{s:blscheme}. 
 %Further, the scheme does not require the knowledge of $\sigma$
 %in \eqref{e:var}.
\end{rem}}

\newest{We end this section by pointing out limitations of a natural scheme for distributed optimization.
\begin{rem}[Limitations of Centering Based Scheme]
We note that our framework allows for quantization schemes were previously quantized gradients are used for  gradient compression at the current iteration. For instance, we can use average of the compressed gradients at the previous iteration to center the current compression. That is, the server broadcasts the average to all the clients and the clients only need to compress the difference between the current stochastic gradient and this communicated average.

If the query points $x_{t-1}$ and $x_t$ do not deviate by much, then such type of compression schemes which are centered around the average of previous quantized gradients may turn out to be very efficient. Also, note that the typical learning rate for smooth optimization is $O(\sqrt{1/T})$, which means that the  difference between the points $x_t$ and $x_{t-1}$ is not very large. Moreover, the smoothness assumption \eqref{e:L_smooth} allows to control the deviation between the true gradients at successive iterations in terms of the points queried at the two iterations. All this hints at the fact that such a scheme where each client uses optimal quantizers for quantizing the difference vector without any side-information may turn out to be optimal. But note that for a very large value of smoothness constant, $L\geq \sigma \sqrt{T}/ D$,  even with small deviation between successive query points, the deviation between the gradients will be large. This would in turn lead to variance of the quantized gradients having a dependence on the maximum gradient norm $\sqrt{B^2-\sigma^2}$, which would in turn lead to the leading term, in terms of $n$ and $T$, in convergence rate depending on $\sqrt{B^2-\sigma^2}$.
\end{rem}}

\section{The Gaussian Wyner-Ziv problem}\label{s:GWZ}
Consider the random vectors $X$ and $Y$,  where the coordinates $\{X(i), Y(i)\}_{i=1}^{d}$ form an i.i.d. sequence. Furthermore, for all $i \in [d]$, let
\[
X(i)=Y(i)+Z(i),
\]
where $Y(i)$ and $Z(i)$ are independent and zero-mean Gaussian random variables with variances $\sigma_y^2$ and $\sigma_z^2$, respectively. The encoder has access to the sequence $X=\{X(i)\}_{i=1}^{d}$, which it quantizes and sends to the decoder. The decoder, on the other hand, has access to $Y$ (note that encoder does not have access to $Y$) and can use it to decode $X$. A pair $(R, D)$ of non-negative numbers is an achievable rate-distortion pair if we can find a quantizer $Q_d$ of precision $dR$
and with mean square error $\E{\norm{Q_d(X, Y)-X}_2^2}\leq dD$.  For $D \geq 0$, denote by $R(D)$ the infimum over all $R$ such that $(R, D)$ constitute an achievable rate-distortion pair for all $d$ sufficiently large. From\footnote{The model considered in \cite{wyner1976rate} and perhaps the more popular Wyner-Ziv model is $Y=X+Z$. Nevertheless, through MMSE rescaling this model can be converted to $X=Y^{\prime}+Z^{\prime}$ (see, for instance, \cite{liu2016polar}).} \cite{wyner1976rate}, $R(D)$ can be characterized as follow:
\eq{
R(D) = \begin{cases} \frac{1}{2}\log \frac{\sigma_z^2}{D} \quad &\text{if} \quad D\leq \sigma_z^2,\\
0 \quad &\text{if} \quad D > \sigma_z^2.
\end{cases}
}
Several constructions that involve computational heavy methods such as error correcting codes and lattice encoding attain the rate-distortion function, asymptotically for large $d$.
 In this section, we show that modulo quantizer with parameters set appropriately attains a rate very close to the rate-distortion function $R(D)$. Moreover, we will show that this rate can be  achieved for arbitrary $Y$ and $Z$, as long as $Z$ is a zero mean subgaussian random variable with variance factor $\sigma_z^2$. Our proposed quantizer $Q_d(X, Y)$ uses the modulo quantizer to quantize $X(i)$ with side information $Y(i)$ at the decoder and the parameter $k, \Delta^{\prime}$ set as follows:
 \begin{align}\label{e:WZ_MQ}
 \nonumber
&\delta = \sqrt{D/308}, 
\quad
\log k= \ceil{\log \left(2+ \sqrt{
\frac{24\sigma_z^2}{D} \ln \frac{308\sigma_z^2}{D}} \right)}\\
&
\Delta^{\prime}= \sqrt{6(\sigma_z^2)\ln (\sigma_z/\delta)},
\quad 
\eps=2\Delta^\prime/(k-2),
 \end{align}
 
 \begin{thm}\label{t:G_WZ}
 Consider random vectors $X, Y$ in $\R^d$
 with $X(i)=Y(i)+Z(i)$ and 
 $Z(i)$ independent of $Y(i)$ being a centered subgaussian random variable with variance factor of $\sigma_z^2$, for all coordinates $i \in \{1, \ldots, d\}$.  
 Then, 
 for $D \leq ({\sigma_z^2}/{308})$,
 the quantizer $Q_d(X, Y)$ described above
  has MSE less than $dD$ and has rate $R$ satisfying
 \[R \leq \frac{1}{2}\log \frac{\sigma_z^2}{D}+O\left(\log \log \frac{\sigma_z^2}{D} \right).\]
 \end{thm}
%\begin{rem} We have the following remarks here about our universal scheme described above:
%\begin{enumerate}
%\item It works under a relaxed assumption of mean-square bounded deviation in \eqref{e:var} than that in  \eqref{e:varaince}.
%\item The scheme does not require the knowledge of variance parameter $\sigma$ in \eqref{e:var}.
%\item For $Kr\geq d\log(1+\ln^\ast(d/3))$,  the gap to the lower bound is $O(\frac{DB}{\sqrt{KT}}\sqrt{\frac{d}{r}})$. In this regime, the universal scheme performs almost comparable to the parallel SGD ($c.f.$ Theorem \ref{t:ParSGD}). Note that the performance can be further improved  with increase in number of clients $K$ as that lessens the effect of $B^2$ in the square root term of $\rho$ above. Specifically, for $Kr\geq (B^2/\sigma^2)d\log(1+\ln^\ast(d/3))$, the gap to the lower bound gets reduced to $\frac{D\sqrt{B\sigma}}{KT}\sqrt{\frac{d}{r}},$ and the universal scheme surpasses the perfomance of parallel SGD.
% \end{enumerate}
% \end{rem}
% \snote{1. describing log star notation. 2. Comment on the 1/T term in upper bounds but not in lower bounds.}

\section{The high-precision regime}\label{s:hpr}
\subsection{RMQ in the high-precision regime.}
For the \known setting, our quantizer RMQ described in Alg. 4 and 5
remains
valid even for $r>d$. We will assume $r=m d$ for integer $m \geq 2$. For each client $i$, we set
\begin{align}
\delta = \frac{\Delta_i}{n^{\frac 12}(2^{r/d}-2)}, \quad \log k = \frac{r}{d}, \quad 
\Delta^\prime &= \sqrt{6(\Delta_i^2/d)\ln \Delta_i/\delta}, \quad \eps = \frac{2\Delta^{\prime}}{k-2}.
 \label{e:param_high_precision_known}
 \end{align}

The performance of 
protocol $\pi^*_k$  using RMQ with parameters set as in \eqref{e:param_high_precision_known} for each client can be characterized as follows.
\begin{thm}
For a fixed $\mathbf{\Delta}=(\Delta_1, ...,\Delta_n)$
and  $r=md$ for integer $m\geq 2$,
the protocol $\pi^*_k$ with parameters set as in  \eqref{e:param_high_precision_known}
satisfies 
\[
\MSE(\pi^*_k, \mathbf{x}, \mathbf{y})\leq \left(
12 \ln n+ \frac{24r}{d}+ {154}/{n}+
166
\right)
\left(\sum_{i \in
[n]}\frac{\Delta_i^2}{n} \cdot \frac{1}{n(2^{r/d}-2)^2} \right), 
\]
for all $\mathbf{x},\mathbf{y}$ satisfying~\eqref{eq:delta_cond}.
\end{thm}
\begin{proof}
Denoting by $Q_{i}$ the quantizer $Q_{{\tt M},R}$
with parameters set for client $i$,
by Lemmas~\ref{l:main} and~\ref{t:RMQ}, we get
\begin{align*}
{\E{ \|\hat{\bar{x}} - \bar{x}\|_2^2}}
&\leq \sum_{i=1}^n\frac{\alpha(Q_i; \Delta_i)}{n^2} +  \sum_{i=1}^n\frac{\beta(Q_i; \Delta_i)}{n}
\end{align*}
Further, since $k\geq 4$ holds when $r\geq 2d$ for our choice of parameters,
by using Lemma~\ref{t:RMQ} and substituting $\delta^2=\Delta_i^2/n(2^{r/d}-2)^2$, we get
  \begin{align*}
\alpha(Q_{i}; \Delta_i)&\leq \frac{12 \Delta_i^2\ln (n(2^{r/d}-2)^2)}{{(2^{{r}/{d}}-2)}^2}
+ \frac{ 154\Delta_i^2}{n(2^{{r}/{d}}-2)^2},
\\
\beta(Q_{i}; \Delta_i)&\leq  \frac{154\Delta_i^2}{n(2^{{r}/{d}}-2)^2}.
\end{align*}

which with the previous bound gives
\begin{align*}
{\E{ \|\hat{\bar{x}} - \bar{x}\|_2^2}}
&\leq 
\left(
12 \ln n+\frac{24r}{d}+ \frac{154}{n}+
154
\right)
\sum_{i=1}^n\frac{\Delta_i^2}{n^2(2^{r/d}-2)^2},
\end{align*}
where use the inequality $ \ln x \leq x$, $\forall x \geq 0$, to bound $\ln (2^{r/d}-2)^2/(2^{r/d}-2)^2$ by $1$.

\end{proof}

\begin{rem}
 Similar to Remark \ref{r:DMQ_lp}, we note that using MQ for each coordinate without rotating (or even with rotation using $R$ as above) with $\Delta^{\prime}=\Delta_i$ yields MSE less than 
 \[ O\left( \sum_{i=1}^{n}\frac{\Delta_i^2}{n}\cdot \frac{d}{n2^{2r/d}}\right),\]
 for $r \geq d$. Thus our approach above allows us to remove the $d$ factor at the cost of a (milder for large $d$) $\log n +r/d$ factor.

\end{rem}

%%%%
\subsection{Boosted RDAQ: RDAQ in the high-precision regime.}
Moving to the \unknown setting, we describe an update to RDAQ
described in Alg. 10 and 11 for the high-precision setting.
For brevity, we denote by $m:=r/d$  the number of bits per dimension.
A straight-forward scheme to make use of the
high precision is to independently implement the RDAQ quantizer
approximately $ \floor{m/\ln^*d}$ times and use the average of the
quantized estimates as the final estimate.
We will see that the MSE incurred by such an
estimator is $O(\Delta \ln^*d/m)$.
We will show that this naive
implementation can be significantly improved and an exponential decay
in MSE with respect to $m$ can be achieved.

We boost RDAQs performance as follows.
Simply speaking, instead of sending the bits produced by multiple
instances of the encoder of RDAQ, we send the ``type'' of each sequence.
A similar idea appeared in~\cite{mayekar2020limits} for the case without
any side information.
At the encoding stage of RDAG
given in Alg. 10 and 11, after random rotation and computing
$z$ in Steps $1$ to $3$ of Alg. 10, we repeat
Step $4$
$N$ times with independent randomness each time and store
only the total number of ones seen for each coordinate $i$ and scale $j$.
Specifically, let $U_t(i, j)$ be an independent uniform random variable in
$[-M_j,M_j]$, for all $i \in [d], j\in [h]_0$, and $t \in [N]$, which are generated using
public randomness between the encoder and the decoder. Using this
randomness, we compute  
$\tilde{x}_{j,t}= \sum_{i=1}^{d} \indic{\{U_t(i,j)\leq x_R(i)\}}e_i$
for all $j \in [h]_0$. Then, instead of storing  
$\tilde{x}_{j,t}$ for each $j$ and $t$, we store the sum
$\sum_{t=1}^n\tilde{x}_{j,t}$ for each $j\in [h]_0$.
Since each coordinate of the sum can be stored in $\log (N+1)$ bits,
the new encoder's output can be stored in  $d(h\log (N+1) +\log h)$.
Thus, we can implement this scheme by using $m=(h\log (N+1) +\log h)$
bits per dimension.

At the decoding stage, we rotate $y$ and compute $z^*$ in precisely
the same manner as done in Steps 1 to 3
of the decoding Alg. 11 of RDAQ. Then,
  using the encoded input received, the side-information $y$, the same
  random variables $U_t(i,j)$ and random matrix $R$ used by the
  encoder, the final estimate $Q(x)$ is 
  \begin{align}\label{e:boosted_RDAQ}
    &Q(x)=R^{-1} \left( \frac{1}{N}\cdot \sum_{i \in [d]}\sum_{t \in
  [N]} \left( B^t_{i, Rx} - B^t_{i, Ry} \right)  e_i +Ry\right), 
  \end{align}
  where $B^t_{i, v} = \indic{\{U_t(i,z^*(i) ) \leq v(i)\}}$ for $v$ in
  $\R^d$. 

The result below characterizes the performance of our quantizer Boosted RDAQ
$Q$.
  \begin{lem}\label{t:bRDAQ}
Let $Q$ be Boosted RDAQ described above.
Then,  we have for $\X=\Y=\R^d$ and every $\Delta > 0$,
we have 
\[
\alpha_u(Q; \Delta) \leq  \frac{ 16 \sqrt{3}\Delta}{N}\,\,
\text{ and }\,\,
    \beta_u(Q; \Delta) =0.
    \]
    Furthermore, the output of the quantizer can be described in
    $d(h \log N + \log h)$ bits. 
  \end{lem}
  Thus, when we have a total precision budget of $r=dm$ bits using the
    Boosted RDAQ algorithm with  number of repetitions $N=
    2^{\floor{(m -\log h)/h}}-1$, we get an exponential decay in MSE
    with respect to $m$. 

We consider the protocol $\pi^*_u$ that uses the $Q$ above for each
client with $M_j$ and $h$ set as in \eqref{e:levels}, $i.e.$, with
\begin{align}
N=2^{\floor{(m -\log h)/h}}-1,\,\, M_j^2= \frac{6e^{*j}}{d}, j\in[h]_0, ~~ \log h =\lceil \log (1+\ln^*(d/6) )\rceil.
\label{e:param_high_precision_unknown}
\end{align}
Therefore, by the previous lemma and Lemma \ref{l:main}, we 
get the following result.

\begin{thm}
For $r=dm$ with integer $m\geq h + \log h$,
the protocol $\pi^*_u$ with parameters as set in \eqref{e:param_high_precision_unknown}
satisfies
\[
\MSE(\pi^*_u, \mathbf{x}, \mathbf{y})\leq \sum_{i \in [n]}\frac{\Delta_i}{n} \cdot \frac{64\sqrt{3}}{n2^{{r}/(d(2+2\ln^*(d/6)) )}}
,
\]
for all {$\mathbf{x},\mathbf{y}$} satisfying \eqref{eq:delta_cond},
for every $\mathbf{\Delta}=(\Delta_1, ...,\Delta_n)$.
\end{thm}
\begin{proof}
Denote by $\hat{\bar{x}}$ the output of the protocol. Then,
by Lemmas~\ref{l:main} and Lemma~\ref{t:bRDAQ}, we get
\begin{align*}
\E{ \|\hat{\bar{x}} - \bar{x}\|_2^2}
&\leq \frac 1 {n^2} \sum_{i=1}^n\alpha(Q; \Delta_i)
 \\
 &\leq 
 \frac {16\sqrt{3}} {n^2 N } \sum_{i=1}^n\Delta_i,
\end{align*}
where the previous inequality is by Lemma~\ref{t:bRDAQ}. 
The 
proof is completed by using 
\[
N\geq
\frac{2^{m/h}}{2^{1+(\log h)/h}}\geq \frac{2^{m/h}}{4} \geq \frac{2^{m/(2+2\ln^*(d/6))}}{4},\]
where the first inequality follows from using $\lfloor x\rfloor \geq x-1$ for the floor function in the value of $N$ in \eqref{e:param_high_precision_unknown}, the second follows from  the fact that $\log x \leq x, \forall x \geq 0$, and the third follows from $\lceil x\rceil \leq x+1$  for the ceil function in the value of $h$ in \eqref{e:param_high_precision_unknown}.
\end{proof}

\section{Numerical Experiments}\label{s:exp}

\newest{We empirically demonstrate the performance of our proposed quantizers on the following mean estimation task.

Each client $i$ has a $d$-dimensional vector $x_i=\mu + U^{c}_i,$ where $\mu$ in $[0, 1]^d$ is constant mean vector and $U^{c}_i$ is a random vector whose each coordinate is a Uniform random variable in $[-\Delta^{\prime}/2, \Delta^{\prime}/2].$ The server has side information $y_i$ corresponding to  $x_i$, where $y_i=\mu + U^{s}_i$, and $U^{s}$, too, is a random vector whose each coordinate is a Uniform random variable in $[-\Delta^{\prime}/2, \Delta^{\prime}/2].$ Note that the distance between each coordinate of $x_i$ and $y_i$ is bounded by $\Delta^\prime$.

We compare three different mean estimation protocols. The first protocol is our first Wyner-Ziv estimator that uses  RMQ for all the clients. Note that this protocol uses the knowledge of $\Delta^\prime$ to set the values of RMQ.  The second protocol is our universal Wyner-Ziv estimator which uses RDAQ for all the clients.  Here, instead of vanilla RDAQ, we will use boosted RDAQ to make use of all the available precision. Recall that this particular protocol operates without the knowledge of $\Delta^{\prime}.$ Our third protocol uses RATQ for all the clients, an efficient quantizer for the $\ell_2$ ball \cite{mayekar2020ratq}. Note that this protocol neither uses the side information $y_i$ nor the distance between side information and the input vectors and will serve as a baseline. We evaluate  the performance of our protocols by root mean square error (RMSE) between  $\bar{x},$ the sample average of $x_i$s, and its estimate formed by the server $\hat{\bar{x}}$.

We fix the number of clients $n=10$. We conduct the experiments at  dimensions  $d=512$ and $d=1024$, and at two different precision levels: $6$ bits per dimension and $10$ bits per dimension.   For all these four experiments we track the performance of our three quantization protocols by changing $\Delta^{\prime}$. All the experiments are averaged over ten runs for statistical consistency. Our implementation is available online at GitHub\footnote{\url{https://github.com/shubhamjha-46/WZ\_estimators}.}.

We use the following parameters for all the quantizers. For RMQ, we set $\epsilon = \frac{2\Delta^\prime}{31}$ and $\frac{2\Delta^\prime}{511}$ for precision $6$ bits and $10$ bits, respectively. For RDAQ and RATQ, we first normalize the vectors $\{x_i, y_i\}_{i\in [n]}$ using an the bound $\sqrt{d}(1+\Delta^\prime/2)$ on their $\ell_2$-norm. Then, we set\footnote{For RATQ too, we set $h=4$ (see \cite{mayekar2020ratq} for more details).} $h=4$ to compute the different scales $M_j$s in \eqref{e:param_high_precision_unknown} for dimensions $d=512, 1024$. In addition, we choose $N=1$ and $N=3$ for implementing $6$ bit and $10$ bit Boosted RDAQ, respectively. 
The final estimate is obtained by multiplying back the decoded output with $\sqrt{d}(1+\Delta^\prime/2)$.

%\paragraph{Performance of RMQ} 
\setcounter{figure}{0}   
  \pgfplotsset{compat=1.12, every axis/.append style={
  		line width=1pt, tick style={line width=0.6pt}}, width=7.2cm,height=5.5cm, tick label style={font=\tiny},
  	label style={font=\small},
  	legend style={font=\tiny}}
We see in Figures \ref{fig1}, \ref{fig2}, \ref{fig3}, and \ref{fig4} that RMQ comfortably outperforms the other two quantizers at possible parameter choices. This is expected, since the RMSE of RMQ is directly proportional to $\Delta^{\prime}$, which is very small in our experiments. Another consequence of this relation to $\Delta^{\prime}$ is that RMSE increases at a much faster rate with increase $\Delta^{\prime}$ for RMQ than any other protocol. In other words, the performance of RMQ will degrade at a much faster rate than RDAQ as the accuracy of side information degrades.

As can be seen in Figures \ref{fig1}, \ref{fig2}, RDAQ outperforms RATQ at $6$ bits per dimension and both values of dimension. At precision level of $10$ bits per dimension, however, RDAQ is better than RATQ at lower values of $\Delta^{\prime}.$

In other direction, we note that for all our protocols there is slight increase in  RMSE for the same $\Delta^{\prime}$ and bit precision as the dimension increases from $512$ to $1024$. This is because $\ell_2$ norm of the input and the $\ell_2$ distance between input and side information depend on the dimension for our example, and our MSE upper bounds for all the quantizers depend on either one or both of these quantities.

\newnewest{Finally, we end with a remark on our choice of precision levels of $6$ bits and $10$ bits per dimension for this experiment. Notice that similar trends can be observed for precision levels lesser than dimension $d$. However,  setting close to optimal parameters for these quantizers would have been much more tedious at precision levels lesser than the dimension. Since our experiment aimed to study the impact of side information on the accuracy of distributed mean estimation, we chose not to experiment with precision levels lesser than the dimension. The reason for not experimenting at $1$ or $2$ bits per dimension is that RDAQ is not operational below $6$ bits per coordinate for the current dimension.}

\begin{figure}[t]
\centering
\begin{minipage}{0.45\textwidth}
\centering
\begin{tikzpicture}
\begin{axis}[xmode = log, ymode=log, ytick={.0001, 0.001, 0.01, 0.1, 1, 10}, xtick = {0.0625e-02, 0.00125, 0.0025, 0.005, 0.01, 0.02, 0.04, 0.08},
ytick style={draw=none}, xtick style={draw=none}, legend style={nodes={scale=0.6}, font=\footnotesize},
 	xlabel= $\Delta^\prime$,
 	ylabel= RMSE,
    legend pos= south east,
 	xmin=0.000625,
 	xmax=0.08,
 	ymin=0,
 	ymax=10,  grid=major, grid style={dotted, gray}
 	]
\addplot [mark=square, color=teal]
table
[x expr=0.000625*(2^\coordindex), 
y expr=\thisrowno{0}
] {RMQ/RMSE_RMQ_6_512.dat}; 	
 	
\addplot [mark=otimes*, color=ta2gray]
table
[x expr=0.000625*(2^\coordindex), 
y expr=\thisrowno{0}
] {B_RDAQ/RMSE_B_RDAQ_6_512.dat}; 	

\addplot [mark=triangle, color=ta2orange]
table
[x expr=0.000625*(2^\coordindex), 
y expr=\thisrowno{0}
] {RATQ/RMSE_RATQ_6_512.dat}; 

\legend{ RMQ\\RDAQ\\RATQ\\}
\end{axis}
\end{tikzpicture}
\caption{ Comparison of RMQ, RDAQ, and RATQ at per coordinate precision of $6$ bits and $d=512$.}	
\label{fig1}
\end{minipage}
\hspace{0.05\textwidth}
\begin{minipage}{0.45\textwidth}
\centering
\begin{tikzpicture}
\begin{axis}[xmode = log, ymode=log, ytick={.0001, 0.001, 0.01, 0.1, 1, 10}, xtick = {0.0625e-02, 0.00125, 0.0025, 0.005, 0.01, 0.02, 0.04, 0.08}, ytick style={draw=none}, xtick style={draw=none}, legend style={nodes={scale=0.6}, font=\footnotesize},
 	xlabel= $\Delta^\prime$,
 	ylabel= RMSE,
    legend pos= south east,
 	xmin=0.000625,
 	xmax=0.08,
 	ymin=0,
 	ymax=10,  grid=major, grid style={dotted, gray}
 	]
\addplot [mark=square, color=teal]
table
[x expr=0.000625*(2^\coordindex), 
y expr=\thisrowno{0}
] {RMQ/RMSE_RMQ_6_1024.dat}; 	
 	
\addplot [mark=otimes*, color=ta2gray]
table
[x expr=0.000625*(2^\coordindex), 
y expr=\thisrowno{0}
] {B_RDAQ/RMSE_B_RDAQ_6_1024.dat}; 	

\addplot [mark=triangle, color=ta2orange]
table
[x expr=0.000625*(2^\coordindex), 
y expr=\thisrowno{0}
] {RATQ/RMSE_RATQ_6_1024.dat}; 

\legend{ RMQ\\ RDAQ\\ RATQ\\}
\end{axis}
\end{tikzpicture}
\caption{ Comparison of RMQ, RDAQ, and RATQ at per coordinate precision of $6$ bits and $d=1024$.}	
\label{fig2}
\end{minipage}
\end{figure}

\begin{figure}[t]
\centering
\begin{minipage}{0.45\textwidth}
\centering
\begin{tikzpicture}
\begin{axis}[xmode = log, ymode=log, ytick={0.000001, 0.00001, 0.0001, 0.001, 0.01, 0.1, 1, 10}, xtick = {0.00015625, 0.0003125, 0.000625, 0.00125, 0.0025, 0.005, 0.01, 0.02}, ytick style={draw=none}, xtick style={draw=none}, legend style={nodes={scale=0.6}, font=\footnotesize},
 	xlabel= $\Delta^\prime$,
 	ylabel= RMSE,
    legend pos= south east,
 	xmin=0.00015625,
 	xmax=0.02,
 	ymin=0,
 	ymax=1, grid= major, grid style={dotted, gray}
 	]
\addplot [mark=square, color=teal]
table
[x expr=0.00015625*(2^\coordindex), 
y expr=\thisrowno{0}
] {RMQ/RMSE_RMQ_10_512.dat}; 	
 	
\addplot [mark=otimes*, color=ta2gray]
table
[x expr=0.00015625*(2^\coordindex), 
y expr=\thisrowno{0}
] {B_RDAQ/RMSE_B_RDAQ_10_512.dat}; 	

\addplot [mark=triangle, color=ta2orange]
table
[x expr=0.00015625*(2^\coordindex), 
y expr=\thisrowno{0}
] {RATQ/RMSE_RATQ_10_512.dat}; 

\legend{ RMQ\\ RDAQ\\ RATQ\\}
\end{axis}
\end{tikzpicture}
\caption{Comparison of RMQ, RDAQ, and RATQ at per coordinate precision of $10$ bits and $d=512$.}
\label{fig3}
\end{minipage}
\hspace{0.05\textwidth}
\begin{minipage}{0.45\textwidth}
\centering
\begin{tikzpicture}
\begin{axis}[xmode = log, ymode=log, ytick={0.000001, 0.00001, .0001, 0.001, 0.01, 0.1, 1, 10}, xtick = {0.00015625, 0.0003125, 0.000625, 0.00125, 0.0025, 0.005, 0.01, 0.02}, ytick style={draw=none}, xtick style={draw=none}, legend style={nodes={scale=0.6}, font=\footnotesize},
 	xlabel= $\Delta^\prime$,
 	ylabel= RMSE,
    legend pos= south east,
 	xmin=0.00015625,
 	xmax=0.02,
 	ymin=0,
 	ymax=1, grid= major, grid style={dotted, gray}
 	]
\addplot [mark=square, color=teal]
table
[x expr=0.00015625*(2^\coordindex), 
y expr=\thisrowno{0}
] {RMQ/RMSE_RMQ_10_1024.dat}; 	
 	
\addplot [mark=otimes*, color=ta2gray]
table
[x expr=0.00015625*(2^\coordindex), 
y expr=\thisrowno{0}
] {B_RDAQ/RMSE_B_RDAQ_10_1024.dat}; 	

\addplot [mark=triangle, color=ta2orange]
table
[x expr=0.00015625*(2^\coordindex), 
y expr=\thisrowno{0}
] {RATQ/RMSE_RATQ_10_1024.dat}; 

\legend{ RMQ\\ RDAQ\\ RATQ\\}
\end{axis}
\end{tikzpicture}
	\caption{Comparison of RMQ, RDAQ, and RATQ at per coordinate precision of $10$ bits and $d=1024$.}
\label{fig4}
\end{minipage}
\end{figure}
}

\section{Proofs}\label{s:proofs}

\subsection{Proof of Lemma \ref{l:main}}

For the estimator $\hat{\bar {x}}$ in~\eqref{e:estimate},
with $\hat{x}_i=Q_i(x_i, y_i)$, we have

\eq{
&\E{\left\|\frac{1}{n}\cdot\sum_{i \in [n]}Q_i(x_i,y_i)- \frac{1}{n}\cdot \sum_{i \in [n]} x_i \right\|_2^2}\\
&= \frac{1}{n^2} \cdot \sum_{i \in [n]} \E{\norm{Q_i(x_i,y_i)-x_i}_2^2} + \frac{1}{n^2} \cdot \sum_{i \neq j}  \E{\langle Q_i(x_i,y_i)-x_i, Q_j(x_j,y_j)-x_j\rangle}\\
&= \frac{1}{n^2} \cdot \sum_{i \in [n]} \E{\norm{Q_i(x_i,y_i)-x_i}_2^2} + \frac{1}{n^2} \cdot \sum_{i \neq j}  {\langle \E{Q_i(x_i,y_i)}-x_i, \E{Q_j(x_j,y_j)}-x_j\rangle}\\
&= \frac{1}{n^2} \cdot \sum_{i \in [n]} \E{\norm{Q_i(x_i,y_i)-x_i}_2^2}+ \left(\frac{1}{n} \cdot \sum_{i} \|\E{Q_i(x_i,y_i)}-x_i\|_2\right)^2
\\
&\hspace{2cm}
- \frac{1}{n^2} \cdot \sum_{i}\|\E{Q_i(x_i,y_i)}-x_i\|_2^2
\\
&\leq \frac{1}{n^2} \cdot \sum_{i \in [n]} \E{\norm{Q_i(x_i,y_i)-x_i}_2^2}+ \frac{(n-1)}{n^2} \cdot \sum_{i}\|\E{Q_i(x_i,y_i)}-x_i\|_2^2,
}
where the second identity uses the independence of $Q_i(x_i,y_i)$ for different $i$
and the final step uses Jensen's inequality.
The result follows by bound each term using 
the fact that $\mathbf{x}$ and $\mathbf{y}$ satisfy (2)
and the 
definitions of $\alpha(Q_i,\Delta_i)$ and $\beta(Q_i,\Delta_i)$, for $i\in[n]$.
\qed

\subsection{Proof of Lemma \ref{t:MQ}}
As mentioned in~\eqref{e:tildez_close},
 the integer $\tilde{z}$ found in Alg.~\ref{a:E_MCQ}
satisfies 
$\E{\tilde{z}\eps}=x$ and $|x-\tilde{z}\eps|< \eps$.
Therefore, 
it suffices to
show that the output of the quantizer satisfies $Q_{\tt M}(x,y)=\tilde{z}\eps$.

To see that $Q_{\tt M}(x,y)=\tilde{z}\eps$, 
denote 
the lattice used in decoding Alg.~\ref{a:D_MCQ}
as $\Z_{w,\eps}:=\{(zk+w )\cdot \eps: z \in \Z\}$. 
The decoding algorithm finds the point in $\Z_{w,\eps}$
that is closest to $y$. Note that
$w=\tilde{z}\mod k$, whereby $\tilde{z}\eps$ is a point in this lattice.
Further, 
for any other point $\lambda\neq \tilde{z}\eps$ in the lattice,  we must have 
\[
|\lambda-\tilde{z}\eps|\geq k\eps,
\]
and so, by triangular inequality, 
that 
\[
|\lambda -y|\geq |\lambda-\tilde{z}\eps|-|\tilde{z}\eps- y| \geq k\eps 
-|\tilde{z}\eps- y|.
\]
Thus, $\tilde{z}\eps$ is closer to $y$ than $\lambda$
if 
\begin{align}
k\eps> 2|\tilde{z}\eps- y|.
\label{e:closest}
\end{align}
Next, by using~\eqref{e:tildez_close} once again, we have
\[
|\tilde{z}\eps- y|\leq |\tilde{z}\eps- x| +|x- y|  < \eps+\Delta^\prime,
\]
which by condition~\eqref{e:parameter_condition} in the lemma 
implies that~\eqref{e:closest} holds. 
It follows that $|\lambda -y|>|\tilde{z}\eps- y|$ for every 
$\lambda\in \Z_{w,\eps}$, which shows that  
$Q_{\tt M}(x,y)=\tilde{z}\eps$ and completes the proof.
\qed
%%%%
\subsection{Proof of Lemma  \ref{t:RMQ}}
Recall from Remark~\ref{r:subg} that
for the random matrix $R$ given in~\eqref{e:R}, for every vector
$z\in \R^d$, the random variables $Rz(i)$, $i\in [d]$, are sub-Gaussian with variance 
parameter $\|z\|_2^2/d$. Furthermore, we need the following bound
for ``truncated moments'' of sub-Gaussian random variables.
\begin{lem}\label{l:variance_tail}
For a sub-Gaussian random $Z$ with variance factor $\sigma^2$
and every $t\geq 0$, we have 
\[
\E{Z^2\indic{\{|Z|>t\}}}\leq 
2(2\sigma^2+t^2)e^{-t^2/2\sigma^2}.
\]
\end{lem}
\begin{proof}
Note that for any nonnegative random variable $U$, 
it can be verified that
\[
\E{U\indic{\{U>x\}}}= xP(U>x) + \int_{x}^\infty P(U>u)\,du.
\]
Upon substituting $U=Z^2$ and $x=t^2$, along with
the fact that $Z$ is sub-Gaussian with variance parameter $\sigma^2$,
we get 
\begin{align*}
    \E{Z^2\indic{\{Z^2>t^2\}}}&= t^2P(Z^2>t^2) + 
\int_{t^2}^\infty P(Z^2>u)\,du
\\
&\leq 2t^2e^{-t^2/2\sigma^2}+ 
2\int_{t^2}^\infty e^{-u/2\sigma^2}\,du
\\
&\leq 2(t^2+2\sigma^2)e^{-t^2/2\sigma^2},
\end{align*}
which completes the proof.
%%\theerthaedit{}{
%%Let
%%\[
%%f(x) = - \int_{y \geq x} e^{-y^2/2} d y - x e^{-x^2/2} .
%%\]
%%Observe that
%%\[
%%f'(x) = e^{-x^2/2} - e^{-x^2/2} + x^2 e^{-x^2/2} = x^2 e^{-x^2/2}. 
%%\]
%%Hence,
%%\E{Z^2\indic{\{|Z|>t\}}} = \frac{t\sigma}{\sqrt{2\pi}} e^{-t^2/2\sigma^2} %%- \sigma^2 \text{erf}(t/\sigma) \leq \frac{t\sigma}{\sqrt{2\pi}} %%e^{-t^2/2\sigma^2} = \frac{\sigma^2}{2} %%\frac{\sqrt{2}t}{\sqrt{\pi}\sigma}e^{-t^2/2\sigma^2} \leq %%\frac{\sigma^2}{2} e^{-t^2/3\sigma^2},
%%\[
%%\]
%%where the last inequality follows by observing $e^{x^2/6} \geq %%\frac{x\sqrt{2}}{\sqrt{\pi}}$ for all $x$.
%%}
\end{proof}

We now handle the MSE $\alpha(Q)$ and bias $\beta(Q)$ separately below.
 \paragraph{Bound for MSE $\alpha(Q)$:}
 Denote by 
 $Q_{{\tt M}, R}(x, y)$ the final quantized value
 of the quantizer RMQ. For convenience, we abbreviate 
 \[
 \hat{x}_R :=
 R\, Q_{{\tt M}, R}(x, y). 
 \]
 Observe that $\hat{x}_R =\sum_{i \in [d]}
 Q_{\tt M}(Rx(i), Ry(i))e_i$, where   $ Q_{\tt M}$ is the MQ of Alg.~\ref{a:E_MCQ} and~\ref{a:D_MCQ} with parameters $k\geq$ and $\Delta^{\prime}$ set as in the statement of the lemma.  
 Since $R$ is a unitary transform, we have
\begin{align}
\E{\norm{Q_{{\tt M},R}(x,y)-x}_2^2}
&=\E{\norm{\hat{x}_{R}-Rx}_2^2}
\nonumber
\\
&=\sum_{i=1}^{d} \E{(\hat{x}_R(i) -Rx(i))^2}
\nonumber
\\
&= \sum_{i=1}^{d} \E{(\hat{x}_R(i)-Rx(i))^2 \indic{\{|R\left(x-y\right)(i)| \leq \Delta^{\prime}\}}}
\nonumber
\\
&\hspace{2cm}+ \sum_{i=1}^{d} \E{(\hat{x}_R(i)-Rx(i))^2 \indic{\{|R\left(x-y\right)(i)| \geq \Delta^{\prime}\}}}
\label{e:error_split}
\end{align}
We consider each error term on the right-side above separately. 
We can view the first term as the error corresponding to
MQ, when the input lies in its ``acceptance range.'' 
Specifically, under the event $\{|R\left(x-y\right)(i)| \leq
   \Delta^{\prime}$\}, we get
   by Lemma~\ref{t:MQ}
   that
   \[
   |\hat{x}_R(i)
   -Rx(i)| \leq \eps=
   \frac{2\Delta^{\prime}}{k-2}, \quad \text{{almost surely}},
   \]
   whereby 
  \begin{align}
  \sum_{i=1}^{d} \E{(\hat{x}_R(i)-Rx(i))^2 \indic{|R\left(x-y\right)(i)| \leq \Delta^{\prime}}}\leq d\,\eps^2.
  \label{e:error_term1}
  \end{align}
The second term on the right-side of~\eqref{e:error_split} corresponds
to the error due to ``overflow'' and is handled using concentration
bounds for the rotated vectors. Specifically, 
 we get
\begin{align}
\lefteqn{\sum_{i=1}^{d} \E{(\hat{x}_R(i)-Rx(i))^2 \indic{\{|R\left(x-y\right)(i)| \geq \Delta^{\prime}\}}}}
\nonumber
\\
&\leq 
2\sum_{i=1}^{d} \left[\E{(\hat{x}_R(i)-Ry(i))^2 \indic{\{|R\left(x-y\right)(i)| \geq \Delta^{\prime}\}}}
+\E{(Rx(i)-Ry(i))^2 \indic{\{|R\left(x-y\right)(i)| \geq \Delta^{\prime}\}}}
\right]
\nonumber 
\\
&\leq 2k^2\eps^2
\sum_{i=1}^{d}P(|R\left(x-y\right)(i)| \geq \Delta^{\prime})
+
2\sum_{i=1}^{d} 
\E{(Rx(i)-Ry(i))^2 \indic{\{|R\left(x-y\right)(i)| \geq \Delta^{\prime}\}}}
\nonumber
\\
&\leq 4dk^2\eps^2e^{-{d{\Delta^\prime}^2}/{2\Delta^2}}
+
2\sum_{i=1}^{d} 
\E{(Rx(i)-Ry(i))^2 \indic{\{|R\left(x-y\right)(i)| \geq \Delta^{\prime}\}}}
\nonumber
\\
&\leq 4dk^2\eps^2 e^{-{d{\Delta^\prime}^2}/{2\Delta^2}}+
4(2\Delta^2+d\Delta^{\prime 2})e^{-\frac{d{\Delta^\prime}^2}{2\Delta^2}},
\label{e:error_term2}
\end{align}
where the second inequality follows upon noting that from the description decoder of MQ in
Alg.~\ref{a:D_MCQ} that 
$|\hat{x}_R(i)-Ry(i)|\leq \eps k$ almost surely for
each $i\in [d]$; the third inequality uses
the fact that $R(x-y)(i)$ is sub-Gaussian with variance parameter
$\|x-y\|_2^2/d\leq \Delta^2/d$; and fourth inequality
is by Lemma~\ref{l:variance_tail}.

%\mnote{The split in the first  the term was $Ry(i) RX(i) $ instead of the correct term $Ry(i)$. Added the square to $\Delta^{prime}$, which was missing in the $e^{-}$ terms. $k$ did not have a square in third inequality above.}

Upon combining~\eqref{e:error_split},~\eqref{e:error_term1}, and~\eqref{e:error_term2}, and substituting
$\eps=2\Delta^{\prime}/(k-2)$ and 
${\Delta^\prime}^2=6(\Delta^2/d) \log \Delta/\delta$, 
we obtain
\begin{align}
\label{e:break}
    \E{\norm{Q_{{\tt M},R}(x,y)-x}_2^2}
    &\leq 
    d\,\eps^2+ 4dk^2\eps^2 e^{-\frac{d{\Delta^\prime}^2}{2\Delta^2}}+
4(2\Delta^2+d\Delta^{\prime 2})e^{-\frac{d{\Delta^\prime}^2}{2\Delta^2}}
\\
\nonumber
&= 24\, \frac{\Delta^2}{(k-2)^2}\ln \frac \Delta \delta 
+96\delta^2 \left(\frac{k}{k-2}\right)^2\cdot\frac{\ln (\Delta/\delta)}{(\Delta/\delta)}
+ 8\delta^2 \cdot \frac{1+3\ln (\Delta/\delta)}{(\Delta/\delta)}
\\
\nonumber
&\leq 24\,\frac{\Delta^2}{(k-2)^2}\ln \frac \Delta \delta + \left(\frac{96}{e} \left(\frac{k}{k-2}\right)^2+\frac{24}{e^{2/3}}\right)\cdot \delta^2,
\end{align}
where we used $(1+3\ln u)/u \leq 3/e^{2/3}$ 
and $(\ln u)/u\leq 1/e$
for every $u>0$. We conclude by noting that
for $k\geq 4$, 
\[
\left(\frac{96}{e}\left(\frac{k}{k-2}\right)^2+\frac{24}{e^{2/3}}\right)\leq 154.
\]

 \paragraph{Bias $\beta(Q)$:}
 The calculation for the bias is similar to that we 
 used to bound the second term on the right-side of
 \eqref{e:error_split}. Using the notation 
 $\hat{x}_R$ introduced above, 
 we have 
\eq{
\lefteqn{\norm{\E{Q_{{\tt M}, R}}-x}_2}
\\
&=\norm{\E{R^{-1}\left(\hat{x}_R-Rx\right)}}_2
\\
&= \norm{R\E{R^{-1}\left(\hat{x}_R-Rx\right)}}_2 \\ &=\norm{\E{RR^{-1}\left(\hat{x}_R-Rx\right)}}_2 \\ &=\norm{\E{\hat{x}_R-Rx}}_2,
}
where the second identity holds since $R$ is a unitary matrix.

Further, since $Q_{{\tt M}}(x,y)$ is an unbiased estimate
of $x$ when $|x-y|\leq \Delta^\prime$ (see Lemma~\ref{t:MQ}), 
  by~\eqref{e:error_term1} and~\eqref{e:error_term2} we obtain
\eq{
\norm{\E{\hat{x}_R-Rx}}_2^2 &\leq \sum_{i=1}^{d} \E{\left(\hat{x}_R(i)-Rx(i)\right)\indic{|R\left(x-y\right)_i| \geq \Delta^{\prime})} }^2\\
&\leq  \sum_{i=1}^{d}\E{\left(\hat{x}_R(i)-Rx(i)\right)^2\indic{|R\left(x-y\right)(i)| \geq \Delta^{\prime})} }\\
&\leq 154\, \delta^2,
}
which completes the proof.
\qed

\subsection{Proof of Lemma  \ref{t:RCS_RAQ_alpha_beta}}
 \paragraph{Mean Square Error $\alpha(Q_{S,R})$:}

From the description of Algorithms \ref{a:E_RCS_RMQ} and \ref{a:D_RCS_RMQ}, we know that the quantized output of subsampled RMQ $Q_{\tt WZ}$ for an input $x$ is
\eq{
&Q_{\tt WZ}(x) =   R^{-1}\hat{x}_R \text{, where}\\
&\hat{x}_R=\frac{1}{\mu}\sum_{i\in
  [d]}\left( Q_{\tt M}( Rx(i), Ry(i)) -Ry(i) \right) \indic{\{i \in S\}}\, e_i +Ry,}
  and $Q_{\tt M}( Rx(i), Ry(i)) $ denotes the quantized output of the modulo quantizer for an input $Rx(i)$ and side-information $Ry(i)$.
  Use the shorthand $Q(Rx(i))$ for $Q_{\tt M}( Rx(i), Ry(i))$, we have
  \begin{align} \nonumber
&\E{ \norm{Q_{\tt WZ}(x) -x}_2^2}\\ \nonumber &= \sum_{i \in [d]} \E{\left( \frac{1}{\mu}\left( Q(Rx(i)) -Ry(i) \right) \indic{\{i \in S\}}- (Rx(i)-Ry(i)) \right)^2 }\\ \nonumber
&\leq 2\sum_{i \in [d]} \E{ \frac{1}{\mu^2}\left( Q(Rx(i)) -Rx(i)\right)^2 \indic{\{i \in S\}}}
\\ \nonumber
&\hspace{2cm}+2\sum_{i \in [d]}\E{\left(\frac{1}{\mu}\left(Rx(i) -Ry(i)\right)\indic{\{i \in S\}}   - (Rx(i)-Ry(i)) \right)^2 } \\
&= \sum_{i \in [d]}\frac{2}{\mu} \E{\left( Q(Rx(i)) -Rx(i) \right)^2  }
+ 2\sum_{i \in [d]}
\E{\left(Rx(i) -Ry(i)\right)^2}\cdot
\E{\left(\frac{1}{\mu}\indic{\{i \in S\}}   - 1 \right)^2 }  \nonumber
\\ 
&= \sum_{i \in [d]}\frac{2}{\mu} \E{\left( Q(Rx(i)) -Rx(i) \right)^2  } + 2\sum_{i \in [d]}
\E{\left(Rx(i) -Ry(i)\right)^2}\cdot \frac{1-\mu}{\mu}  \label{e:intermed_subsampling}\\ \nonumber
&\leq \frac{2\alpha(Q_{{\tt M}, R})}{\mu}+ \frac{2\Delta^2}{\mu},
  \end{align}
  where we used the inequality: $(a+b)^2\leq 2(a^2+b^2)$, the independence of $S$ and $R$ in the second identity
  and used the fact that $R$ is unitary in the final step.

  \paragraph{ Bias $\beta(Q_{S, R})$:} 
 This follows upon noting that the conditional expectation (over $S$) of the output of
 subsampled RMQ given $R$ 
 is the vector $R^{-1} \sum_{i \in [d]}Q_{\tt M}(Rx(i), Ry(i))e_i$, which, in turn, is equivalent in distribution to the output of RMQ.
\qed

\subsection{Proof of Theorem \ref{t:lb_k}}
We denote $\Delta_{min} =  \min_{i \in [d]}\Delta_{i}$
and set $y_i$s to be $0$.
Let $x_1,...,x_n$ be an $iid$ sequence with common
distribution 
such that for all $j \in [d]$ we have 
\[
x_1(j) =\begin{cases}  \frac{\Delta_{min}}{\sqrt{d}}
\quad \text{w.p.} \frac{1+\alpha(j)\delta}{2}\\
- \frac{\Delta_{min}}{\sqrt{d}}
\quad \text{w.p.} \frac{1-\alpha(j)\delta}{2},
\end{cases}
\]
where $\alpha \in \{-1, 1\}^d$ is generated uniformly at random.
We have the following Lemma for such $x_i$s, which 
provides a 
lower bound for the MSE of any estimator of the mean of the distribution of $x_i$s.

\begin{lem}\label{l:lowbound}
{For $x_1, ..., x_n$ generated as above and}
 any estimator $\hat{\bar{x}}$ of the mean formed using only $r$-bit
quantized version of $x_i$s, we have\footnote{Note that the side information $y_i$s are all set to $0$.}
\[\E{ \left\|\hat{\bar{x}}- \frac{\delta \Delta_{min}}{\sqrt{d}} \alpha\right\|_2^2} \geq c^\prime \cdot   \frac{d\Delta_{min}^2}{nr},\]
where $c^\prime <1$ is a universal constant. 
\end{lem}
\noindent Proof of Lemma \ref{l:lowbound} follows from either \cite[Proposition 2]{duchi2014optimality} or \cite[Theorem 11]{acharya2020general}.

The proof of Theorem~\ref{t:lb_k} is completed by using this claim. Specifically, using $2a^2+2b^2 \geq (a+b)^2$, we have
\eq{
2\E{ \norm{\hat{\bar{x}}- \bar{x} }_2^2} + 2\E{\norm{ \bar{x}- \frac{\delta \Delta_{min}}{\sqrt{d}} \alpha}_2^2}  \geq \E{ \norm{\hat{\bar{x}}- \frac{\delta \Delta_{min}}{\sqrt{d}} \alpha}_2^2},
}
which, along with the observation that
\[\E{\norm{ \bar{x}- \frac{\delta \Delta_{min}}{\sqrt{d}} \alpha}_2^2} \leq \frac{\Delta_{min}^2}{n},
\]
gives
\eq{
\E{ \norm{\hat{\bar{x}}- \bar{x} }_2^2} \geq & \frac{c^\prime d\Delta^2_{min}}{2nr} -\frac{\Delta^2_{min}}{n}\\
\geq & \frac{c^\prime \Delta^2_{min}d}{4nr},
}
when $(d/r)\geq 4/c^\prime$. The proof is completed by setting
$c=c^\prime/4$.
\qed

\begin{rem}
Since the lower bound in \cite{acharya2020general} holds for sequentially interactive protocols, 
if we allow interactive protocols for mean estimation where client $i$ gets to see the messages transmitted by the clients
$j$ in $[i-1]$, 
%before sending its message, denoted by $\{\pi_j(x_j, U)\}_{j \in [i-1]}$ for a protocol $\pi$,
and can design its quantizers based on these previous messages, even then the lower bound above will hold. 
\end{rem}
%A slight difference from the proof of
%Proposition 2 in \cite{duchi2014optimality} is the fact that we have
%acess to side-information $y_i$s. Nevertheless, note that the mutual information between $\alpha$ and the final quantized values $\{Q_i(x)\}_{i \in [n]}$ can be bounded as follows $
%I(\alpha \wedge
%\{Q(x_i)\}_{i \in [n]})\leq I(V \wedge\ \{\Qenc(x_i)\}_{i \in [n]}, \{y_i\}_{i \in [n]}),$
%since $\alpha$ and $Q(x_i)$ are conditionally independent
%given
%$(\Qenc(x_i), y_i)$. Further,
% $I(V \wedge \Qenc(x_i), y_i)=I(V \wedge \Qenc(x_i))$ since $y_i$ is a
%constant. Thus, the same strong data processing
%inequality can be used as in \cite{duchi2014optimality}  to upper bound $I(\alpha \wedge
%\{Q(x_i)\}_{i \in [n]})$. Barring this
%fact, the proof of the claim is exactly the same as that of
%\cite[Proposition 2]{duchi2014optimality}.  

\subsection{Proof of Lemma  \ref{t:DAQ}}

We will prove a general result which will not only prove Lemma \ref{t:DAQ} but will also be useful in the proof of Lemma \ref{t:RDAQ}. Consider $x$ and $y$ in $\R^d$ such that each coordinate of both $x$ and $y$ lies in $[-M, M]$. Also, consider the following generalization of DAQ:
\[
Q_{\tt D}(x, y)  =  \sum_{i=1}^{d}2 M\left( \indic{\{U(i) \leq x(i)\}} - \indic{\{U(i) \leq y(i)\}}\right)e_i  +y, 
\]
where $\{U_i\}_{i \in [d]}$ are $iid$ uniform random variables in $[-M,M]$. We will show that 
\begin{align}\label{e:intermediate}
\E{Q_{\tt D}(x, y)}=x \quad \text{and} \quad  \E{\norm{Q_{\tt D}(x, y)-x}_2^2} \leq 2M \norm{x-y}_1,
\end{align}
which upon setting $M=1$ proves Lemma \ref{t:DAQ}.
 
Towards proving $\eqref{e:intermediate}$, note that
from the estimate formed by $Q_{\tt D}$, it is easy to see that $\E{Q_{\tt D}(x, y)}=x$.
The MSE can be bounded as follows:
\eq{\E{\norm{Q_{\tt D}(x, y)-x}_2^2}&=\sum_{i=1}^{d} \E{   (2M\left(
    \indic{\{U_i \leq x(i)\}}  -\indic{\{U_i \leq y(i)\}} \right)- \left(x(i) -y(i)\right))^2  }  
\\
&= \sum_{i=1}^{d}  4M^2 \frac{|x(i)-y(i)|}{2M}-\norm{x-y}_2^2 
\\
&=2 M\norm{x-y}_1 - \norm{x-y}_2^2,
}
where we used the observations that $ 2M \left( \indic{\{U_i \leq x(i)\}}  -\indic{\{U_i \leq y(i)\}} \right)$ is an unbiased estimate of $\left(x(i) -y(i)\right)$ and that  $\left( \indic{\{U_i \leq x(i)\}}  -\indic{\{U_i \leq y(i)\}} \right)^2$ equals one if and only if exactly one of the indicators is one, which in turn happens with probability $\frac{|x(i)-y(i)|}{2M}$.\qed.

\subsection{Proof of Lemma  \ref{t:RDAQ}}
\paragraph{Worst-case bias $\beta(Q_{{\tt D}, R} \Delta)$:}
Since the final interval $[-M_{h-1}, M_{h-1}]$ contains $[-1, 1]$, we can see that $\E{Q_{{\tt D}, R} (x, y)} = x$.

~\paragraph{Worst-case MSE $\alpha(Q_{{\tt D}, R} ; \Delta)$:}
We denote by  $B^{x}_{ij}$ and $B^{y}_{ij}$ the bits
\eq{
B^{x}_{ij} = \indic{\{U(i,j) \leq Rx(i)\}} ~\text{ and}~
B^{y}_{ij} = \indic{\{U(i, j) \leq Ry(i)\}}.
}
Then, the final quantized value of the quantizer RDAQ  can be expressed as
$Q_{{\tt D}, R}(X) =  R^{-1} \hat{x}_R$ 
where, with 
$z^*(i)$  denoting the smallest $M_j$ such that
the interval $[-M_j, M_j]$ contains $Rx(i)$ and $Ry(i)$ and $[h]_0=\{0, \ldots, h-1 \}$,
\begin{align*}
\nonumber \hat{x}_R :=\sum_{i \in \{1, \ldots, d\}} \left( \sum_{j \in [h]_0} 2M_j \cdot \left( B^{x}_{ij}  - B^{y}_{ij} \right)  +Ry(i) \right)\indic {\{z^*(i)=j\}}  e_i.
\end{align*}

Since $R$ is a unitary transform, 
we get

\eq{
\E{\norm{Q_{{\tt D}, R}(x)-x}_2^2}
&=
\E{\norm{RQ_{{\tt D}, R}(x)-Rx}_2^2}
\\
&=\E{\norm{\hat{x}_R-Rx}_2^2} 
\\
&=\sum_{i \in [d]} \E{ (\hat{x}_R(i) -Rx(i))^2  }\\
&=\sum_{i \in [d]} \E{ \left(  \sum_{j \in [h]_0}(2M_j \cdot \left( B^{x}_{ij}  - B^{y}_{ij} \right)  +Ry(i) -Rx(i))\indic {\{z^*(i)=j\}}\right)^2  }\\
&=\sum_{i \in [d]}\sum_{j \in [h]_0}   \E{ \left( 2M_j \left( B^{x}_{ij}  - B^{y}_{ij} \right) +Ry(i) -Rx(i) \right)^2  \indic {\{z^*(i)=j\}} ,}}
where the last identity uses $\indic{\{z^*(i)=j_1\}} \indic{\{z^*(i)=j_2\}}=0$ for all $j_1 \neq j_2, $ to cancel the cross-terms in the expansion of  $(\hat{x}_R(i) -Rx(i))^2$. Conditioning on $R$
and using the independence 
of $\indic{\{z^*(i)=j\}}$ from the randomness used in MQ, we get

\begin{align}\E{\norm{Q_{{\tt D}, R}(x)-x}_2^2}
\nonumber
&=\sum_{i \in [d]}\sum_{j \in [h]_0}   \E{ \E{\left( 2M_j \left( B^{x}_{ij}  - B^{y}_{ij} \right) +Ry(i) -Rx(i) \right)^2\mid R } \indic {\{z^*(i)=j\}} }\\
\nonumber
&\leq \sum_{i \in [d]}\sum_{j \in [h]_0}   \E{ 2M_j |Rx(i) -Ry(i)| \indic{\{z^*(i)=j\}} },\\
\nonumber
&\leq \sum_{i \in [d]} \E{ 2M_0 |Rx(i) -Ry(i)|\indic{\{z^*(i)=0\}} }\\& \nonumber ~~~+ \sum_{i \in [d]}\sum_{j \in [h-1]}   \E{ 2M_j |Rx(i) -Ry(i)|\indic{\{z^*(i)=j\}} },\\
\nonumber
&\leq \sum_{i \in [d]} \E{ 2M_0 |Rx(i) -Ry(i)| }\\& ~~~+ \sum_{i \in [d]}\sum_{j \in [h-1]}   \E{ 2M_j |Rx(i) -Ry(i)|\indic{\{z^*(i)=j\}} },
\label{e:overallterm}
\end{align}

where the  first inequality follows from \eqref{e:intermediate} in the proof of Lemma \ref{t:DAQ}.

Next, noting that 
\[
\indic{\{z^*(i)=j\}} \leq  \indic{\{|RX(i)| \geq
  M_{j-1}\}}+ \indic{\{|RY(i)| \geq M_{j-1}\}} \quad \text{ almost surely}, 
  \]
  an
application of the Cauchy-Schwarz inequality yields
\begin{align}
\nonumber &\E{ 2M_j |Rx(i) -Ry(i)|\indic{\{z^*(i)=j\}} }\\ 
\nonumber
&\leq 2M_j \E{(Rx(i)-Ry(i))^2}^{1/2} \E{(\indic{\{|RX(i)| \geq
  M_{j-1}\}}+ \indic{\{|RY(i)| \geq M_{j-1}\}})^2}^{1/2}\\
 \nonumber &\leq 2M_j \E{(Rx(i)-Ry(i))^2}^{1/2} \left(2P(|Rx(i) |\geq M_{j-1} )+ 2P(|Ry(i)| \geq M_{j-1} )\right)^{1/2}\\
  &\leq 2M_j \E{(Rx(i)-Ry(i))^2}^{1/2} \left(8 e^{\frac{-d M^2_{j-1}}{2}}\right)^{1/2},
  \label{e:term_2}
  \end{align}
where the second inequality uses $(a+b)^2 \leq 2a^2+2b^2$ and the third uses subgaussianity of $Rx(i)$ and $Ry(i)$.

Substituting the upper bound in \eqref{e:term_2} for the second term  in the RHS of \eqref{e:overallterm} and using $\E{X} \leq \E{X^2}^{1/2}$ for the first term, we get
\eq{
%\E{\norm{\hat{x}_R-Rx}_2^2}
\E{\norm{Q_{{\tt D}, R}(x)-x}_2^2}
&\leq \sum_{i \in [d]}  \E{  |Rx(i) -Ry(i)|^2}^{1/2} \left(2M_0+\sum_{j \in [h-1]} 2M_j \cdot \left(8e^{-\frac{dM_{j-1}^2}{2}}\right)^{1/2}\right)\\
&\leq  \sqrt{d \cdot \E{ \norm{Rx -Ry}_2^2}} \left(2M_0+\sum_{j \in [h-1]} 2M_j \cdot \left(8e^{-\frac{dM_{j-1}^2}{2}}\right)^{1/2}\right)\\
&= \sqrt{d \cdot \norm{x -y}_2^2} \left(2M_0+\sum_{j \in [h-1]} 2M_j \cdot \left(8e^{-\frac{dM_{j-1}^2}{2}}\right)^{1/2}\right)\\
&=\sqrt{d \cdot \norm{x -y}_2^2} \left(2\sqrt{\frac{6}{d}}+\sum_{j \in [h-1]} 2 \sqrt{\frac{6e^{*j}}{d}} \cdot \left(8e^{-1.5e^{*(j-1)}}\right)\right)\\
&= 8\sqrt{3} \cdot \sqrt{ \norm{x -y}_2^2} \left( 1 + \sum_{j\in[h-1]} e^{-0.5 e^*(j-1)} \right)\\
& \leq 16 \sqrt{3}\cdot \sqrt{ \norm{x -y}_2^2},
\\}
where the second inequality uses the fact that $\sum_{i} \norm{a}_1 \leq \sqrt{d}\norm{a}_2$, the first and second identities follow from the fact that $R$ is unitary transform and substituting for $M_i$s, the final inequality follows from the bound of $1$ for $\sum_{j=1}^{\infty}e^{-0.5 e^*(j-1)}$, which, in turn, can seen as follows
\eq{
e^{-0.5 e^*(j-1)} &={e^{-0.5}}+  {e^{-0.5e}}+ {e^{-0.5e^e}}
+\sum_{j=3}^{\infty}e^{-0.5{e^{*(j)}}}
\\
&\leq {e^{-0.5}}+  {e^{-0.5e}}+ {e^{-0.5e^e}}
+\sum_{j=3}^{\infty}e^{-0.5{je^e}}
\\
&\leq {e^{-0.5}}+  {e^{-0.5e}}+ {e^{-0.5e^e}}+\frac1{e^{e^e}-1}
\\
&\leq 1.
}
\qed

\subsection{Proof of Lemma \ref{t:sRDAQ}}\label{s:proof_RDAQ}

~\paragraph{Worst-case bias $\beta(Q_{{\tt WZ}, u}; \Delta)$:}
It is straightforward to see that $\E{Q_{{\tt WZ}, u}(x)} = x$.

~\paragraph{Worst-case MSE $\alpha(Q_{{\tt WZ}, u}; \Delta)$:}

We denote by  $B^{x}_{ij}$ and $B^{y}_{ij}$ the bits
\eq{
B^{x}_{ij} = \indic{\{U(i,j) \leq Rx(i)\}} ~\text{ and}~
B^{y}_{ij} = \indic{\{U(i, j) \leq Ry(i)\}}.
}
Then,
the quantized output can be stated as follows:
{noting that $Q_{{\tt WZ}, u}(x) =  R^{-1} \hat{x}_R$
where,  with 
$z^*(i)$  denoting the smallest $M_j$ such that
the interval $[-M_j, M_j]$ contains $Rx(i)$ and $Ry(i)$,}
\[
\hat{x}_R :=\left(\sum_{i \in \{1, \ldots, d\}}\sum_{j \in \{0, \ldots, h-1\}} 2M_j \cdot \left( B^{x}_{ij}  - B^{y}_{ij} \right) \indic {\{z^*(i)=j\}} \indic { \{ i \in S\}}\cdot e_i +Ry\right),
\]

Since $R$ is a unitary transform, 
the mean square error between
$Q_{{\tt WZ}, u}(x)$ and $x$ 
can be bounded as
in the proof of Lemma~\ref{t:RDAQ} as follows:
\eq{
\E{\norm{Q_{{\tt WZ}, u}(x)-x}_2^2}
&=\E{\norm{\hat{x}_R-Rx}_2^2}
\\
&=\E{\norm{\hat{x}_R-Rx}_2^2} 
\\
&=\sum_{i \in [d]} \E{ \hat{x}_R(i) -Rx(i))^2  }\\
&=\sum_{i \in [d]}\sum_{j \in [h]}   \E{ \left( 2M_j \left( B^{x}_{ij}  - B^{y}_{ij} \right)\indic { \{ i \in S\}} +Ry(i) -Rx(i) \right)^2  \indic {\{z^*(i)=j\}} }\\
&=\sum_{i \in [d]}\sum_{j \in [h]}   \E{ \E{\left( 2M_j \left( B^{x}_{ij}  - B^{y}_{ij} \right)\indic { \{ i \in S\}} +Ry(i) -Rx(i) \right)^2\mid R } \indic {\{z^*(i)=j\}} }\\
&\leq \sum_{i \in [d]}\sum_{j \in [h]}   \E{ \frac{2M_j}{\mu} \cdot |Rx(i) -Ry(i)|\cdot \indic{\{z^*(i)=j\}} },
}
where the inequality follows from similar calculations in the proof of Lemma \ref{t:DAQ}.
The rest of the analysis proceeds as that in the proof of Lemma \ref{t:RDAQ}.
\qed

\subsection{ Proof of Theorem \ref{t:lb_smooth_coom_constraints}  }\label{p:lb}
 \newest{
 Note that affine functions are $0$-smooth and admitted in the class of $L$ smooth functions. We use affine functions as difficult functions  and 
 follow the general recipe of \cite{ACMT21informationconstrained}, which in turn builds on \cite[Section 4.5]{acharya2021nips} and \cite{acharya2020general},  to show the lower bounds for convex, Lipschitz optimization under  communication constraints.
%proceed in the same manner as in  the lower bounds for convex, Lipschitz optimization under  communication constraints (\cite[Section 4.5]{acharya2021nips};
%see, also,~\cite{acharya2020general}), since the lower bounds for convex Lipschitz optimization also use affine functions as difficult functions. 
%We follow the general recipe of \cite{ACMT21informationconstrained} to prove our lower bounds.
The difficult functions we construct are the same as in
many existing lower bounds for convex functions such as \cite{agarwal2012information}. We consider the domain $\X=\{x \in \R^d:\norm{x}_\infty\leq D/(2\sqrt{d} \}$, and consider the following class of functions on $\X$: For $v \in \{-1, 1 \}^d$, let
\begin{equation}
  \label{eq:def:gv:convex:p2}
  f_{v}(x) \eqdef \frac{2\sigma\delta}{\sqrt{d}} \sum_{i=1}^{d} \mleft|x(i) - \frac{v(i) D}{2 \sqrt{d}}\mright|, \quad \forall\, x\in \X,
\end{equation}
and $x_v^\ast$ be its minimizer.
Note that the gradient $g_v(x)$ of $f_v$ at $x\in \X$ is equal to $-2\sigma\delta v /\sqrt{d}$, i.e., constant $\forall x$. 
%We will fix our noisy subgradient oracle $O_v$ later.
%For any gradient client $O_v$, let $\hat{g}_{t,i}$ denote the output of its $i$th copy in iteration $t$. 
For each $f_v$ in \eqref{eq:def:gv:convex:p2}, consider a sequence of $n$ clients $\clientset$ that output $d$-dimensional gradient vectors
$\{\hat{g}_i(x_t)\}_{i\in [n]}$, each of whose coordinates takes value
$-\sigma/\sqrt{d}$ or $\sigma/\sqrt{d}$ independently with probabilities $(1+2\delta v(i))/2$ and $(1-2\delta v(i))/2$, respectively. The parameter $\delta >0$ is to be chosen later.
Note that the above client construction satisfies the set of assumptions in \eqref{e:asmp_unbiasedness}, \eqref{e:varaince} and \eqref{e:bounded_est}.
%$\hat{g}_t$ are product Bernoulli distributed vectors with mean $-2B\delta v/\sqrt{d}$.
%We will
%consider $k$ copies of this noisy client which outputs $\{\hat{g}_{t,i}\}_{i=1}^k$ that are i.i.d. from a distribution $p_v$
%with mean $-\sigma\delta v /\sqrt{d}$.

Draw ${V \sim \mathtt{Unif}\{-1, 1\}^d}$. With respect to the associated random function $f_{V},$ each client $\tC_i$ chooses a quantizer ${Q_{i,t}}$ to generate output $Q_{i,t}(\hat{g}_i(x_t)).$
%uses encoder $\varphi_i$ of length $\ell$. Let $C_{t,i}=\varphi_i(g_{t,i})$, $1\leq t\leq T,  1\leq i\leq k$, which is transmitted over $k$ independent Gaussian quantizers to obtain outputs $\{Y_{t,i}\}_{i=1}^k$ given by \eqref{e:Gaussian_quantizer}.
%where $Z_{t,i}\sim \cN(0,\sigma_z^2/k)$ corresponds to  Gaussian noise for $\varphi_i$.
Denote by $Q^{nT}{=}(\{Q_{i,1}, \ldots, Q_{i,T}\}_{i\in [n]})$ the vector of quantized outputs observed at the server.
%Since $Y_i^T, 1\leq i\leq K,$ are conditionally independent given $\{V=v\},$
The following lower bound can be established by using results from%\footnote{Note that the result in \cite{ACMT21informationconstrained} is for a more general class of adaptive quantizers.} 
\cite[Lemma 3, 4]{ACMT21informationconstrained}:
\begin{equation}\label{e:C_P2}
  \E{f_{V}(\bar{x}_T)-f_{V}(x_V^\ast)} \geq  {\frac{D\sigma\delta}{3}}\Bigg[1 -  \sqrt{
\frac{2}{d}
      \sum_{j=1}^d \mutualinfo{V(j)}{Q^{nT}}}\Bigg].
\end{equation}
It remains to bound the mutual-information term for which one can use the independence across the clients and derive the following data-processing inequality based on the other techniques from \cite{ACMT21informationconstrained}:
\begin{align*}
\sum_{j=1}^d \mutualinfo{V(j)}{Q^{nT}}\leq 29nT\delta^2(d\wedge r),
\end{align*}
where ${\delta\in (0, 1/6)}.$
%where the parameter $\delta\in [0,1]$ is to be chosen appropriately.
Combining this with \eqref{e:C_P2} and setting $\delta=\sqrt{d/(232(d\wedge r)nT)},$ we finally get 
\[\E{f_{V}(\bar{x}_T)-f_{V}(x_V^\ast)}\geq \frac{1}{12\sqrt{58}}\frac{D\sigma}{\sqrt{nT}}\sqrt{\frac{d}{d\wedge r}},\]
where we need $T\geq d/(6nr)$ in order to enforce $\delta\leq 1/6.$  
The proof is completed by noting that
%\begin{align} \label{e:minmax_sub-optimality}
  $\cE^\ast(\X, \oO_{\tt sc}, T, \Q_r) \geq  \E{f_{V}(\bar{x}_T)-f_{V}(x_V^\ast)}.$

\subsection{Proof of Lemma \ref{l:general}}
Define $x^\ast=\text{argmin}_{x\in \X}f(x)$.
We have that 
\begin{align}\label{e:sum}
\E{f(x_{t+1})-x^\ast}=\E{f(x_{t+1})-f(x_t)} + \E{f(x_{t})-x^\ast}.
\end{align}
By smoothness, 
\begin{align*}
\E{f(x_{t+1})-f(x_t)|x_t}&\leq \nabla f(x_t)^\top\E{x_{t+1}-x_t|x_t}+\frac{L}{2}\E{\|x_{t+1}-x_t\|^2|x_t}\\
&=-\eta \nabla f(x_t)^\top\E{\M(\op^n_t)|x_t}+\frac{L\eta^2}{2}\E{\|\M(\op^n_t)\|^2|x_t}\\
&\leq -\eta \nabla f(x_t)^\top\E{\M(\op^n_t)|x_t}+\frac{\eta}{2}\E{\|\M(\op^n_t)\|^2|x_t}\\
&=-\frac{\eta}{2}\|\nabla f(x_t)\|^2+\frac{\eta}{2}\E{\|\M(\op^n_t)-\nabla f(x_t)\|^2|x_t},
\end{align*}
which further using the definition of $\alpha^\prime$ in $\eqref{def_a}$ and the law of total expectation imply
\begin{align}\label{e:c1}
\E{f(x_{t+1})-f(x_t)}\leq -\frac{\eta}{2}\E{\|\nabla f(x_t)\|^2}+\frac{\eta}{2}\alpha^{\prime 2}(\M).
\end{align}
By convexity, 
\begin{align*}
\E{f(x_{t})-x^\ast}
&\leq \E{\nabla f(x_t)^\top(x_t-x^\ast)}\\
%&= \frac{1}{2\eta}\E{\|\eta^2\nabla f(x_t)\|^2+\|x_t-x^\ast\|^2-\|x_t-\eta\nabla f(x_t)-x^\ast\|^2}\\
%&= \frac{\eta}{2}\E{\nabla f(x_t)\|^2}+\frac{1}{2\eta}\E{\|x_t-x^\ast\|^2-\|(x_{t+1}-x^\ast)+\eta(\M(\op^n_t)-\nabla f(x_t))\|^2}\\
%&=\frac{\eta}{2}\E{\nabla f(x_t)\|^2}+\frac{1}{2\eta}\E{\|x_t-x^\ast\|^2-\|x_{t+1}-x^\ast\|^2}-\E{\|\eta(\M(\op^n_t)-\nabla f(x_t))\|^2}-\E{(x_{t+1}-x^\ast)^\top(\M(\op^n_t)-\nabla f(x_t))}\\
&=\E{(\nabla f(x_t)-\M(\op^n_t))^\top (x_t-x^\ast)}+\E{\M(\op^n_t)^\top(x_t-x^\ast)}\\
&=\E{(\nabla f(x_t)-\M(\op^n_t))^\top (x_t-x^\ast)}+\frac{1}{2\eta}\E{\eta^2\|\M(\op^n_t)\|^2}\\
& \qquad \qquad +\E{\|x_t-x^\ast\|^2-\|x_{t}-\eta \M(\op^n_t)-x^\ast\|^2}\\
&\leq\E{(\nabla f(x_t)-\M(\op^n_t))^\top (x_t-x^\ast)}+\frac{1}{2\eta}\E{\eta_t^2\|\M(\op^n_t)\|^2}\\
&\qquad \qquad +\E{\|x_t-x^\ast\|^2-\|\Gamma_{\X}(x_{t}-\eta \M(\op^n_t))-x^\ast\|^2}\\
%&=\E{(\nabla f(x_t)-\M(\op^n_t))^\top (x_t-x^\ast)}+\frac{\eta}{2}\E{\|\M(\op^n_t)-\nabla f(x_t)\|^2+\|\nabla f(x_t)\|^2+2(\M(\op^n_t)-\nabla f(x_t))^\top \nabla f(x_t)}\\
%&\qquad \qquad \qquad +\frac{1}{2\eta}\E{\|x_t-x^\ast\|^2-\|x_{t+1}-x^\ast\|^2}\\
&\leq \beta^\prime(\M)\cdot D+\frac{\eta}{2}\alpha^{\prime 2}(\M)+\frac{\eta}{2}\E{\|\nabla f(x_t)\|^2}+\eta B\cdot \beta^\prime(\M)\\
&\qquad \qquad+\frac{1}{2\eta }\E{\|x_t-x^\ast\|^2-\|x_{t+1}-x^\ast\|^2}, \numberthis \label{e:c2}
\end{align*}
where second inequality is due to a well known property of the projection operator $\Gamma_{\X}$ (see, for instance, Lemma 3.1, \cite{bubeck2015convex}),  third inequality follows from Cauchy-Schwarz inequality and using the definitions in \eqref{def_a} and \eqref{def_b}.
Plugging \eqref{e:c1} and \eqref{e:c2} in \eqref{e:sum}, we have
\begin{align*}
\E{f(x_{t+1})-x^\ast}
&\leq  \beta^\prime(\M)\cdot (D+\eta B)+\eta\cdot \alpha^{\prime 2}(\M)+
\frac{1}{2\eta}\E{\|x_t-x^\ast\|^2-\|x_{t+1}-x^\ast\|^2}.
\end{align*}
Summing from $t=0$ to $T-1$, dividing by $T$,  using 
the assumption that the domain $\X$ has diameter at most $D$,  and setting $\eta$ as provided,
% we have 
%\begin{align*}
%\E{\frac{1}{T}\sum_{i=0}^{t-1}f(x_{t+1})-x^\ast}
%\leq D\cdot\left(\frac{\sqrt{2}\alpha(Q)}{\sqrt{T}}+\beta(Q)\cdot\left(1+\frac{1}{\alpha(Q)\sqrt{2T}}\right)\right)+\frac{LD^2}{2T}.
%\end{align*}
the proof is completed.
%by using Jensen's inequality and taking supremum over $(f,\bar{O})\in \cO$ on the left-side of the above inequality.
This general convergence bound will be used in our upper bound proofs below.
}

\subsection{Proof of Theorem \ref{t:ParSGD}}
From \cite[Theorem 3.7]{mayekar2020ratq}, we use the following result.
\begin{lem}\label{l:subsampled_RATQ}
Let $Q_{\mathtt{RATQ}}$ be the subsampled version of RATQ  using $r \geq 3+ \ceil{\log(1+ \ln^\ast(d/3))}$ bits. Then for $Y$ such that $\norm{Y}_2 \leq B^2$, we have
\begin{align*}
\E{Q_{\mathtt{RATQ}}(Y)\mid Y}&=Y\qquad \text{  and   }\qquad  \E{\norm{Q_{\mathtt{RATQ}}(Y)-Y}^2} \leq \frac{dB^2}{\frac{r}{3+ \ceil{\log(1+ \ln^\ast(d/3)}}) -1}.
\end{align*}
\end{lem}
%From \cite[Theorem 6.3]{bubeck2015convex}, we have
%\[\EE{\bar{x}_T}{F(\bar{x}_T)}-\min_{x \in \X}F(x) \leq  R\sigma\sqrt{\frac{2}{KT}}+ \frac{\beta R^2}{T},\]
%where $\sigma$ is the variance of the gradient estimate  $\sum_{i \in K} Q(\hat{g}_i(x_t))/K$. 
\noindent Further, for ${t \in [T]}$, we have ${\M(\op^n_t)=\frac{1}{n}\sum\limits_{i\in [n]} Q_{\mathtt{RATQ}}(\hat{g}_i(x_t))}$ as in \eqref{e:Q_ParSGD}. Thus we have,
\eq{ \alpha^{\prime 2}(\M)
&\leq  \E{\left\|\frac{1}{n} \sum_{i \in [n]} Q_{\mathtt{RATQ}}(\hat{g}_i(x_t)) -\frac{1}{n} \sum_{i \in [n]} \hat{g}_i(x_t)\right\|^2}+\E{\left\| \frac{1}{n}\sum_{i \in [n]} \hat{g}_i(x_t)- \nabla f(x_t)\right\|^2}.
}
Since $\M(\op^n_t)$ is  an unbiased estimate,  $\beta^{\prime}(\M)=0$. The proof is completed by bounding the two terms in the right-side above followed by using Lemma \ref{l:general}, which we do as follows.
From Lemma \ref{l:subsampled_RATQ}, it follows that for $t \in [T]$, we have
\begin{align*}
&\E{\left\|\frac{1}{n} \sum_{i \in [n]} Q_{\mathtt{RATQ}}(\hat{g}_i(x_t))-\frac{1}{n} \sum_{i \in [n]} \hat{g}_i(x_t)\right\|^2} \leq \frac{ dB^2}{n\left(\frac{r}{ 3+ \ceil{\log(1+ \ln^\ast(d/3))}} -1\right)}. 
\end{align*}
From \eqref{e:var}, we have
\[\E{\left\| \frac{1}{n}\sum_{i \in [n]} \hat{g}_i(x_t)- \nabla f(x_t)\right\|^2} \leq \frac{\sigma^2}{n}.\]

\subsection{Proof of Theorem \ref{t:conv_WZ_SGD}}\label{s:scheme}

\paragraph{Subgaussian and subexponential norms.}
For our analysis, it will be convenient to recall the definition of subgaussian \footnote{$\norm{\cdot}_{\psi_2}$ is indeed a norm.}  and subexponential norms of a random variable. %the subexponential norm of a random variable.
\begin{defn}[{\cite{vershynin2018high}}]\label{d:sg_rv}
A subgaussian norm of a subgaussian random variable $X$, denoted $\norm{X}_{\psi_2},$ is defined as $\norm{X}_{\psi_2} := \inf\{t>0: \E{e^{X^2/t^2}}\leq 2  \}.$ It follows that for a centered subgaussian random variable $X,$ 
$\Pr(|X|)\geq t) \leq 2 e^{-\frac{t^2 }{\norm{X}_{\psi_2}^2}}.$
\end{defn}
\begin{defn}[{\cite[Def. 2.7.5]{vershynin2018high}}]
A subexponential norm of a  subexponential random variable $X$, denoted $\norm{X}_{\psi_1},$ is defined as $\norm{X}_{\psi_1} := \inf\{t>0: \E{e^{\frac{|X|}{t}}}\leq 2  \}.$ It follows that for a centered subexponential  random variable $X,$ $\Pr(|X|)\geq t) \leq 2 e^{-\frac{t }{\norm{X}_{\psi_1}}}.$
\end{defn}

\paragraph{Side information is close to gradient estimates.}
We begin by noting that side-information $Y$\footnote{For convenience, we drop the iteration subscript $t$ in this Section.} is close to the stochastic gradient estimates computed by clients in $\clientset_2$.
Specifically, setting the parameters as
$\log \ell_1= \ceil{\log \frac{2B}{\sigma}+1}$ and  $r_1=r/\ceil{\log \frac{2B}{\sigma}+1} $
for clients in $\clientset_1$, we get the following.
\begin{lem}\label{l:conc_SI_grad}
For all $x \in \R^d$,  $j \in \clientset_2, i \in [d],$ and a universal constant $c_3>0,$ we have
\begin{align*}
 \Pr(|{R}\hat{g}_j(x)(i)-{R}Y(i)|\geq t   ) &\leq   2 e^{- c_3\min \lbrace \frac{t^2}{\sigma^{\prime 2}}, \frac{t\sqrt{d}}{\sigma^{\prime}}\rbrace}{+}2e^{- c_3\frac{t^2d}{ \sigma^{ 2}}},
\end{align*}
where $R$ is a {random Hadamard matrix \eqref{e:R} and for another universal constant $c_4>0$,}
\begin{align}\label{e:sigma_prime}
\sigma^{\prime 2}=\frac{c_48d\sigma^2 \ceil{\log (2B/\sigma+1)}}{nr}.
\end{align}
\end{lem}

\begin{rem}
  In the analysis for RMQ presented in Section \ref{s:RMQ}, the difference between the coordinates of the rotated input and rotated side information had subgaussian tails. However, note that in Lemma \ref{l:conc_SI_grad}, we can only prove a slightly weaker concentration result.
% which is a key differentiator between the analysis of RMQ in both these works. 
\end{rem}

Towards proving Lemma \ref{l:conc_SI_grad}, we begin by showing
the following result which holds from the subgaussian properties of
uniform quantizer error and 
standard properties of subgaussian random variables.
\begin{lem} \label{l:SGnormSI}
For all $x \in \R^d$ and  $i \in [d]$ we have\[\displaystyle{\hspace{2.5cm}\norm{Y(i) - \nabla f (x)(i)}_{\psi_2}^2\leq \sigma^{\prime 2}.}\]
\end{lem}
 \begin{proof}
 We will prove the theorem for $Y(1)$ since the argument remains the same for all $Y(i)$s. 
 From the description of CUQ, we note that $Q_{\tt u}( \hat{g}_{j}(x)(1))$ satisfies 
 \[\displaystyle{\norm{Q_{\tt u}( \hat{g}_{j}(x)(1)) - \hat{g}_{j}(x)(1)}_{\psi_2}^2 \leq \frac{4c_4B^2}{(\ell_1-1)^2}, \quad \forall j \in S_1},\]
for some universal constant $c_4> 0.$ 
  Also, from \eqref{e:varaince}, we have
 $\displaystyle{\norm{\hat{g}_{j}(x)(1) -\nabla f (x)(1)}_{\psi_2}^2 \leq  c_4\sigma^2}$ for the same constant $c_4$ above.  Further, using the triangle inequality and the fact that $(a+b)^2 \leq 2a^2+2b^2$,  \[\displaystyle{\norm{Q_{\tt u}( \hat{g}_{j}(x)(1)) - \nabla f (x)(i)}_{\psi_2}^2 \leq   \frac{c_48B^2}{(\ell_1-1)^2}+ 2c_4\sigma^2.}\]
 The proof is completed upon noting that the average of $N$ $iid$  zero mean subgaussian random variables $\{X_i\}_{i \in [N]}$ has a subgaussian norm square equal to $\norm{X_1}_{\psi_2}^2/N$  and the fact that we use $N=nr_1/(2d)$ samples to form $Y(1).$
 \end{proof}

\begin{rem}
  \newest{In order to quantize a $d$-dimensional gradient to $r \leq d$ bits, the technique of uniform sampling has been used in recent papers on distributed optimization ($cf.$ \cite{suresh2017distributed}, \cite{mayekar2020ratq}).  However, notice that these works merely required the quantized gradient estimate to be close to the true gradient in mean square sense. In our case, in order to leverage our Wyner-Ziv compression algorithms,  we need side-information to be close to the true gradient in s much stronger sense. Therefore, we refrain from using uniform sampling and instead use the clients to quantize separate, smaller blocks of coordinates.}
\end{rem}

%\paragraph{Rotation of Side-information is very close to the rotation of true gradient}
%We will now show that the rotated side information vector  is close to the rotated gradient vector.
\paragraph{Rotation of side-information is  close to the rotation of true gradient.}
Using standard properties of subgaussian random variables
(see~\cite[Lemma 2.7.7 and Theorem 2.8.1]{vershynin2018high}), we can show the following.
\begin{lem}\label{l:conc_SI_mean} For all $x \in \R^d$ and $i \in [d]$ we have for a universal constant $c_5>0$\\
\[{\hspace{0.3cm}\Pr(|{R}Y(i)-{R}\nabla f(x)(i)|\geq t) \leq 2 e^{\left(- c_5\min \lbrace t^2/ \sigma^{\prime 2}, t\sqrt{d}/\sigma^{\prime} \rbrace \right)}.}\]
\end{lem}
 \begin{proof} 
   The proof follows from combining two facts. First, note that for a sequence $\{X_i\}_{i \in [N]}$ of zero mean, $iid$ subexponential random variables, we have from \cite[Theorem 2.8.1]{vershynin2018high}
 \[\Pr\bigg\{\bigg|\sum_{i=1}^N X_i\bigg| \geq t\bigg\} \leq 2 e^{- c_5\min \lbrace \frac{t^2}{N\norm{X_1}_{\psi_1}^2}, \frac{t}{ \norm{X_1}_{\psi_1}} \rbrace},\] for some universal constant $c_5>0.$

 Also, note that the product of two subgaussian random variables $X$ and $Y$ is subexponential random variable with subexponential norm bounded as follows (see \cite[Lemma 2.7.7]{vershynin2018high}):
 $  \norm{XY}_{\psi_1} \leq \norm{X}_{\psi_2}\norm{Y}_{\psi_2}.$
 Notice that $|{R}Y(i)-{R}\nabla f(x)(i)| =\sum_{i \in [d]} U(i) V(i)$, where $U(i)$, $V(i)$ are zero mean, $iid$, subgaussian random variables with subgaussian norms as $1/\sqrt{d}$ and $\sigma^{\prime}$ ($cf.$ Lemma \ref{l:SGnormSI}). 
% Using the bound on subgaussian norm in Lemma \ref{l:SGnormSI}, the proof is completed.
 \end{proof}
Finally,
proceeding in the same manner as in \cite[Lemma 5.8]{mayekar2020ratq},
\newer{we can show that the coordinates of the rotated stochastic gradient are close to the coordinates of the rotated true gradient.} 
\begin{lem} \label{l:conc_grad_mean}
For all $x \in \R^d$ and the universal constant $c_4>0$ as in Lemma \ref{l:SGnormSI}, we have
\[{\hspace{2cm}\norm{{R}\hat{g}_j(x)(i) -{R}\nabla f(x)(i)}_{\psi_2}^2 \leq c_1 \sigma^2/d.}\]
\end{lem}
Thus,  random rotation allows us to convert the $\ell_2$ norm bound in assumption \eqref{e:varaince} to a $\ell_{\infty}$ bound.

\noindent \newest{We now choose $c_3=\min \{c_4,c_5\}.$ Using the inequality: $\max(a,b)\leq a+b$ and the property of subgaussian random variable in Definition \ref{d:sg_rv},  Lemma \ref{l:conc_SI_grad} follows from combining
Lemmas \ref{l:conc_SI_mean} and \ref{l:conc_grad_mean}.}
%and using the fact that $\Pr(|X+Y|\geq t)\leq \Pr(|X|\geq t/2)+\Pr(|Y|\geq t/2),$ we get Lemma \ref{l:conc_SI_grad}
\newest{Next, we present a lemma similar to Lemma \ref{l:variance_tail} towards evaluating bounds on $\alpha^{\prime}(\M)$ and $\beta^{\prime}(\M)$.}
\begin{lem}\label{l:variance_tail2}
For a random variable $Z$ such that 
\[  \Pr( |Z| \geq t)  \leq   2 e^{- \frac{c_3t^2}{\sigma^{\prime 2}}} +2e^{-\frac{c_3t\sqrt{d}}{\sigma^{\prime}}} +2e^{- \frac{c_3 t^2d}{ \sigma^2}},\] where $c_3>0$ is some universal constant, we have
\[
\E{Z^2\indic{\{|Z|>t\}}}\leq 
 2\left(\frac{\sigma^{\prime 2}}{c_3}+t^2\right)e^{- \frac{c_3 t^2}{\sigma^{\prime 2}}}+ 2\left(\frac{\sigma^{ 2}}{dc_3}+t^2\right)e^{- \frac{c_3 t^2 d}{\sigma^{2}}} +2\left(\frac{2\sigma^{\prime 2}}{dc_3^2}+\frac{2 \sigma^{\prime}t}{c_3\sqrt{d}} +t^2\right)e^{-\frac{c_3t\sqrt{d}}{\sigma^{\prime}}}.
\]
\end{lem}

\begin{proof}

For any nonnegative random variable $U$, 
it can be seen
\[
\E{U\indic{\{U>x\}}}= x\Pr(U>x) + \int_{x}^\infty \Pr(U>u)\,du.
\]
Upon substituting $U=Z^2$ and $x=t^2$, along with
the fact that $Z$ has the tail behaviour described above,
we get 
\begin{align*}
    \E{Z^2\indic{\{Z^2>t^2\}}}&= t^2\Pr(Z^2>t^2) + 
\int_{t^2}^\infty \Pr(Z^2>u)\,du
\\
&\leq 2t^2 \left(e^{- \frac{c_3 t^2}{\sigma^{\prime 2}}} +e^{-\frac{c_3t\sqrt{d}}{\sigma^{\prime}}} +e^{- \frac{c_3 t^2 d}{\sigma^{2}}}\right) 
+ 2\int_{t^2}^\infty e^{-\frac{c_3u}{\sigma^{\prime 2}}}\,du  + 
2\int_{t^2}^\infty e^{-\frac{c_3ud}{\sigma^2}}\,du \\
&\qquad +2\int_{t^2}^\infty e^{-\frac{c_3\sqrt{ud}}{\sigma^\prime}}\,du  \\
&\leq 2\left(\frac{\sigma^{\prime 2}}{c_3}+t^2\right)e^{- \frac{c_3 t^2}{\sigma^{\prime 2}}}+ 2\left(\frac{\sigma^{ 2}}{dc_3}+t^2\right)e^{- \frac{c_3 t^2 d}{\sigma^{2}}} +2\left(\frac{2\sigma^{\prime 2}}{dc_3^2}+\frac{2 \sigma^{\prime}t}{c_3\sqrt{d}} +t^2\right)e^{-\frac{c_3 t\sqrt{d}}{\sigma^{\prime}}}.
\end{align*}
\end{proof}

\paragraph{Bounds on $\alpha^\prime(\M)$ and $\beta^\prime(\M)$} 

%Let the grid size $\eps$ of modulo quantizer be set as follows: $\eps=\frac{2\Delta^{\prime}}{\ell_2-2}$, where $\Delta^{\prime}=\frac{3\sigma}{\sqrt{cd}} \cdot\log (\frac{2 \sigma}{\sqrt{c} \delta})$ {for some parameter $\delta$ to be specified later.}
%\begin{rem}\label{r3}
%Note that such a choice of $\eps$ ensures that  whenever a coordinate of the rotated vector $R\hat{g}_j(x)$ is within  $\Delta^{\prime}$ of the corresponding coordinate of the rotated side information there is no error in decoding ($cf. $ Lemma \ref{t:MQ}). Therefore, the output of modulo quantizer is unbiased and $\eps$ close to input under this event. Also, note that because of the minimum distance decoding at the decoder, each  coordinate decoded by modulo quantizer is always $\ell \eps$ close to the side information.
%\end{rem}

\newest{Recall that $Q_{{\tt RM}, j}$ denotes the rotated modulo quantizer without any subsampling {for client $j\in \clientset_2$}. From the description of RMQ in Algorithm \ref{a:D_RMQ}, we have
\vspace{-0.25cm}
\eq{ Q_{{\tt M}, R_j, j}( \hat{g}_j(x), Y){=}R_j^{-1}\left(\sum_{i\in [d]}Q_{\tt M}(R_j\hat{g}_j(x)(i), R_jY(i)) \cdot e_i\right).}
}
The key step of the proof is bounding MSE and bias of RMQ. Towards that, we have the following lemma.

\newest{\begin{rem}
The calculation for MSE and bias are different in  the proof of Lemma \ref{l:conc_SI_grad} compared to those in Proof of Lemma \ref{t:RMQ}. This is because of the weaker concentration results available in this case.
\end{rem}
}

\begin{lem} Under the condition that $nr \geq c_4 d^2 \log (B/\sigma)$, we have for all $x \in \R^d$, $j \in \clientset_2$, and for some parameter $\delta \in (0, \sigma/\sqrt{c_3} )$ that
\eq{
& \E{\norm{ Q_{{\tt M}, R_j, j}( \hat{g}_j(x), Y)- \hat{g}_j(x)}_2^2}\leq  \frac{36\sigma^2}{c_3(\ell_2-2)^2}\left(\ln \frac{\sigma}{ \sqrt{c_3}\delta}\right)^2+237\delta^2, \\
&\norm{\E{ Q_{{\tt M}, R_j, j}( \hat{g}_j(x), Y)}-\hat{g}_j(x)}_2^2  \leq 237\delta^2,
}
where $c_3$ and $c_4$ are the universal constants same as in Lemma \ref{l:conc_SI_grad}.
\end{lem}
\begin{proof}
By considering events $\{|R_j(\hat{g}_j(x)-Y)(i)| \leq \Delta^{\prime}\}$ and $\{|R_j(\hat{g}_j(x)-Y)(i)| \geq \Delta^{\prime}\}$, and then using the facts for modulo quantizer, we have
\eq{
\E{\norm{ Q_{{\tt M}, R_j, j}( \hat{g}_j(x), Y)- \hat{g}_j(x)}_2^2}
&\leq d\eps^2+ \sum_{i=1}^{d}\E{{( Q_{{\tt M}, R_j, j}( \hat{g}_j(x), Y)-\hat{g}_j(x))(i)}^2\indic{\{|R_j(\hat{g}_j(x)-Y)(i)| \geq \Delta^{\prime}\}}}\\
&\leq d\eps^2+2\ell_2^2\eps^2\sum_{i=1}^{d}\Pr( |R_j(\hat{g}_j(x)-Y)(i)| \geq \Delta^{\prime})\\&+2\sum_{i=1}^{d}\E{(R_j(\hat{g}_j(x)-Y)(i))^2\indic{\{|R_j(\hat{g}_j(x)-Z)(i)| \geq \Delta^{\prime}\}}}\\
&\leq d\eps^2+4\ell_2^2\eps^2 d\left(e^{- \frac{c_3\Delta^{\prime 2}}{\sigma^{\prime 2}}} +e^{-\frac{c\Delta^\prime \sqrt{d}}{\sigma^{\prime}}} +e^{- \frac{c_3 \Delta^{\prime 2}d}{ \sigma^2}}\right)
\\&+ 4d\left(\frac{\sigma^{\prime 2}}{c_3}+\Delta^{\prime 2}\right)e^{- \frac{c_3 \Delta^{\prime 2}}{\sigma^{\prime 2}}}\\
& + 4d\left(\frac{\sigma^{ 2}}{dc_3}+\Delta^{\prime 2}\right)e^{- \frac{c_3 \Delta^{\prime 2} d}{\sigma^{2}}} +4d\left(\frac{2\sigma^{\prime 2}}{dc_3^2}+\frac{2 \sigma^{\prime}\Delta^{\prime}}{c_3\sqrt{d}} +\Delta^{\prime 2}\right)e^{-\frac{c_3\Delta^{\prime}\sqrt{d}}{\sigma^{\prime}}}.\numberthis \label{e:c}
}
\newer{For some parameter $\delta\in (0, \sqrt{d}\Delta)$,} we substitute the other parameters as
$\eps=2\Delta^{\prime}/(\ell_2-2)$ and 
${\Delta^\prime}^2=9\Delta^2 (\ln (\sqrt{d} \Delta/\delta))^2$,  where $\Delta^2 = \max\{ \frac{\sigma^{\prime 2}}{c_3},  \frac{\sigma^{\prime 2}}{dc_3^2}, \frac{\sigma^{2}}{dc_3}\}$. That gives
\newer{\begin{align*}
%2\ell_2^2\eps^2\sum_{i=1}^{d}\Pr( |R_j(\hat{g}_j(x)-Z)(i)| \geq \Delta^{\prime})
4\ell_2^2\eps^2 d\left(e^{- \frac{c_3\Delta^{\prime 2}}{\sigma^{\prime 2}}} +e^{-\frac{c_3\Delta^\prime \sqrt{d}}{\sigma^{\prime}}} +e^{- \frac{c_3 \Delta^{\prime 2}d}{ \sigma^2}}\right)
&=
\left(\frac{12\ell_2}{\ell_2-2}\right)^2d\Delta^2\left(\ln \frac{\sqrt{d}\Delta}{\delta}\right)^2\left(\frac{2}{(\sqrt{d}\Delta/\delta)^9 } +\frac{1}{(\sqrt{d}\Delta/\delta)^3}\right )\\
&\leq \left(\frac{12\ell_2}{\ell_2-2}\right)^2d\Delta^2\left(\ln \frac{\sqrt{d}\Delta}{\delta}\right)^2\left(\frac{3}{(\sqrt{d}\Delta/\delta)^3 }\right) \\
&= \left(\frac{12\ell_2}{\ell_2-2}\right)^2\frac{\left(\ln \frac{\sqrt{d}\Delta}{\delta}\right)^2}{\sqrt{d}\Delta/\delta }\delta^2 \\
&\leq \left(\frac{24\ell_2}{e(\ell_2-2)}\right)^2\delta^2 \\
&\leq 139\delta^2, \numberthis \label{e:c1}
\end{align*}}
%\eq{
%&4dk^2\eps^2\left(\exp \left(- c {\Delta^{\prime 2}}/ 4\sigma^{\prime 2}\right) +\exp\left(-c {\Delta^\prime}\sqrt{d}/2\sigma^{\prime} \right) +\exp \left(- c {\Delta^{\prime 2}}d/ 4\sigma^{ 2} \right)\right) \\&= 144 \left(\frac{k}{k-2}\right)^2 d\Delta^{ 2}(\ln \frac{\sqrt{d}\Delta}{\delta})^2\cdot \left(\frac{2}{(\sqrt{d}\Delta/\delta)^9 } +\frac{1}{(\sqrt{d}\Delta/\delta)^3}\right ) 
%\\&\leq  144 \left(\frac{k}{k-2}\right)^2 d\Delta^{ 2}\cdot \left(\frac{2(\sqrt{d}\Delta/\delta e)^2}{(\sqrt{d}\Delta/\delta)^9 } +\frac{4 (\sqrt{d}\Delta/\delta e^2)}{(\sqrt{d}\Delta/\delta)^3}\right ) 
%\\&= 144 \left(\frac{k}{k-2}\right)^2 d\Delta^{ 2}\cdot \left( \frac{2}{e^2(\sqrt{d}\Delta/\delta)^7}+\frac{4 }{e^2 (\sqrt{d}\Delta/\delta)^2}\right)
%\\
%&\leq  144 \left(\frac{k}{k-2}\right)^2 d\Delta^{ 2}\cdot \left( \frac{2}{e^2(\sqrt{d}\Delta/\delta)^2}+\frac{4 }{e^2 (\sqrt{d}\Delta/\delta)^2}\right)\\
%&\leq  \frac{864}{e^2}\left(\frac{k}{k-2}\right)^2 \delta^2 \\
%&\leq \frac{153}\delta^2
%}%
\noindent where the first inequality uses the fact that $\delta \in (0, \sqrt{d}\Delta)$,  the second inequality uses $\ln x \leq 2 \sqrt{x}/e$, and {the final inequality uses the assumption that $n \geq 8.$} For the last three terms in \eqref{e:c}, we have
\begin{align*}
&4d\left(\frac{\sigma^{\prime 2}}{c_3}+\Delta^{\prime 2}\right)e^{- \frac{c_3 \Delta^{\prime 2}}{\sigma^{\prime 2}}}+4d\left(\frac{\sigma^{ 2}}{dc_3}+\Delta^{\prime 2}\right)e^{- \frac{c_3 \Delta^{\prime 2} d}{\sigma^{2}}}+4d\left(\frac{2\sigma^{\prime 2}}{dc_3^2}+\frac{2 \sigma^{\prime}\Delta^{\prime}}{c_3\sqrt{d}} +\Delta^{\prime 2}\right)e^{-\frac{c_3\Delta^{\prime}\sqrt{d}}{\sigma^{\prime}}}\\
&\leq 8d\left(\Delta^2+ \Delta^{\prime 2}\right)e^{-\frac{\Delta^{\prime 2}}{\Delta^2}} +4d\left(2\Delta^{ 2}+2\Delta\Delta^{\prime}+\Delta^{\prime 2}\right)e^{-\frac{\Delta^{\prime }}{\Delta}}\\
&\leq 8d\left(\Delta^2+ \Delta^{\prime 2}\right)e^{-\frac{\Delta^{\prime 2}}{\Delta^2}} +4d\left(3\Delta^{ 2}+2\Delta^{\prime 2}\right)e^{-\frac{\Delta^{\prime }}{\Delta}}\\
&= \frac{8d\left(\Delta^2+ 9\Delta^{ 2} (\ln (\sqrt{d}\Delta/\delta))^2\right)}{ (\sqrt{d}\Delta/\delta)^9}+ \frac{ 4d\left(3\Delta^{ 2}+18\Delta^{ 2}(\ln (\sqrt{d}\Delta/\delta))^2\right) }{(\sqrt{d}\Delta/\delta)^3}\\
&\leq \frac{4d\left(5\Delta^2+ 36\Delta^{ 2} (\ln (\sqrt{d}\Delta/\delta))^2\right)}{ (\sqrt{d}\Delta/\delta)^3} \\
&\leq \frac{4d\left(5\Delta^2+ 144\Delta^{ 2} \sqrt{d}\Delta/(e^2\delta)\right)}{ (\sqrt{d}\Delta/\delta)^3} \\
&= \frac{20\delta^2}{(\sqrt{d}\Delta/\delta)} +\frac{576\delta^2}{e^2}\\
&\leq 98\delta^2, \numberthis \label{e:c2}
\end{align*}
where the first inequality is due to choice of $\Delta,$ the second is AM-GM inequality, the third one uses the fact: $\delta \in (0, \sqrt{d}\Delta)$,  and the fourth one uses $\ln x \leq 2 \sqrt{x}/e.$
\newest{Substituting \eqref{e:c1} and \eqref{e:c2} in \eqref{e:c}, we get
\begin{align*}
\E{\norm{ Q_{{\tt M}, R_j, j}( \hat{g}_j(x), Y)- \hat{g}_j(x)}_2^2}
\leq \frac{36d\Delta^{2}}{(\ell_2-2)^2}\left(\ln \frac{\sqrt{d}\Delta}{\delta}\right)^2+237\delta^2.
\end{align*}
Finally, we note that whenever $nr \geq c_4 d^2 \log (B/\sigma)$, $\sigma^{\prime 2}\leq \sigma^2/d$ (see \eqref{e:sigma_prime}). 
With our earlier choice of $\Delta^2=\max\{ \frac{\sigma^{\prime 2}}{c_3},  \frac{\sigma^{\prime 2}}{dc_3^2}, \frac{\sigma^{2}}{dc_3}\},$ this further implies $\Delta\leq \frac{\sigma}{\sqrt{dc_3}}$.  Using this fact in the bound above establishes the MSE bound.}
%%\begin{align*}
%%&4d\left(\frac{4\sigma^{\prime 2}}{c}+\Delta^{\prime 2}\right)\exp{(-c \Delta^{\prime 2}/4\sigma^{\prime 2})}+ 4d\left(\frac{4\sigma^{ 2}}{dc}+ \Delta^{\prime 2}\right)\exp{(-c \Delta^{\prime 2}d/4\sigma^{ 2})}  \\
%%&\hspace{2cm}+4d\left(\frac{8\sigma^{\prime 2}}{dc^2}+\frac{4 \sigma^{\prime}\Delta^{\prime}}{c\sqrt{d}}+\Delta^{\prime 2}\right)\exp{(-c \Delta^{\prime }\sqrt{d}/2\sigma^{\prime})}
%%\\&\leq 8d\left(\Delta^2+ \Delta^{\prime 2}\right)\exp{(-\Delta^{\prime 2}/\Delta^2)} +4d\left(2\Delta^{ 2}+2\Delta\Delta^{\prime}+\Delta^{\prime 2}\right)\exp{(-c \Delta^{\prime }\sqrt{d}/2\sigma^{\prime})}
%%\\&\leq 
%%8d\left(\Delta^2+ \Delta^{\prime 2}\right)\exp{(-\Delta^{\prime 2}/\Delta^2)} +4d\left(3\Delta^{ 2}+2\Delta^{\prime 2}\right)\exp{(-\Delta^{\prime }/\Delta)}
%%\\&\leq \frac{8d\left(\Delta^2+ 9\Delta^{ 2} (\ln (\sqrt{d}\Delta/\delta))^2\right)}{ (\sqrt{d}\Delta/\delta)^9}+ \frac{ 4d\left(3\Delta^{ 2}+18\Delta^{ 2}(\ln (\sqrt{d}\Delta/\delta))^2\right) }{(\sqrt{d}\Delta/\delta)^3}\\
%%&\leq \frac{4d\left(5\Delta^2+ 36\Delta^{ 2} (\ln (\sqrt{d}\Delta/\delta))^2\right)}{ (\sqrt{d}\Delta/\delta)^3} \\
%%&\leq \frac{4d\left(5\Delta^2+ 144\Delta^{ 2} \sqrt{d}\Delta/(e^2\delta)\right)}{ (\sqrt{d}\Delta/\delta)^3} \\
%%&\leq \frac{20\delta^2}{(\sqrt{d}\Delta/\delta)} +\frac{576\delta^2}{e^2}\\
%%&\leq 98\delta^2
%%\end{align*}
%\begin{align}
%\label{e:break}
%    \E{\norm{Q_{{\tt M},R}(x,y)-x}_2^2}
%&\leq \,\frac{36d\Delta^2}{(k-2)^2}\left(\ln \frac{\sqrt{d}\Delta}{ \delta}\right)^2 + \delta^2,
%\end{align}
\paragraph*{Bound for Bias.}
Using the fact that for $\{|R_j(\hat{g}_j(x)-Y)(i)| \leq \Delta^{\prime}\}$ the modulo quantizer gives an unbiased estimate and the Jensen's inequality and, we have 
\begin{align*}
\norm{\E{ Q_{{\tt M}, R_j, j}( \hat{g}_j(x), Y)}-\hat{g}_j(x)}_2^2&= \sum_{i=1}^{d} \E{\left(\hat{x}_R(i)-Rx(i)\right)\indic{|R\left(x-y\right)_i|
\geq \Delta^{\prime})} }^2\\
&\leq  \sum_{i=1}^{d} \E{\left(\hat{x}_R(i)-Rx(i)\right)^2\indic{|R\left(x-y\right)_i| \geq \Delta^{\prime})} }\\
&\leq 237\delta^2.
\end{align*}
\end{proof}
\paragraph*{Completing the proof.} We now calculate the MSE and bias for our Wyner-Ziv quantizer with subsampled RMQ. 

\newest{
Note that the inequality \eqref{e:intermed_subsampling} derived in Lemma \ref{t:RCS_RAQ_alpha_beta}  holds in this case. 
Therefore, we have for $j\in \clientset_2$ that
\begin{align*}
\E{\norm{Q_{{\tt WZ}, j}( \hat{g}_j(x), Y)- \hat{g}_j(x)}_2^2}
&\leq \frac{2d}{r_2} \E{\norm{ Q_{{\tt M}, R_j, j}( \hat{g}_j(x), Y)- \hat{g}_j(x)}_2^2} \\
&\qquad + 
\frac{2d}{r_2}\E{\|R\hat{g}_j(x) -RY\|_2^2}\\
&\leq \frac{2d}{r_2}\left(\frac{36\sigma^2}{c_3(\ell_2-2)^2}\left(\ln \frac{\sigma}{ \sqrt{c_3}\delta}\right)^2+237\delta^2\right)\\
&\qquad + 
\frac{2d}{r_2}\E{\|R\hat{g}_j(x) -R\nabla f(x)\|_2^2} + 
\frac{2d}{r_2}\E{\|RY -R\nabla f(x)\|_2^2}\\
&\leq \frac{2d}{r_2}\left(\frac{36\sigma^2}{c_3(\ell_2-2)^2}\left(\ln \frac{\sigma}{ \sqrt{c_3}\delta}\right)^2+237\delta^2+\sigma^2\right)\\
&+ \frac{2d}{r_2}\E{\|Y-\nabla f(x)\|^2},
\end{align*}
where the second last inequality uses bound from Lemma \ref{l:conc_SI_grad} and the fact \[\E{\|R\hat{g}_j(x) -RY\|^2}=\E{\|R\hat{g}_j(x) -R\nabla f(x)\|^2}+\E{\|RY -R\nabla f(x)\|^2},\] and the final inequality uses the fact that $R$ is a unitary matrix  and assumption \eqref{e:bounded_est}.

Recall from Lemma \ref{l:SGnormSI} that the quantity $(\nabla f(x)(i) -Y(i))$ is subgaussian with variance parameter $\sigma^{\prime 2}.$ Thus, 
$\E{(Y(i)-\nabla f(x)(i))^2}\leq c_5\sigma^{\prime 2}$ for some universal constant $c_5>0$ (see, for instance, \cite{vershynin2018high}).  Again using the fact that $d\sigma^{\prime 2}<\sigma^2$, whenever $nr \geq c_4 d^2 \log (B/\sigma)$, $\E{\|Y-\nabla f(x)\|^2}\leq \sigma^2$.

Further, the bias remains unchanged compared to without subsampling case, i.e,  \[\norm{\E{Q_{{\tt WZ}, j}( \hat{g}_j(x), Y)}-\hat{g}_j(x)} = \norm{\E{ Q_{{\tt M}, R_j, j}( \hat{g}_j(x), Y)}-\hat{g}_j(x)} \leq \sqrt{237}\delta.\]
%\snote{$c_5=4/constant $} 
}
%Notice in the standard setup of compression with side-information we expect MSE bounds of the form $\sigma^2/(\ell -1)^2$ with $\log \ell$ bits of precision, where $\sigma$ is the distance between the input and the side information. A large value of $Kr$ allows us to get the side-information vector $\sigma$ close to the side information in subexponential sense.
\newnewest{At last, we set the following parameters for $\tt WZ\text{-}SGD$:
\[\log \ell_2=\ceil{c \log \log nT}, \qquad \delta = \frac{2\sigma}{nT}.\]
Accordingly,  we need to sample $r_2=r/\ceil{c \log \log nT}$ coordinates at each client.  Using the standard bounds for averaging of vectors in Lemma \ref{l:main}, respectively, we obtain
\[{\alpha^{\prime}}^2(\M_t) \leq  c_1\frac{\sigma^2 \log \log nT}{n} \cdot \frac{d}{r}, \qquad {\beta^{\prime}}^2(\M_t) \leq \frac{c_2\sigma^2}{n^2T^2},\]
for suitably chosen constants $c_1, c_2>0$ and $\M_t$ as defined in \eqref{e:descent_WZ_SGD}. 
The proof is completed by using the bounds on $\alpha^\prime$ and $\beta^\prime$ with Lemma \ref{l:general}.}
%\snote{Why not a constant $\ell_2$? Am I missing anything?}

\qed

\newest{\subsection{Proof of Theorem \ref{c:PSGD+}}
\noindent We proceed in a way similar to the proof of Theorem \ref{t:ParSGD}.
Towards that,  we first use the mean square assumption in \eqref{e:var} to write
\[\E{\left\| \frac{2}{n}\sum_{i \in \clientset_2} \hat{g}_i(x_t)- \nabla f(x_t)\right\|^2} \leq \frac{2\sigma^2}{n}.
\]
Then, it only remains to bound the term $\E{\left\|\frac{2}{n} \sum_{i \in \clientset_2} Q_{\mathtt{RDAQ}}(\hat{g}_i(x_t),Y_t)-\frac{2}{n} \sum_{i \in \clientset_2} \hat{g}_i(x_t)\right\|^2},$ where ${Y_t=\frac{2}{n}\sum_{i\in \clientset_2}Q_{\mathtt{RATQ}}(\hat{g}_i(x_t))}$ is the side-information.
For that, we follow the proof of Lemma \ref{t:sRDAQ} in Section \ref{s:proof_RDAQ}. 
%Following the proof
%\begin{lem}[{\cite[Lemma 4.3]{mayekar2021wyner}}]\label{l:subsampled_RDAQ}
%Let $Q_{\mathtt{RDAQ}}(\cdot , \cdot)$\footnote{The arguments denote quantizer's input  and the available side-information, respectively.} be the subsampled version of RDAQ using ${d \geq r\geq h+\log h}$
% bits. Then for every input ${x\in [-B,B]^d}$ and side-information ${y\in [-B,B]^d}$,  we have 
%\begin{align*}
%\E{Q_{\mathtt{RDAQ}}(x,y)}&=x\quad \text{and} \\
%\E{\norm{Q_{\mathtt{RDAQ}}(x,y)-x}_2^2} &\leq 
%\frac{16\sqrt{3}dB\|x-y\|_2}{\frac{r}{\ceil{h+\log h}}-1},
%\end{align*}
%where $h{=}1+ \ln^\ast(d/6).$
%\end{lem}

Fix any arbitrary client ${i \in \clientset_2}$. Conditioning on its gradient estimate $\hat{g}_i(x_t),$ the available side information at server $Y_t=y_t$,  we have using the proof of Lemma \ref{t:sRDAQ} that
\begin{align*}
&\E{\left\| Q_{\mathtt{RDAQ}}(\hat{g}_i(x_t),Y_t)-\hat{g}_i(x_t)\right\|^2\mid \hat{g}_i(x_t), Y_t=y_t}\leq \frac{16\sqrt{3}dB\|\hat{g}_i(x_t)-y_t\|}{\frac{r}{\ceil{h+\log h}}-1}.
\end{align*}
By the law of total expectations, we also have
\begin{align*}
&\E{\left\| Q_{\mathtt{RDAQ}}(\hat{g}_i(x_t),Y_t)-\hat{g}_i(x_t)\right\|^2}\leq \frac{16\sqrt{3}dB\E{\|\hat{g}_i(x_t)-Y_t\|}}{\frac{r}{\ceil{h+\log h}}-1}.
\end{align*}
The proof is completed by noting that 
\begin{align*}
\E{\|\hat{g}_i(x_t)-Y_t\|}^2&\leq \E{\|\hat{g}_i(x_t)-Y_t\|^2}\\
&=\E{\|\hat{g}_i(x_t){-}\nabla f(x_t)\|^2}+\E{\|\nabla f(x_t){-}Y_t\|^2}\\
&\leq \sigma^2+\frac{2\sigma^2}{n}+ \frac{ 2dB^2}{n\left(\frac{r}{ 3+ \ceil{\log(1+ \ln^\ast(d/3))}} -1\right)},
\end{align*}
where the first line is using Jensen's inequality,  the only identity is due to the unbiased property of subsampled RATQ ($c.f.$ Lemma \ref{l:subsampled_RATQ}), and the last line is due to \eqref{e:var} and applying the value of $\alpha(\M_t)$ for $\clientset_1$ in Theorem \ref{t:ParSGD}.
\qed}
\subsection{Proof of Theorem \ref{t:G_WZ}}
 The proof of this Theorem is similar to that of Lemma \ref{t:RMQ}.
 We denote by $Q(X(i), Y(i))$ the output of the modulo quantizer with side information $Y(i)$ and parameters $k$, $\Delta^{\prime}$ set as in \eqref{e:WZ_MQ}. Then, we have
 
\begin{align}
\nonumber
\E{\norm{Q_d(X, Y)-X}^2} &\leq \sum_{i=1}^{d} \E{(Q(X(i), Y(i)) -X(i))^2}\\
\nonumber
&\leq \sum_{i=1}^{d} \E{(Q(X(i), Y(i))-X(i))^2 \indic{\{|\left(X(i)-Y(i)\right)| \leq \Delta^{\prime}\}}}
\\& \hspace{2cm}+ \sum_{i=1}^{d} \E{(Q(X(i), Y(i))-X(i))^2 \indic{\{|\left(X(i)-Y(i)\right)| \geq \Delta^{\prime}\}}}.
\label{e:error_split_GWZ}
\end{align}
 
We bound the first term on the right-side in a similar manner as the bound in \eqref{e:error_term1}. Specifically,
under the event 
$\{|X(i)-Y(i)| \leq
   \Delta^{\prime}\}$, we get
   by Lemma~\ref{t:MQ}
   that
   \[
   |Y(i)
   -X(i)| \leq \eps=
   \frac{2\Delta^{\prime}}{k-2}, \quad \text{{almost surely}},
   \]
   whereby 
  \begin{align}
  \sum_{i=1}^{d} \E{(Y(i)-X(i))^2 \indic{\{|X(i)-Y(i)\}| \leq \Delta^{\prime}}}\leq d\,\eps^2.
  \label{e:error_term1_GWZ}
  \end{align}

For the second term in the RHS note that $X(i)-Y(i)$ is subgaussian with  variance factor $\sigma_z^2$. Therefore, by proceeding in a similar manner as the derivation of \eqref{e:error_term2} we get
\begin{align}\nonumber &\sum_{i=1}^{d} \E{(Q(X(i), Y(i))-X(i))^2 \indic{\{|X(i)-Y(i)| \geq \Delta^{\prime}\}}}\\
&\leq 
2\sum_{i=1}^{d} \left[\E{( Q(X(i), Y(i))- Y(i))^2 \indic{\{|X(i)-Y(i)| \geq \Delta^{\prime}\}}}
+\E{(Y(i)-X(i))^2 \indic{\{|X(i)-Y(i)| \geq \Delta^{\prime}\}}}
\right]
\nonumber 
\\
&\leq 2k^2\eps^2
\sum_{i=1}^{d}P(|X(i)-Y(i)| \geq \Delta^{\prime})
+
2\sum_{i=1}^{d} 
\E{(X(i)-Y(i))^2 \indic{\{|X(i)-Y(i)| \geq \Delta^{\prime}\}}}
\nonumber
\\
&\leq 4dk^2\eps^2e^{-{d{\Delta^\prime}^2}/{2 \sigma_z^2}}
+
2\sum_{i=1}^{d} 
\E{(X(i)- Y(i))^2 \indic{\{|X(i)-Y(i)| \geq \Delta^{\prime}\}}}
\nonumber
\\
&\leq 4dk^2\eps^2 e^{-{{{\Delta^\prime}}^2}/{2\sigma_z^2}}+
4(2\sigma_z^2+d\Delta^{\prime 2})e^{-\frac{{\Delta^{\prime}}^2}{2\sigma_z^2}},
\label{e:error_term2_GWZ}
\end{align}
where the second inequality follows upon noting from the description decoder of MQ in
Alg.~\ref{a:D_MCQ} that 
$|Q(X(i), Y(i))- Y(i)|\leq \eps k$
almost surely for
each $i\in [d]$; the third inequality uses
the fact that $X(i)-Y(i)$ is sub-Gaussian with variance parameter
$\sigma_z^2$; and the fourth inequality
is by Lemma~\ref{l:variance_tail}.

Upon bounding the two terms on the right-side of \eqref{e:error_split_GWZ} 
from above
using \eqref{e:error_term1_GWZ}, \eqref{e:error_term2_GWZ},  we obtain
\begin{align}
\nonumber
\E{\norm{Q_d(X, Y)-X}^2}\leq d\eps^2+ 4dk^2\eps^2 e^{-{{\Delta^\prime}^2}/{2\sigma_z^2}}+
4(2\sigma_z^2+d\Delta^{\prime 2})e^{-\frac{{{\Delta^{\prime}}^2}}{2\sigma_z^2}}.
\end{align}
Note that the RHS in the upper bound above is precisely the same as in \eqref{e:break} with $\sigma_z^2$ replacing $\Delta^2/d$.Therefore proceeding in the same manner as in \eqref{e:break}, we get
\[\E{\norm{Q_d(X, Y)-X}^2}\leq 24 \frac{\sigma_z^2}{(k-2)^2}\ln \frac{\sigma_z}{\delta} + 154 \delta^2.\]
Substituting the value of $k$ and $\delta$ completes the proof.
 \qed
 
 \subsection{Proof of Lemma \ref{t:bRDAQ}}
For $Q(x)$ as in \eqref{e:boosted_RDAQ}, we have \[Q(x) =\sum_{i=1}^{N} q_i/N,\] where $q_i$ for all $i \in \{1, \ldots N\}$ is an unbiased estimate of $x$ and equals in distribution the output of the RDAQ quantizer for an input $x$ and side information $y$.  Moreover, $q_i$s are mutually independent conditioned on $R$. Therefore,
\eq{
\E{\norm{Q(x) -x}_2^2 } &= \E{\norm{\sum_{i=1}^{N}\frac{q_i}{N}-x}_2^2}\\
&=\E{ \E{\norm{\sum_{i=1}^{N}\frac{q_i}{N}-x}_2^2|R}}\\
&=\E{\sum_{i=1}^{N}\frac{1}{N^2}\E{\norm{q_i-x}_2^2|R}}\\
&\leq 16 \sqrt{3}\, \frac{\Delta}{N},
}
where the third identity follows from the conditional independence of
$q_i$s after conditioning on $R$ and the fact that $q_i$ is an
unbiased estimate of $x$.
The final inequality
 follows from the fact that $q_i$ equals in distribution the output of the RDAQ quantizer and then using Lemma \ref{t:RDAQ}.
\qed

\bibliography{tit2018}
\bibliographystyle{IEEEtranS}
\end{document}

%% file: preamble.tex
\usepackage[utf8]{inputenc} % allow utf-8 input
\usepackage{hyperref} %hyperlinks
\usepackage{url} % simple URL typesetting
\usepackage{amsfonts} % blackboard math symbols
\usepackage[noend]{algcompatible}
\usepackage[noend]{algpseudocode}
\usepackage{soul}
\usepackage{mleftright}
\usepackage{varwidth}
\usepackage[utf8]{inputenc} % allow utf-8 input
\usepackage[T1]{fontenc} % use 8-bit T1 fonts
\usepackage{hyperref}%hyperlinks
\usepackage{url} % simple URL typesetting
\usepackage{comment}

\usepackage[usenames,dvipsnames,table]{xcolor}
\usepackage{amsmath}
\usepackage{mathrsfs}
\usepackage{algorithm}
\usepackage[noend]{algpseudocode}
\algsetblock[Name]{Encoder}{Stop}{3}{1cm}
\algsetblock[Name]{Decoder}{Stop}{3}{1cm}
\algsetblock[Name]{SharedRandomness}{Stop}{3}{1cm}

\usepackage[noadjust]{cite}
\usepackage{%
    amssymb,%
    amsthm,%%babel,%
    bbm,%
    etoolbox,%
    mathtools,%
    stmaryrd%
}

\usepackage{tikz,pgf,pgfplots,circuitikz}
\theoremstyle{plain} \newtheorem{thm}{Theorem}[section]
\newtheorem{lem}[thm]{Lemma}

\theoremstyle{definition} \newtheorem{defn}[thm]{Definition}

\theoremstyle{remark} \newtheorem{rem}{Remark}

\definecolor{lightgray}{gray}{0.9}

\newcommand{\ep}{\mathcal{E}}
\newcommand{\eps}{\varepsilon}
\newcommand{\eqdef}{:=}
\newcommand{\eq}[1]{\begin{align*}#1\end{align*}}
\newcommand{\ceil}[1]{\left\lceil#1\right\rceil}%
\newcommand{\norm}[1]{\|#1\|}%
\newcommand\numberthis{\addtocounter{equation}{1}\tag{\theequation}}
 \newcommand{\R}{\mathbb{R}}

\newcommand{\E}[1]{\mathbb{E}\left[#1\right]}

\newcommand{\Z}{\mathbb{Z}} 
\newcommand{\B}{\mathscr{B}}

\newcommand{\X}{\mathcal{X}}

\newcommand{\Y}{\mathcal{Y}}

\newcommand{\oO}{\mathcal{O}}
\newcommand{\Q}{\mathcal{Q}}

\newcommand{\uRoman}[1]{\uppercase\expandafter{\romannumeral#1}}

\DeclarePairedDelimiter\floor{\lfloor}{\rfloor}

\usepackage{subcaption}% http://ctan.org/pkg/subcaption
\DeclareCaptionSubType*{algorithm}

\DeclareCaptionLabelFormat{alglabel}{Alg.~#2}
\ifnum\withcolors=1
  \newcommand{\newer}[1]{{\color{brown} {#1}}} % even newer
  \newcommand{\newest}[1]{{\color{blue} {#1}}} % the absolute latest stuff
  \newcommand{\newnewest}[1]{{\color{purple} {#1}}} 
  \newcommand{\mcolor}[1]{{\color{ForestGreen}#1}} % Prathamesh
  \newcommand{\tcolor}[1]{{\color{Orange}#1}} % Himanshu
  \newcommand{\jcolor}[1]{{\color{Cyan}#1}}
\else
  
  \newcommand{\newer}[1]{{{#1}}}
  \newcommand{\newest}[1]{{{#1}}}
   \newcommand{\newnewest}[1]{{{#1}}}
  \newcommand{\mcolor}[1]{{#1}}
  
  \newcommand{\tcolor}[1]{{#1}}
  \newcommand{\jcolor}[1]{{#1}}
\fi

\ifnum\withnotes=1
  \newcommand{\mnote}[1]{\par\mcolor{\textbf{M: }\sf #1}} % Prathamesh
  \newcommand{\tnote}[1]{\par\tcolor{\textbf{T: }\sf #1}} % Himanshu
 \newcommand{\snote}[1]{\par\jcolor{\textbf{J: }\sf #1}}  %Shubham
  \newcommand{\anote}[1]{\par\tcolor{\textbf{A: }\sf #1}} % Theertha
  
  \newcommand{\htnote}[1]{{\color{orange}\tt [HT: #1]}}
  \newcommand{\anote}[1]{\par\tcolor{\textbf{A: }\sf #1}} 
\else
  \newcommand{\mnote}[1]{}
  \newcommand{\tnote}[1]{}
  \newcommand{\snote}[1]{} 
  \newcommand{\htnote}[1]{}
  \newcommand{\anote}[1]{} 
  
\fi
\newcommand{\ignore}[1]{\leavevmode\unskip} % eat unnecessary spaces before
\newcommand{\Qenc}{Q^{{\tt e}}}
\newcommand{\Qdec}{Q^{{\tt d}}}
\newcommand{\known}{known $\mathbf{\Delta}$ }
\newcommand{\unknown}{unknown $\mathbf{\Delta}$ }

\newcommand{\indic}[1]{\mathbbm{1}_{#1}}

\usepackage{soul}

\newcommand{\cE}{\mathcal{E}}
\newcommand{\cO}{\mathcal{O}}

\newcommand{\mutualinfo}[2]{ I\mleft(#1 \land #2\mright) }

\newcommand{\clientset}{\mathtt{C}} %% changing \mathcal to \mathtt as earlier was for communication

\newcommand{\A}{\mathtt{A}}
\newcommand{\cA}{\mathcal{A}}
\newcommand{\tC}{\mathtt{Cl}}
\newcommand{\MSE}{\mathtt{MSE}}
\newcommand{\M}{\mathcal{M}}
\newcommand{\op}{C}